\documentclass[12pt]{amsart}

\usepackage[utf8]{inputenc}
\usepackage[T1]{fontenc}
\usepackage{url,  mathrsfs}
\usepackage{amsmath}
\usepackage{amssymb}
\usepackage{amsthm}
\usepackage{graphicx, epstopdf}
\usepackage{subfig}
\usepackage{color}
\usepackage{bbm}
\usepackage{dsfont}
\usepackage{subfloat}
\usepackage[draft]{fixme}
\usepackage{bm}
\usepackage{bbm} 
\usepackage{ifpdf}
\usepackage{array}
\usepackage{mathtools}
\usepackage{multicol}
\usepackage{multirow}
\usepackage{graphicx}
\usepackage{tikz}
\usepackage{hyperref}
\usepackage[capitalize]{cleveref}

\renewcommand{\S}{\mathcal{S}}
\newcommand{\A}{\mathrm{LY}}

\newcommand{\D}{\mathbb{D}}
\newcommand{\Q}{\mathbb{Q}}
\newcommand{\N}{\mathbb{N}}
\newcommand{\Z}{\mathbb{Z}}
\newcommand{\R}{\mathbb{R}}

\newcommand{\C}{\mathbb{C}}
\renewcommand{\P}{\mathbb{P}}

\newcommand{\weak}{\stackrel{\mathcal{D}}{\rightarrow}}
\newcommand{\inv}{^{\dagger}}

\newcommand{\T}{\mathbb{T}}

\renewcommand{\H}{\mathcal{H}}

\newcommand{\x}{\mathbf{x}}
\newcommand{\balpha}{\mathbf{\alpha}}
\newcommand{\bphi}{\mathbf{\phi}}
\newcommand{\bd}{\mathbf{d}}
\newcommand{\z}{\mathbf{z}}
\newcommand{\mm}{\mathbf{m}}

\newcommand{\lv}{\ell}

\newcommand{\xv}{\mathbf{x}}
\newcommand{\yv}{\mathbf{y}}

\DeclarePairedDelimiterX{\floor}[1]{\lfloor}{\rfloor}{#1}
\DeclarePairedDelimiterX{\ceil}[1]{\lceil}{\rceil}{#1}
\DeclarePairedDelimiterX{\card}[1]{\ellert}{\rvert}{#1}
\DeclarePairedDelimiterX{\abs}[1]{\ellert}{\rvert}{#1}
\DeclarePairedDelimiterX{\norm}[1]{\ellert}{\rVert}{#1}
\DeclarePairedDelimiterX{\tuple}[1]{\lparen}{\rparen}{#1}
\DeclarePairedDelimiterX{\parens}[1]{\lparen}{\rparen}{#1}
\DeclarePairedDelimiterX{\brackets}[1]{\lbrack}{\rbrack}{#1}
\DeclarePairedDelimiterX{\set}[1]\{\}{#1}
\let\Pr\relax
\DeclarePairedDelimiterXPP{\Pr}[1]{\mathbb{P}}[]{}{#1}
\DeclarePairedDelimiterXPP{\PrX}[2]{\mathbb{P}_{#1}}[]{}{#2}
\DeclarePairedDelimiterXPP{\Ex}[1]{\mathbb{E}}[]{}{#1}
\DeclarePairedDelimiterXPP{\ExX}[2]{\mathbb{E}_{#1}}[]{}{#2}

\DeclarePairedDelimiterX{\xt}[1]{\lbrack}{\rbrack}{#1}

\theoremstyle{plain}
\newtheorem{thm}{Theorem}[section]
\newtheorem{lem}[thm]{Lemma}
\newtheorem{prop}[thm]{Proposition}
\newtheorem{cor}[thm]{Corollary}

\newtheorem*{thm*}{Theorem} 

\theoremstyle{definition}
\newtheorem{Def}[thm]{Definition}
\newtheorem{example}[thm]{Example}
\newtheorem{remark}[thm]{Remark}

\DeclareMathOperator{\diag}{diag}

\DeclareMathOperator{\im}{Im}
\DeclareMathOperator{\supp}{supp}

\DeclareMathOperator{\reg}{reg}
\DeclareMathOperator{\sing}{sing}
\DeclareMathOperator{\mult}{m}

\DeclareMathOperator{\vol}{vol}

\usepackage[inner=1in,outer=1in,top=1in,bottom=1in]{geometry}

\makeatletter
\newcommand{\addresseshere}{%
	\enddoc@text\let\enddoc@text\relax
}
\makeatother

\begin{document}
		\title[Gap distributions of Fourier quasicrystals]{Gap distributions of Fourier quasicrystals\\ via Lee-Yang polynomials}
		\author{Lior Alon and Cynthia Vinzant}
	 
\begin{abstract}
Recent work of Kurasov and Sarnak provides a method for constructing one-dimensional Fourier quasicrystals (FQ) from the torus zero sets of a special class of multivariate polynomials called Lee-Yang polynomials. In particular, they provided a non-periodic FQ with unit coefficients and uniformly discrete support, answering an open question posed by Meyer.
Their method was later shown to generate all one-dimensional Fourier quasicrystals with $\mathbb{N}$-valued coefficients ($ \mathbb{N} $-FQ). 

In this paper, we characterize which Lee-Yang polynomials  
give rise to non-periodic $ \mathbb{N} $-FQs with unit coefficients and uniformly discrete support, and show that this property is generic among Lee-Yang polynomials. We also show that the infinite sequence of gaps between consecutive atoms of any $\mathbb{N}$-FQ 
has a well-defined distribution, which, under mild conditions, is absolutely continuous. This generalizes previously known results for the spectra of quantum graphs to arbitrary $\mathbb{N}$-FQs. Two extreme examples are presented: first, a sequence of $\mathbb{N}$-FQs whose gap distributions converge to a Poisson distribution. Second, a sequence of random Lee-Yang polynomials that results in random  $\mathbb{N}$-FQs whose empirical gap distributions converge to that of a random unitary matrix (CUE).
\end{abstract}

\maketitle

\section{Introduction}
An $ \N $-FQ is a counting measure of a discrete set (with multiplicities), whose Fourier transform is also supported on a discrete set and has moderate growth. A recent sequence of works \cite{KurasovSarnak20jmp,OlevUlan20,AloCohVin} established that all one-dimensional $ \N $-FQs arise from the torus zero sets of a special class of multivariate polynomials called Lee-Yang polynomials.

\begin{figure}[h!]
	\parbox[m]{5in}{
		\begin{center}
			\includegraphics[width=0.6\textwidth]{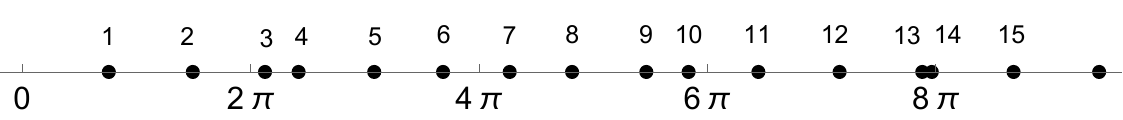}
	\end{center}\strut } \vspace{-.3in}
	\parbox[m]{5in}{  
		\begin{center}
			\includegraphics[width=0.3\textwidth]{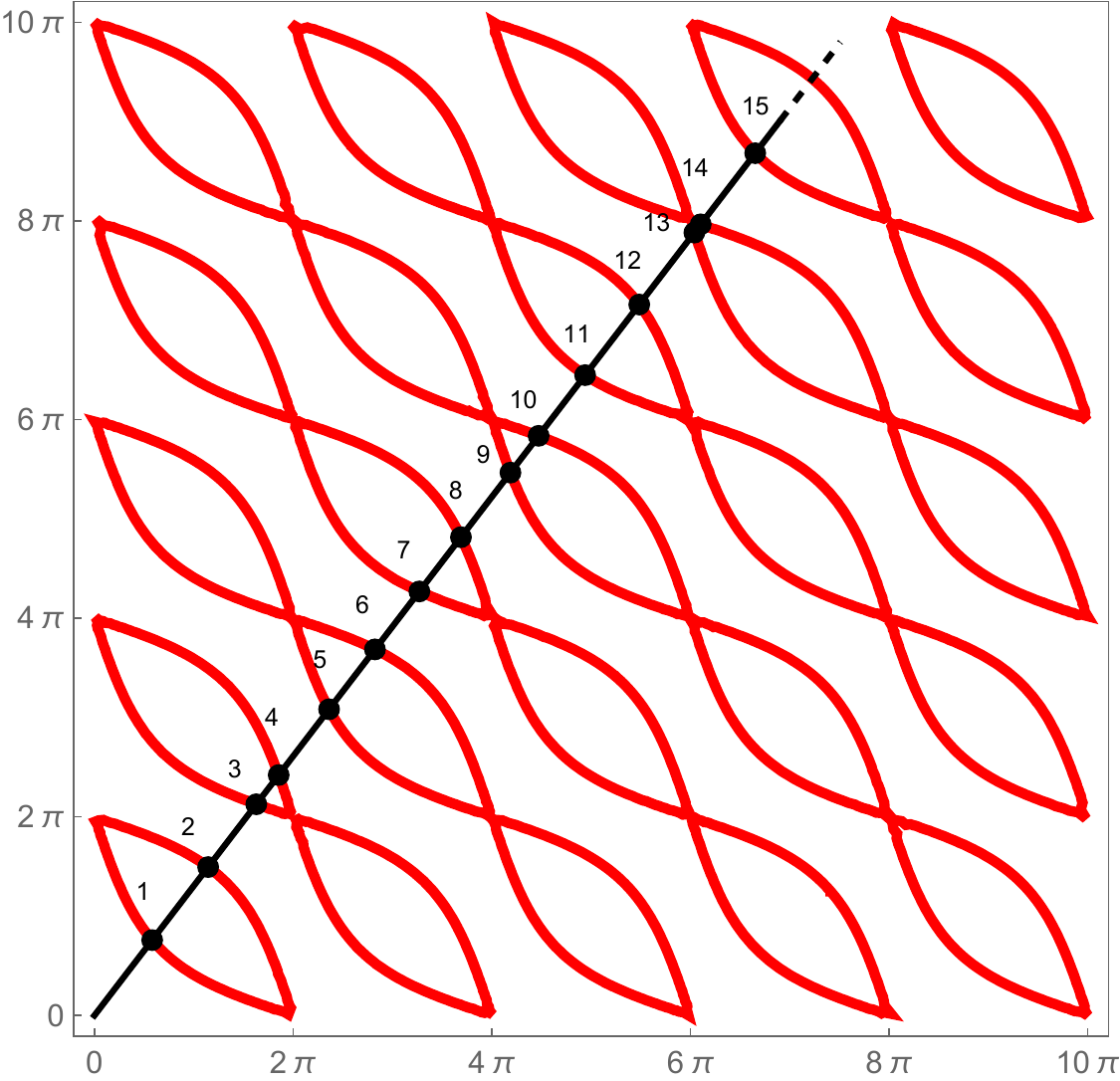} \hspace{.4in}
			\includegraphics[width=0.288\textwidth]{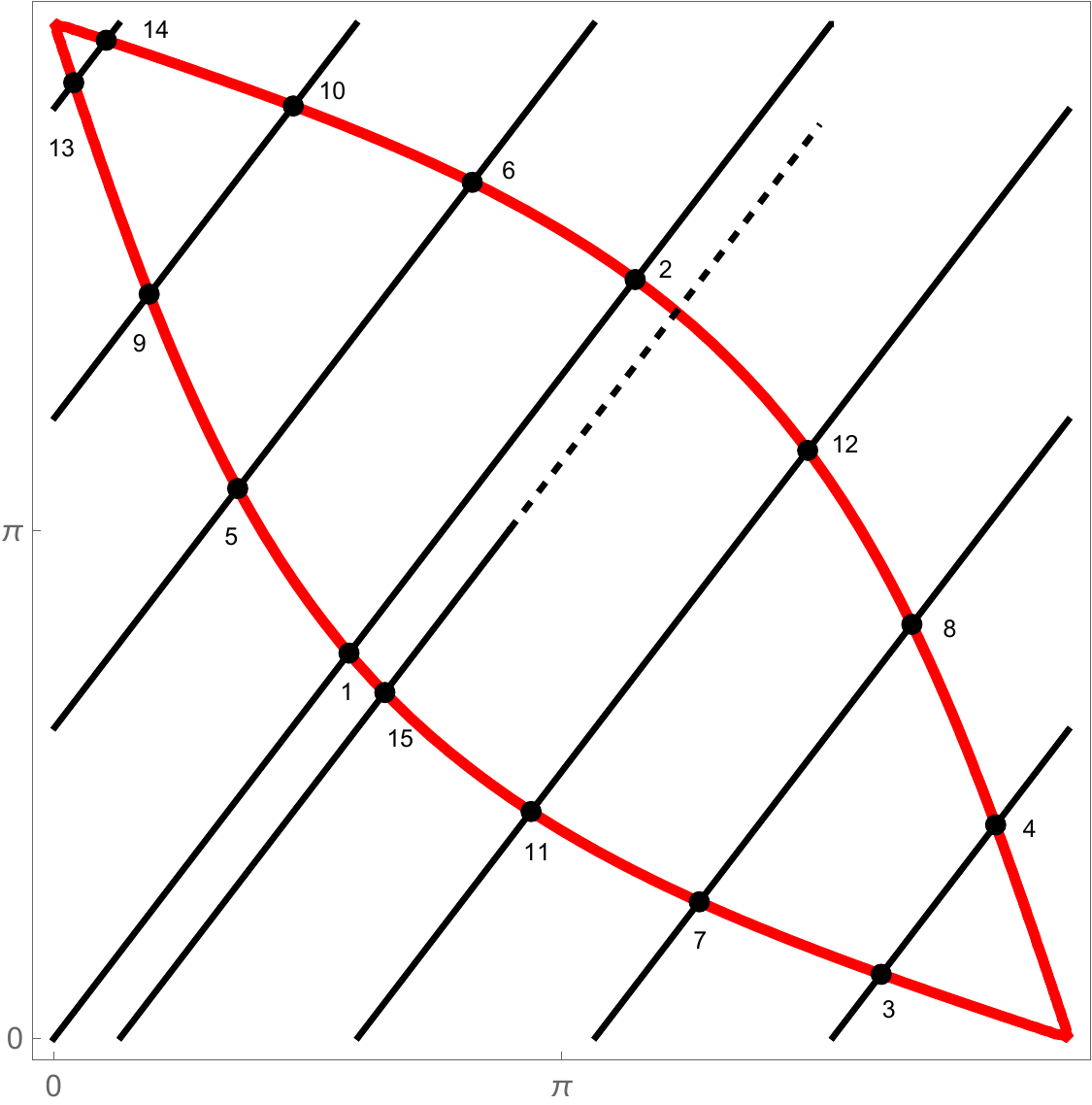} 
		\end{center}
	\strut}\vspace{.1in}
	\caption{The Kurasov-Sarnak construction of an $\N$-FQ from the zero-set of a Lee-Yang polynomial
		in the torus $\T^2$. See \Cref{ex:running}.
	}
	\label{fig: flow on secman}
\end{figure}


Given a discrete periodic set $\Lambda\subset\mathbb{R}$ with period $\Delta>0$, the Poisson summation formula states that 
  \[\sum_{x\in\Lambda}f(x)=\frac{2\pi}{\Delta}\sum_{k\in\Lambda^{*}}\hat{f}(k),\quad\text{for all}\ f\in\mathcal{S}(\mathbb{R}),\]
  where $ \Lambda^{*}=\{k\in\R\ : \ \forall x\in\Lambda, e^{ikx}=1 \}$ and 
$\mathcal{S}(\mathbb{R})$ is the space of Schwartz functions: smooth functions on $ \R $ that rapidly decay to zero at $ \pm\infty $ (properly defined in \Cref{sec:backgroundCrystals}). Fourier quasicrystals are generalizations of the Poisson summation formula to sets which are not periodic but exhibit similar features.   
\subsection{Fourier quasicrystals } Elements in the dual space $\mathcal{S}'(\mathbb{R})$ are called \emph{tempered distributions}, and the Fourier transform of $\mu\in\mathcal{S}'(\mathbb{R})$ is the tempered distribution $\hat{\mu}$ defined by duality, $\int f d\hat{\mu}:=\int\hat{f}d\mu$. For example, if $ \Lambda $ is periodic as above, then $\mu=\sum_{x\in\Lambda}\delta_x$ is tempered and its Fourier transform is $\hat{\mu}=\frac{2\pi}{\Delta}\sum_{k\in\Lambda^{}}\delta_k$ by the Poisson summation formula. However, in general, if a tempered distribution $\mu$ is supported on a non-periodic discrete set $\Lambda$, i.e., $\mu=\sum_{x\in\Lambda}a_x\delta_x$ for some complex coefficients $(a_x)_{x\in\Lambda}$, it is unlikely that $ \hat{\mu} $ will also be supported on a discrete set. 

If $\mu\in\mathcal{S}'(\mathbb{R})$ satisfies the condition that both $\mu$ and $\hat{\mu}$ are supported on discrete (locally finite) sets, then $\mu$ is called a \emph{crystalline measure} \cite{Meyer2016measures}. A crystalline measure $\mu=\sum_{x\in\Lambda}a_x\delta_x$ with Fourier transform $\hat{\mu}=\sum_{k\in S}c_k\delta_k$ (so that $ S\subset\R $ is some discrete set) is called a \emph{Fourier Quasicrystal} if  $|\mu|=\sum_{x\in\Lambda}|a_x|\delta_x$ and $|\hat{\mu}|=\sum_{k\in S}|c_k|\delta_k$
 are tempered as well \cite{LevOlev2016quasicrystals}. We say that $ \mu $ is $\mathbb{N}$-valued if $a_x\in\mathbb{N}$ for all $x\in\Lambda$, and we abbreviate $\mathbb{N}$-valued Fourier quasicrystals as $\mathbb{N}$-FQs.

\subsection{Lee-Yang polynomials} Following \cite{Ruelle2010LeeYang}, we call a polynomial $p  \in \C[z_1, \hdots, z_n]$ a 
\emph{Lee-Yang polynomial} if it has no zeros in the product $\D^n$ of the open unit disk, $\D = \{z\in \C : |z|<1\}$, and it has no zeros in the product of the outer disk $(\C\setminus\overline{\D})^n$. One fundamental example is a determinant
\[p(z_{1},z_{2},\ldots,z_{n})=\det({\rm diag}(z_{1},\hdots, z_{n}) + U),\] 
where $U$ is an $n\times n$ unitary matrix. The name Lee-Yang polynomials refers to the elegant proof of the Lee-Yang Circle Theorem \cite{BB1, BB2} by Br\"and\'en and Borcea. Lee-Yang polynomials are intimately related, by Mobius tranformations to the class of real stable polynomials, i.e., $ p\in\C[z_{1},z_{2},\ldots,z_{n}] $ with the property that  $p({\bf a})$ is nonzero whenever ${\bf a} = (a_1, \hdots, a_n)\in \C^n$ 
has imaginary part $ \im (a_j)>0 $ for all $j=1, \hdots, n$ or $\im (a_j)<0$ for all $j=1, \hdots, n$. Br\"and\'en and Borcea developed a classification of linear operations preserving stability and used this to prove the Lee-Yang Circle Theorem \cite{BB1, BB2}, among many other things. See \cite{Wagner} for a survey of these techniques. Many properties of determinants, especially those involving eigenvalues, also hold and have elegant proofs 
for general real stable polynomials. See, for example, \cite{BGLS}. 

\subsection{$ \N $-FQs and Lee-Yang polynomials}

Motivated by physical quasicrystals, Meyer posed an intriguing question:
Are there any non-periodic crystalline measures $ \mu=\sum_{x\in\Lambda}\delta_{x} $, with unit coefficients ($ a_{x}\equiv1 $) and uniformly discrete\footnote{A set $ \Lambda\subset\R $ is said to be \emph{uniformly discrete} if $ \exists r>0 $ such that $ |x-x'|\ge r>0 $ for any distinct $ x,x'\in\Lambda $.} support $ \Lambda $? 

In their notable work \cite{KurasovSarnak20jmp}, Kurasov and Sarnak presented a general construction of $\mathbb{N}$-FQs. Using this construction, they answered Meyer's question by providing an explicit example of a non-periodic FQ $ \mu $ with unit coefficients and a uniformly discrete support. A question addressed in this paper is whether these properties are common among all $\mathbb{N}$-FQs.  

To describe the Kurasov-Sarnak construction, suppose that $ p(z_{1},z_{2},\ldots,z_{n})=\sum_{\boldsymbol{\alpha}} c_{\boldsymbol{\alpha}}\z^{\boldsymbol{\alpha}}$ is a Lee-Yang polynomial, where we use the multi-index notation $ \z^{\boldsymbol{\alpha}}=\prod_{j=1}^{n}z_{j}^{\alpha_{j}} $, and let $ \lv=(\lv_{1},\ldots,\lv_{n})\in\R_{+}^{n} $. Then, the univariate exponential polynomial
\[
f(x) = p(\exp(ix\lv)) = p(e^{ix\ell_1}, \hdots, e^{ix\ell_n})=\sum_{\boldsymbol{\alpha}} c_{\boldsymbol{\alpha}}e^{ix\langle\boldsymbol{\alpha},\lv\rangle},
\]
is \emph{real-rooted}, namely $ f(x)=0\Rightarrow\im(x)=0 $, since $ \exp(ix\lv)\in\D^{n}\cup(\C\setminus\overline{\D})^{n} $ when $ \im(x)\ne 0 $. If $ f(x)=0 $ let $ m(x) $ denote the multiplicity\footnote{The multiplicity of a zero $ x $ of an analytic function $ f $ is the minimal $ n\in\N $ for which the $ n$-th derivative is non-zero $ f^{(n)}(x)\ne0 $.} of $ x $ as a zero of $ f $.

\begin{thm}[Kurasov-Sarnak construction\cite{KurasovSarnak20jmp}\label{thm: KS}]
Given a positive vector $ \lv\in\R_{+}^{n} $ and a Lee-Yang polynomial $ p(z_{1},z_{2},\ldots,z_{n}) $, let $ \Lambda $ denote the zero set of $ f(x)=p(\exp(ix\lv)) $ and $ m(x) $ the multiplicity of $ x\in\Lambda $. Then, the measure 
\[\mu_{p,\lv}:=\sum_{x\in\Lambda}m(x)\delta_{x},\]
is an $ \N $-FQ.
\end{thm}
\begin{example}\label{ex:running}
The polynomial $p(z_1,z_2) = 16(1+z_1^2z_2^2) - 8(z_1+z_2+z_1^2z_2+z_1z_2^2) + (z_{1}-z_{2})^{2}$ is Lee-Yang and the vector $ \lv=(\frac{5}{22}\pi,1)$ has $ \Q $-linearly independent entries. Let $ \Lambda=\{t\in\R:p(\exp(it\lv))=0\} $ be the support of $ \mu_{p,\lv} $. \Cref{fig: flow on secman} (top) shows the points of $ \Lambda $ in the interval $ [0,6\pi] $. (bottom left) shows the zero set of $p(e^{ix},e^{iy})$ and line $(x,y)=t\lv$ for $ 0\le t\le 6\pi $. (bottom right) the image of these sets in $\R^2/(2\pi\Z)^2$.
\end{example}

Olevskyii and Ulanovskii \cite{OlevUlan20} proved that any one-dimensional $ \N $-FQ has the form $ \mu=\sum_{x\in\Lambda}m(x)\delta_{x} $, where $ \Lambda $ and $ (m(x))_{x\in\Lambda} $ are the zero set and multiplicities for some real-rooted exponential polynomial $ f $. Together with Cohen \cite{AloCohVin}, the current authors showed that every real-rooted exponential polynomial $ f $ is of the form\footnote{Up to a non-vanishing factor.} $ f(x) = p(\exp(ix\lv)) $, for some Lee-Yang polynomial $ p $ and positive vector $ \lv\in\R_{+}^{n} $ that has $ \Q $-linearly independent entries.
All together this gives: 

\begin{thm}[Inverse result]\cite{OlevUlan20, AloCohVin}\label{thm: every NFQ is mupl}
	Let $ \mu\in\S'(\R) $ be an $ \N $-FQ. Then, $ \mu=\mu_{p,\lv} $ as in the Kurasov-Sarnak construction, for some $ n\in\N $, a Lee-Yang polynomial $ p\in\C[z_{1},z_{2},\ldots,z_{n}] $ and a positive vector $ \lv\in\R_{+}^{n} $ whose entries are $ \Q $-linearly independent.  
\end{thm}
Given a set $ A\subset\R $ let $ \dim_{\Q}(A) $ denote the dimension (as a $ \Q $-vector space) of the $ \Q $-linear span of the elements of $ A $. For a vector $ \lv\in\R^{n} $, $ \dim_{\Q}(\lv)=n $ means that its entries are $ \Q $-linearly independent. 
\begin{thm}\cite[Theorem 3]{KurasovSarnak20jmp}
	  Any $\N$-FQ, say $\mu=\sum_{x\in\Lambda}m(x)\delta_{x}$, has uniformly bounded coefficients, and has two integers $ r,c\ge0 $ such that its support $ \Lambda=L_{1}\cup L_{2}\cup\ldots\cup L_{r}\cup N $ is the union of $ r $ infinite arithmetic progressions and a set $ N $ which, if not empty, has $ \dim_{\Q}(N)=\infty $ and $ |N\cap L|\le c $ for any arithmetic progression $ L $.
\end{thm}

 We elaborate on \cite[Theorem 3]{KurasovSarnak20jmp} and the relation between the decomposition of the measure and the decomposition of the polynomial into irreducible factors (a proof provided in \Cref{sec: proof of elaboration}). A polynomial is said to be \emph{binomial} if it has only two monomials. 

\begin{thm}[Decomposition and non-periodicity]\label{thm: irreduicible} Given an $ \N $-FQ $ \mu $, there are $ p $ and $ \lv $, a Lee-Yang polynomial and a positive $ \Q $-linearly independent vector, such that $ \mu=\mu_{p,\lv} $. The polynomial $ p$ decomposes into distinct irreducible Lee-Yang polynomials $ p=\prod_{j=1}^{N}q_{j}^{c_{j}} $, each factor $ q_{j} $ appears with a power $ c_{j}\in\N $. Let $ \Lambda $ be the support of $ \mu $ and $ \Lambda_{j} $ be the support of $ \mu_{q_{j},\lv} $ for each $ q_{j} $. Then, 
	\[\mu_{p,\lv}=\sum_{j=1}^{N}c_{j}\mu_{q_{j},\lv}\quad\text{ and }\quad  \Lambda=\bigcup_{j=1}^{N}\Lambda_{j}. \]
	If $ q_{j} $ is binomial, then $ \mu_{q_{j},\lv} $ has unit coefficients and $ \Lambda_{j} $ is an infinite arithmetic progression. \\
	If $ q_{j} $ is non-binomial, let $ D $ denote its total degree and let $ \mu_{q_{j},\lv}=\sum_{x\in\Lambda_{j}}m_{j}(x)\delta_{x} $, then
	\begin{enumerate}
		\item (almost all unit coefficients): The coefficients are bounded by $ m_{j}(x)\le D $, and $ m_{j}(x)=1 $ for almost every $ x\in\Lambda_{j} $:   
		\[\lim_{R\to\infty}\frac{\left|\{|x|<R\ : \ x\in\Lambda_{j},\ m_{j}(x)=1\}\right|}{\left|\{|x|<R\ : \ x\in\Lambda_{j}\}\right|}=1.\]
		\item (dimension over $ \Q $): The support has $ \dim_{\Q}(\Lambda_{j})=\infty $ with uniform bounds $ |\Lambda_{j}\cap A|\le c=c(m,D) $ for any set $ A\subset\R $ with $ \dim_{\Q}(A)=m $.
	\end{enumerate}
\end{thm}

Part (2) extends the statement of \cite[Theorem 3]{KurasovSarnak20jmp}, and [Theorem 10.1]\cite{kurasovBook2023}, from arithmetic progressions (encompassed by the case $m=2$) 
to arbitrary $m$, however the proofs of these statements are identical \cite{Sarnak} and similar to the proof of \cite[Theorem 2]{KurasovSarnak20jmp}, which treats a specific instance.  For the sake of completeness, 
we include the proof in \Cref{sec: proof of elaboration}.

\begin{remark}[Quasicrystals and cut-and-project sets]
The mathematical definition of a \emph{quasicrystal} (not to be confused with Fourier quasicrystal) is a set $ \Lambda\subset\R^{n} $ which is uniformly discrete, relatively dense\footnote{Relatively dense means $ \exists R>0 $ such that $ \Lambda $ intersects any ball of radius $ R $.}, and its set of differences $ \Lambda-\Lambda=\{x-y:x,y\in\Lambda\} $ is contained in finitely many translates of $ \Lambda $, see \cite[Definition 6]{Meyer1995quasicrystals}. A \emph{model set} (also known as cut-and-project set) $ \Lambda\subset \R^{n} $ is the projection of a set $ (B\times\R^{n})\cap L  $, where $ L\subset \R^{m}\times\R^{n} $ is a lattice in generic location\footnote{Lattice in $ \R^{m}\times\R^{n} $ such that the projection to $ \R^{m} $ is dense and the projection to $ \R^{n} $ is injective} and $ B\subset\R^{m} $ which is bounded with non-empty interior. Meyer showed that any model set is a quasicrystal, and any quasicrystal lies in finitely many translates of model sets \cite[Theorem 1]{Meyer1995quasicrystals}.  In particular, in such case $ \dim_{\Q}(\Lambda)\le n+m $.
\end{remark}
\begin{cor}
	If $ p $ is an irreducible non-binomial Lee-Yang polynomial, then the support of $ \mu_{p,\lv} $, for any $ \Q $-linearly independent $ \lv\in\R_{+}^{n} $, intersect any quasicrystal and any model set in at most finitely many points. 
\end{cor}
\begin{remark}[non-uniqueness of the decomposition]
	The decomposition in \Cref{thm: irreduicible} is not unique due to the fact that the ring of exponential polynomials is not a factorization ring. For $n=1$, one can factor the exponential polynomial $1 - \exp(ix)$ into arbitrarily many factors, starting
	\[1 - \exp(ix) = (1 + \exp(ix/2))(1 - \exp(ix/2))
	= (1 + \exp(ix/2))(1 + \exp(ix/4))(1 - \exp(ix/4))
	\]
	and so on. For $n> 1$, this decomposition can also fail to be unique in non-trivial ways. The Lee-Yang polynomial $ p(z_{1},z_{1}) $ in \Cref{ex:running} is irreducible, however $p(z_1^2,z_2^2)$ factors as the product of four Lee-Yang polynomials $q_{\sigma} = 2 + \sigma_1 z_1 + \sigma_2 z_2 + 2\sigma_1 \sigma_2 z_1 z_2$ for $(\sigma_1, \sigma_2) \in \{\pm 1\}^2$. Therefore, for any $\ell\in \R_+^2$,
	\[\mu_{p, \ell}=\sum_{\sigma\in\{\pm1\}^{2}}\mu_{q_{\sigma},\frac{1}{2}\lv}.\]
\end{remark}

\subsection{Main Results}
To set up notations, let $ \A_{\bd}(n) $ denote the set of Lee-Yang polynomials $ p\in\C[z_{1},\ldots,z_{n}] $ of degrees $ \bd=(d_{1},\ldots,d_{n}) $, i.e., $ p $ that has degree $ d_{j} $ in every $ z_{j} $. Let $ |\bd|=d_{1}+d_{2}+\ldots+d_{n} $ denote the total degree. Let $ \T^{n}=\{\z\in\C^{n}:|z_{j}|=1\ \forall j\} $.  
\begin{thm}[density and maximal gap]\label{thm: zeros density and upper bound}
	Let $ p\in\A_{\bd}(n) $ and $ \lv\in\R_{+}^{n} $. Then,
	\begin{enumerate}
		\item The growth rate of $ \mu_{p,\lv} $ is 
		\[\mu_{p,\lv}([x,x+T])=\frac{\langle\bd,\lv\rangle}{2\pi}T+\mathrm{err}(x,T),\qquad\mbox{for all}\quad x\in\R,\ T>0.\]
		with a uniformly bounded error term $ |\mathrm{err}(x,T)|\le|\bd| $. 
		\item The gap between any pair of consecutive atoms in $ \mu_{p,\lv} $ is at most $2\pi\frac{|\bd|}{\langle\bd,\lv\rangle} $.
	\end{enumerate}
\end{thm}
\begin{remark}
	The bounds in \Cref{thm: zeros density and upper bound} are tight. For any choice of $ n\in\N $ and $ \bd\in\N^{n} $, let $ p(\z)=\prod_{j=1}^{n}(1-z_{j})^{d_{j}}\in\A_{\bd}(n) $ and $ \lv=(2\pi,2\pi,\ldots,2\pi) $. Then, $ \mu_{p,\lv}=\sum_{x\in\Z}|\bd|\delta_{x} $, with gaps that are all equal to $ 1= 2\pi\frac{|\bd|}{\langle\bd,\lv\rangle}$. The error term for $ \mu_{p,\lv}([-\epsilon,\epsilon])=|\bd| $ is $ \mu_{p,\lv}([-\epsilon,\epsilon])-\frac{\langle\bd,\lv\rangle}{2\pi}2\epsilon=  |\bd|-2|\bd|\epsilon $ for arbitrary small $ 1>\epsilon>0 $.
\end{remark}
The next theorem shows that generically, an $ \N $-FQ enjoys the desired properties of having uniformly discrete support and having all unit coefficients.
\begin{thm}[minimal gap for generic FQ]\label{thm: minimal gap} Let $ n\ge 2 $ be an integer and $ \bd\in\Z_{>0}^{n} $. For any $ \Q $-linearly independent $ \lv\in\R_{+}^{n} $ and $ p\in\A_{\bd}(n) $, the measure $ \mu_{p,\lv} $ is non-periodic with unit coefficients and uniformly discrete support if and only if $ p $ satisfies:
	\begin{enumerate}
		\item[(i)]$ \nabla p(\z)\ne0 $ whenever $ \z\in\T^{n} $ such that $ p(\z)=0 $, and
		\item[(ii)] $ p $ has a non-binomial factor.   
	\end{enumerate} 
The set of Lee-Yang polynomials in $ \A_{\bd}(n) $ that satisfies both (i) and (ii) is a semi-algebraic open dense subset of $ \A_{\bd}(n) $. 
 Furthermore, we provide an explicit perturbation taking any $ p\in\A_{\bd}(n)  $ to a one parameter family polynomials $ p_{\lambda}=p+\sum_{j=1}^{|\bd|-1}\lambda^{j}q_{j} $, such that $ p_{0}=p $ and $ p_{\lambda} $ satisfies (i) for any $ \lambda>0 $.
\end{thm}

\Cref{fig: perturbation}  shows the effect of the perturbation $ p\mapsto p_{\lambda} $ for the \Cref{ex:running}. 
\begin{remark}
	There is no loss of generality by considering only $ \Q $-linearly independent $ \lv $'s, due to \cite{AloCohVin}. Nevertheless, we point out that if $ p$ satisfies (i) then $ \mu_{p,\lv} $ will have unit coefficients and uniformly discrete support for any $ \lv\in\R_{+}^{n} $.
	Kurasov and Sarnak also exploit this fact for a specific instance in \cite[Theorem~2]{KurasovSarnak20jmp} 
	and more general Lee-Yang polynomials coming from metric graphs in a forthcoming  paper \cite{KS_forthcoming}. 
	\end{remark}

\begin{figure}[h!]
	\includegraphics[width=0.4\textwidth]{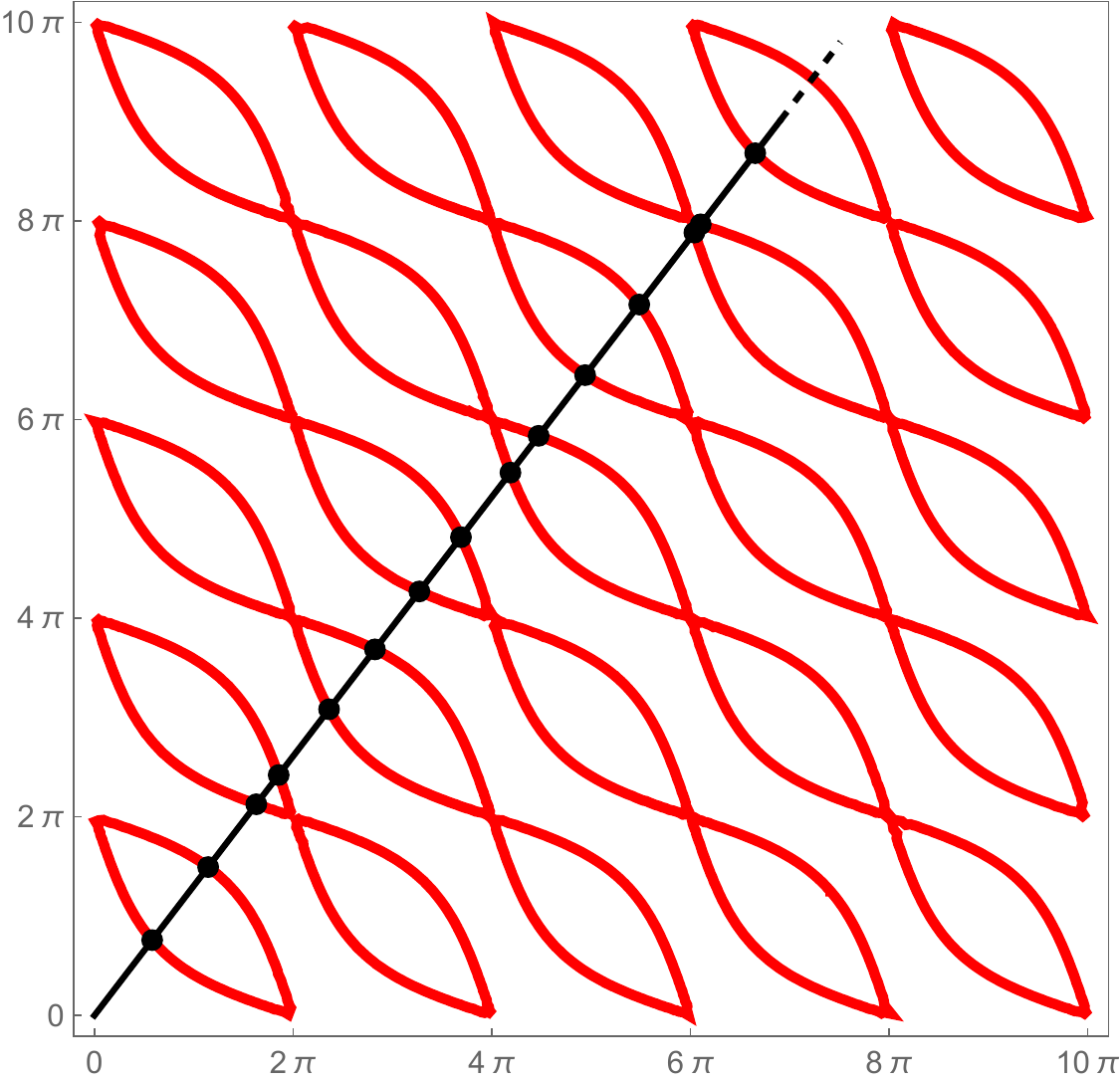}  \hspace{.5in}	\includegraphics[width=0.4\textwidth]{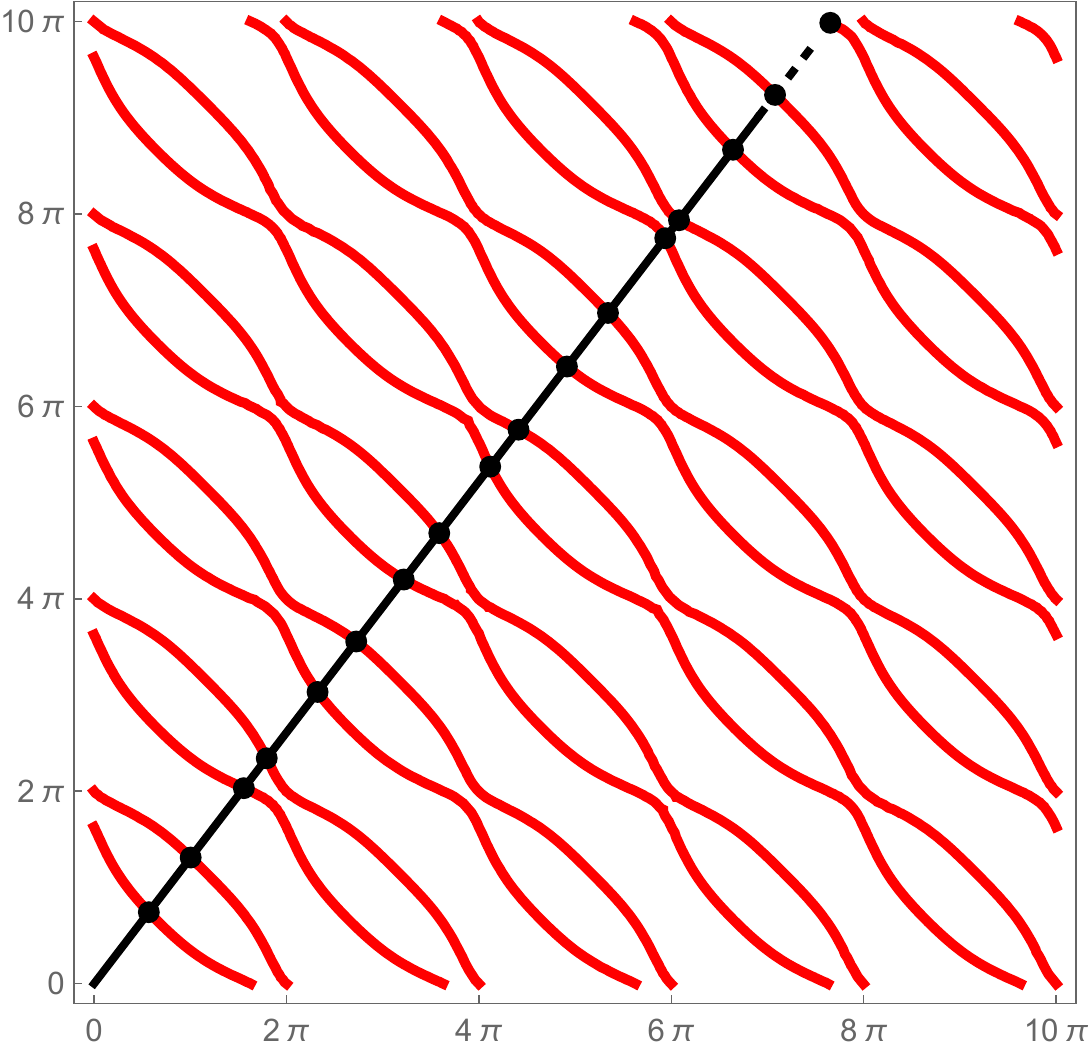} 
	\caption{
		(Left) The singular zero set of $ p $  and the line in direction $ \lv $, as in \Cref{fig: flow on secman}. (Right) The regular zero set of the perturbed polynomial $ p_{\lambda} $ for $ \lambda=0.2 $ and the same line in direction $ \lv $.}
	\label{fig: perturbation}
\end{figure}

Given $ p\in\A_{\bd}(n)$ and $ \lv\in\R_{+}^{n}$, let $ (x_{j})_{j\in\Z} $ be the zeros $ p(\exp(ix\lv)) $, numbered increasingly with multiplicity, so that $ \mu_{p,\lv}=\sum_{j\in\Z}\delta_{x_{j}} $. A random measure of the form $ \sum_{j\in\Z}\delta_{x_{j}}  $, for random $ x_{j} $'s, is called a \emph{point process} and it can be defined it in terms of the gaps $ \Delta_{j}=x_{j+1}-x_{j} $ which are often taken to be i.i.d $ \Delta_{j} \sim \rho$ for some probability distribution $ \rho$. The next theorem shows that the deterministic gaps between atoms in $ \mu_{p,\lv} $ obey a well defined ``gaps-distribution'' $ \rho=\rho_{p,\lv} $. We say that $ \mu_{p,\lv} $ has a \emph{gap distribution} $ \rho $, if 
\[\frac{1}{N}\sum_{j=1}^{N}\delta_{(x_{j+1}-x_{j})}\weak \rho,\]
where $ \weak $ stands for convergence in distribution, which means that for any continuous $ f$,
\begin{equation}\label{eq: gap distribution def}
	\lim_{N\to\infty}\frac{1}{N}\sum_{j=1}^{N}f(x_{j+1}-x_{j})=\int fd\rho.
\end{equation}

\begin{thm}[Existence of gaps distribution]\label{thm: existance of gaps distribution} Every $ \N $-valued FQ $ \mu $ has a well-defined gap distribution $ \rho $ with the following properties:
	\begin{enumerate}
		\item It has finitely many atoms, say $ (r_{j})_{j=1}^{M} $, such that $ \rho=\rho_{ac}+\sum_{j=1}^{M}\rho(\{r_{j}\})\delta_{r_{j}} $, and $ \rho_{ac} $ is absolutely continuous with respect to the Lebesgue measure on $ \R $. 
		\item $ \rho_{ac} =0$ if and only if $ \mu $ is periodic.
		\item If $ \Delta\ge0 $ is any gap between consecutive atoms of $ \mu $, then $ \rho(I)>0 $ for any open neighborhood $ I\subset\R $  of $ \Delta$.
	\end{enumerate} 
\end{thm}

As discussed above, every $ \N $-valued FQ $ \mu $ can 
be written as $\mu_{p, \ell}$ for some Lee-Yang polynomial $p$ and vector $\ell\in \R_+^n$ where $\ell$ has linearly independent entries over $\Q$.  
We explore the dependence of the gap distribution $\rho_{p,\ell}$ on both the polynomial $p$ and vector $\ell$.  
First we note that the gap distribution is independent of torus actions on $p$ and give conditions on the factorization of $p$ under which the $\rho_{p,\ell}$  has atoms. 

\begin{thm}[$ p $ dependence of the gap distribution]\label{thm: p dependence} Suppose $ p\in\A_{\bd}(n) $ and let $ \lv\in\R_{+}^{n} $ with $ \Q $-linearly independent entries, then
	\begin{enumerate}
		\item For any fixed $ \xv\in\R^{n} $, the polynomial $ q(\z)=p(\exp(i\xv)\z)=p(e^{ix_{1}}z_{1},\ldots,e^{ix_{n}}z_{n}) $ is in $ \A_{\bd}(n) $, and $ \rho_{q,\lv}=\rho_{p,\lv} $.
	\item The distribution $ \rho_{p,\lv} $ has an atom at $ \Delta\ge0 $ if and only if there are two irreducible factors of $ p $, say $ q_{i} $ and $ q_{j} $, such that $ q_{j}(\z)=q_{i}(\exp(i\Delta\lv)\z) $. Moreover,
	\item if $ \Delta>0 $ and $ q_{i}=q_{j} $, namely $q_{i}(\z)= q_{i}(\exp(i\Delta\lv)\z) $, then $ q_{i} $ is binomial. 		
		\end{enumerate}     
\end{thm}
\begin{cor}\label{cor: atoms}Suppose $ p\in\A_{\bd}(n) $ and let $ \lv\in\R_{+}^{n} $ with $ \Q $-linearly independent entries.
	\begin{enumerate}
			\item If $ p $ is irreducible and not binomial then $ \rho_{p,\lv} $ is absolutely continuous.   
			\item If $ p $ is binomial (i.e, $ \mu_{p,\lv} $ is a Dirac Comb), then $ \rho_{p,\lv} $ is the atomic measure at $ \frac{2\pi}{\langle\bd,\lv\rangle} $. 
			\item $ \rho_{p,\lv} $ has an atom at $ 0 $ if and only if $ p $ has a square factor.
			\item Suppose that $ p $ has $ N+M $ distinct irreducible factors, $ M $ which are binomial and $ N $ non-binomial. Then, $ \rho_{p,\lv} $ has at most $ { N \choose 2 }+M+1 $ atoms.
	\end{enumerate}   
\end{cor}

Next we show that the gap distribution $\rho_{p,\ell}$ varies 
continuously in $\ell$ when we restrict to vectors $\ell$ with 
$\Q$-linearly independent entries.  
For arbitrary $\ell\in \R_+^n$, this gives rise to a well-defined 
limiting distribution $\nu_{p,\ell}$ that agrees when $\rho_{p,\ell}$
when $\ell$ has $\Q$-linearly independent entries. 
The limiting measure $ \nu_{p,\lv} $ is defined explicitly in \Cref{defn: nu}.

\begin{thm}[$ \lv $ dependence of the gap distribution]\label{thm: l dependenece } Let $ p\in\A_{\bd}(n) $ and $ \lv\in\R_{+}^{n} $. Then $ \rho_{p,\lv} $ is supported inside $ [0,2\pi\frac{|\bd|}{\langle\bd,\lv\rangle}] $. 
There is a distribution $\nu_{p,\ell}$ 
such that for any converging sequence $\lv^{(j)}\to\lv$ in which each $ \lv^{(j)} $ has $ \Q $-linearly independent entries, 
	\[\rho_{p,\lv^{(j)}}\weak \nu_{p,\lv}.\]
In particular, $ \nu_{p,\lv}=\rho_{p,\lv} $ whenever $ \lv $ has $ \Q $-linearly independent entries.  
\end{thm}
\begin{center}
\begin{figure}[h!]
	\includegraphics[width=0.4\textwidth]{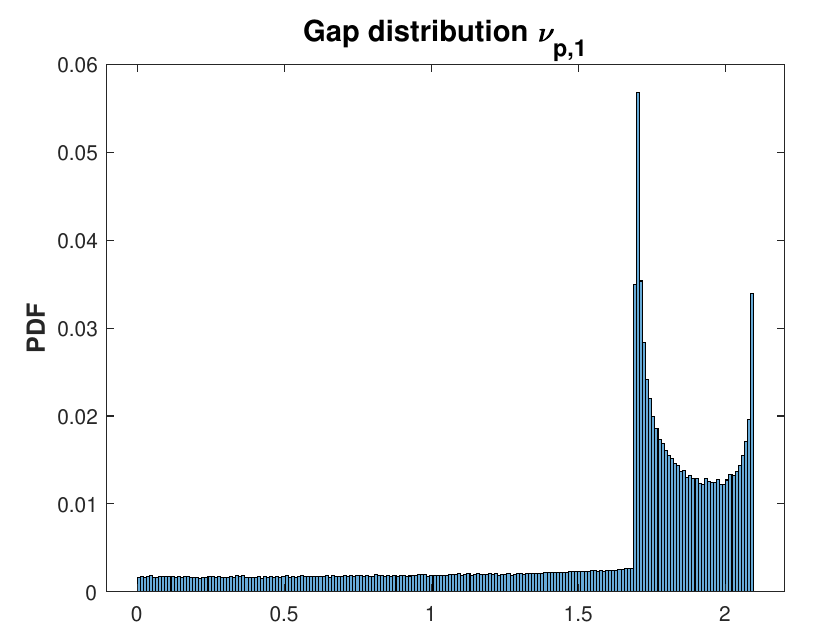}  \hspace{.5in}	\includegraphics[width=0.4\textwidth]{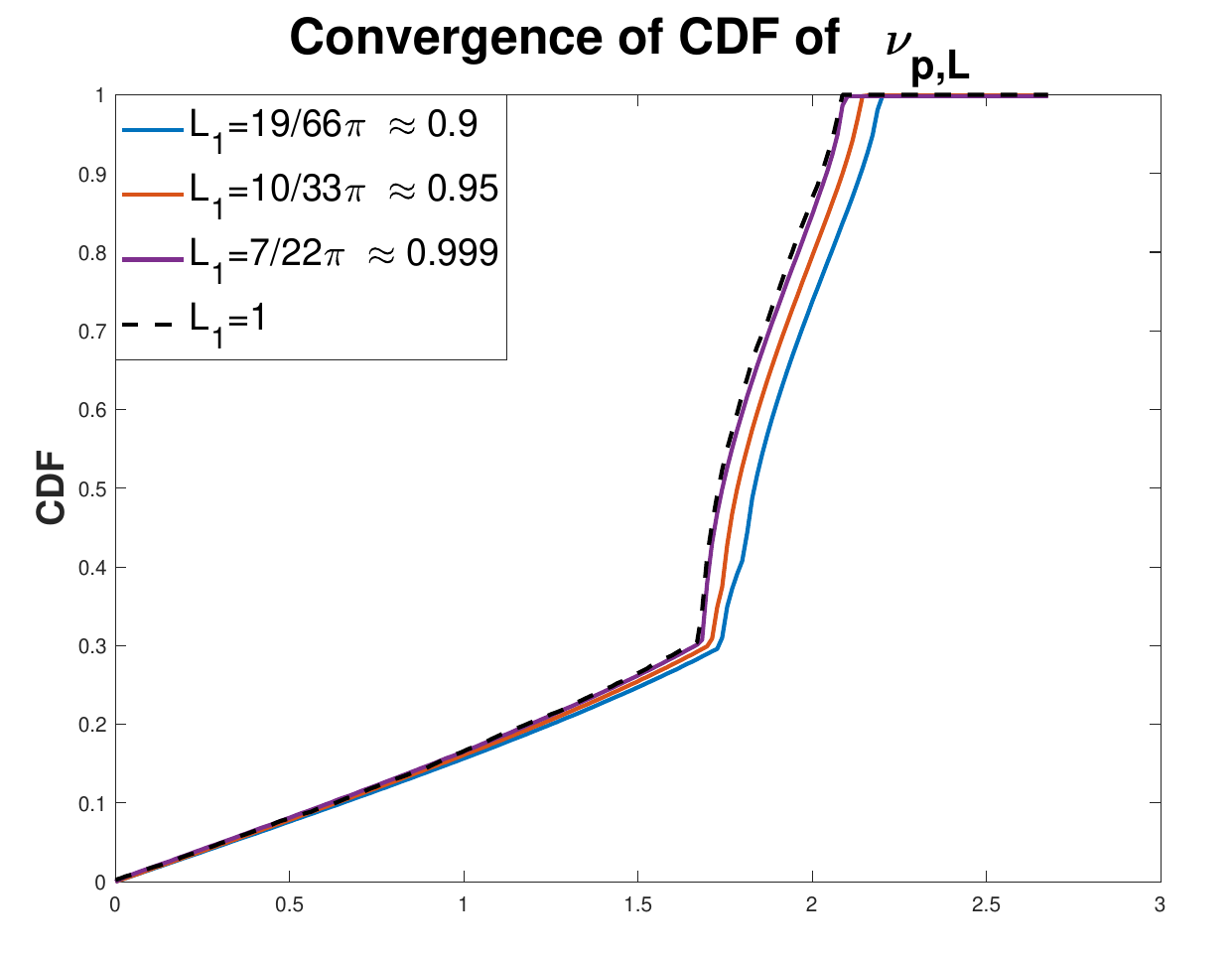} 
		\caption {Example of $ \rho_{p,\lv}\weak \nu_{p,\textbf{1}} $. The gap distributions for $ p $ as in \Cref{fig: flow on secman}, and $ \lv=(L_{1},1) $ with $ L_{1} $ converging to $ 1 $. (Left) the probability distribution function of $ \nu_{p,\textbf{1}} $. (Right) the cumulative distribution functions of $ \nu_{p,\textbf{1}} $ (dashed) and of $ \rho_{p,\lv} $ for three different values of $ L_{1} $. Each $ \rho_{p,\lv} $ was computed from the gaps in the interval $ [0,10^{4}] $, while $ \nu_{p,\lv} $ was computed as in \Cref{thm: the l to 1 case}, by sampling $ 10^4 $ random points on the torus.}
	\label{fig: gap distribution convergence}
\end{figure}
\end{center}
A particularly interesting case is the limit $ \nu_{p,\textbf{1}} $ for $ \lv=\textbf{1}:=(1,1,\ldots,1) $, which can be calculated explicitly, as follows. \Cref{fig: nup1} displays the distributions $ \nu_{p,\textbf{1}} $ for two important examples of Lee-Yang polynomials $p$. \begin{Def}
	Let $ p\in\A_{N}(1) $ be a univariate Lee-Yang polynomial of degree $ N $ and denote its roots by $ \{e^{i\theta_{j}}\}_{j=1}^{N} $ with $ 0\le\theta_{1}\le\ldots\le\theta_{N}<2\pi $. By convention $ \theta_{N+1}=\theta_{1}+2\pi $. Then the \emph{gap distribution} of $ p $ is a probability measure on $ [0,2\pi] $ given by 
	\[\mathrm{gaps}(p)=\frac{1}{N}\sum_{j=1}^{N}\delta_{\theta_{j+1}-\theta_{j}}.\]
	If $ U $ is a unitary matrix, then $ p(s)=\det(s-U) $ and $ q(s)=\det(1-sU) $ have the same gap distribution, and we denote it by $ \mathrm{gaps}(U) $.
\end{Def}
For a fixed $ p\in\A_{\bd}(n) $ and a fixed point $ \xv\in[0,2\pi]^{n} $, define the univariate polynomial $ p_{\xv}(s):=p(se^{ix_{1}},se^{ix_{2}},\ldots,se^{ix_{n}}) $ so that $ p_{\xv}\in\A_{N}(1) $ with $ N=|\bd| $. We may then take $ \xv $ uniformly at random.
\begin{thm}[$ \lv\to\textbf{1} $]\label{thm: the l to 1 case}
	Let $ p\in\A_{\bd}(n) $. Let $ \xv $ be a uniformly random point in $ [0,2\pi]^{n} $. Then $ \nu_{p,\textbf{1}} $, for $ \lv=\textbf{1}$, is given by 
	\[\nu_{p,\textbf{1}}=\mathbb{E}\left[\mathrm{gaps}(p_{\xv})\right].\]
	Namely, for any sequence $ \lv^{(j)}\to\textbf{1} $, such that each $ \lv^{(j)} $ has $ \Q $-linearly independent entries, 
	\[\lim_{j\to\infty}\int f d\rho_{p,\lv^{(j)}}=\frac{1}{(2\pi)^{n}}\int_{\xv\in[0,2\pi]^{n}}\left[\frac{1}{|\bd|}\sum_{j=1}^{|\bd|}f(\theta_{j+1}(\xv)-\theta_{j}(\xv))\right]d\xv, \quad \forall f\in C(\R),\] 
	
	where $ \{e^{i\theta_{j}(\xv)}\}_{j=1}^{|\bd|} $ are the ordered roots of $ p_{\xv} $ for every $ \xv $, and $ \theta_{|\bd|+1}(\xv):=\theta_{1}(\xv)+2\pi $.  
\end{thm}
Using \Cref{thm: the l to 1 case} we can provide examples of limiting gap distributions that correspond to the following special distributions.
\begin{example}[Poisson]
 If $ p(z_{1},\ldots,z_{n})=\prod_{j=1}^{n}(1-z_{j}) $, then $ \nu_{p,\textbf{1}} $ is the distribution of gaps between $ n $ random points in the circle, chosen unifromly and independently. It is well known that this distribution converges to Poisson distribution in the limit of $ n\to\infty $.   
\end{example}

\begin{example}[CUE]
 Given a fixed unitary $ n\times n $ matrix $ u $, let $ p_{u}(z_{1},\ldots,z_{n}):=\det(1-\diag(z_{1},\ldots,z_{n})u) $. Then,
		\[\nu_{p_{u},\textbf{1}}=\mathbb{E}\left[\mathrm{gaps}(\diag(\exp(i\xv))u)\right],\quad\xv\sim U([0,2\pi]^{n}).\] 
		For a random $ u $, Haar uniformly from $ U(n) $, the \emph{empirical gap distribution} is 
		\[\mathbb{E}(\nu_{p_{u},\textbf{1}})=\mathbb{E}\left[\mathrm{gaps}(\diag(\exp(i\xv))u)\right]=\mathbb{E}\left[\mathrm{gaps}(u)\right],\quad u\sim\mathrm{Haar}(U(n)).\]
The distribution $ \mathbb{E}\left[\mathrm{gaps}(u)\right] $ for $ u\sim\mathrm{Haar}(U(n)) $ is well known, and as $ n\to\infty $ it converges to the CUE (circular unitary ensemble) gap distribution. 
\end{example}
\begin{remark}[Computation of gap distributions]
	For any positive constant vector $ t\textbf{1} $, the measure $\nu_{p,t\textbf{1}}$ scales by $\nu_{p,t\textbf{1}}([a,b])=\nu_{p,\textbf{1}}([\frac{a}{t},\frac{b}{t}])$. For any rational $ \textbf{q}=\frac{1}{m}(k_{1},\ldots,k_{n})\in\Q_{+}^{n} $, the polynomial $ \tilde{p}(\z)=p(z_{1}^{k_{1}},\ldots,z_{n}^{k_{n}}) $ is Lee-Yang, and $ \nu_{p,\textbf{q}}=\nu_{\tilde{p},m\textbf{1}} $.
	
	According to \Cref{thm: l dependenece } any gap distribution $ \rho_{p,\lv} $, with $ \Q $-linearly independent $ \lv $, can be approximated by $ \nu_{p,\textbf{q}} $ for a rational approximation  $ \textbf{q} $ of $ \lv $. Moreover, $ \nu_{p,\textbf{q}} $ can be efficiently computed as an average through \Cref{thm: the l to 1 case} and the argument above.
\end{remark}

\begin{figure}
	 \includegraphics[width=0.4\textwidth]{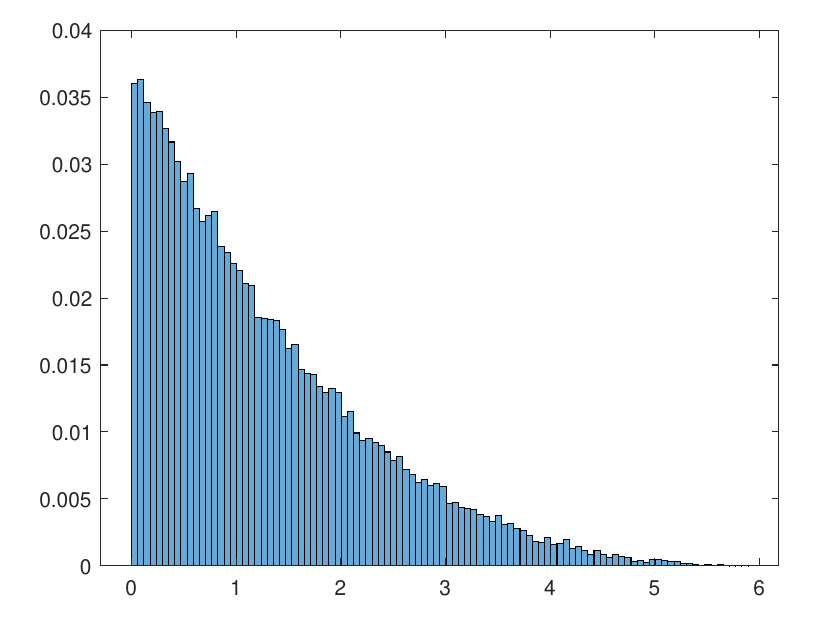}\hspace*{0.5in}\includegraphics[width=0.4\textwidth]{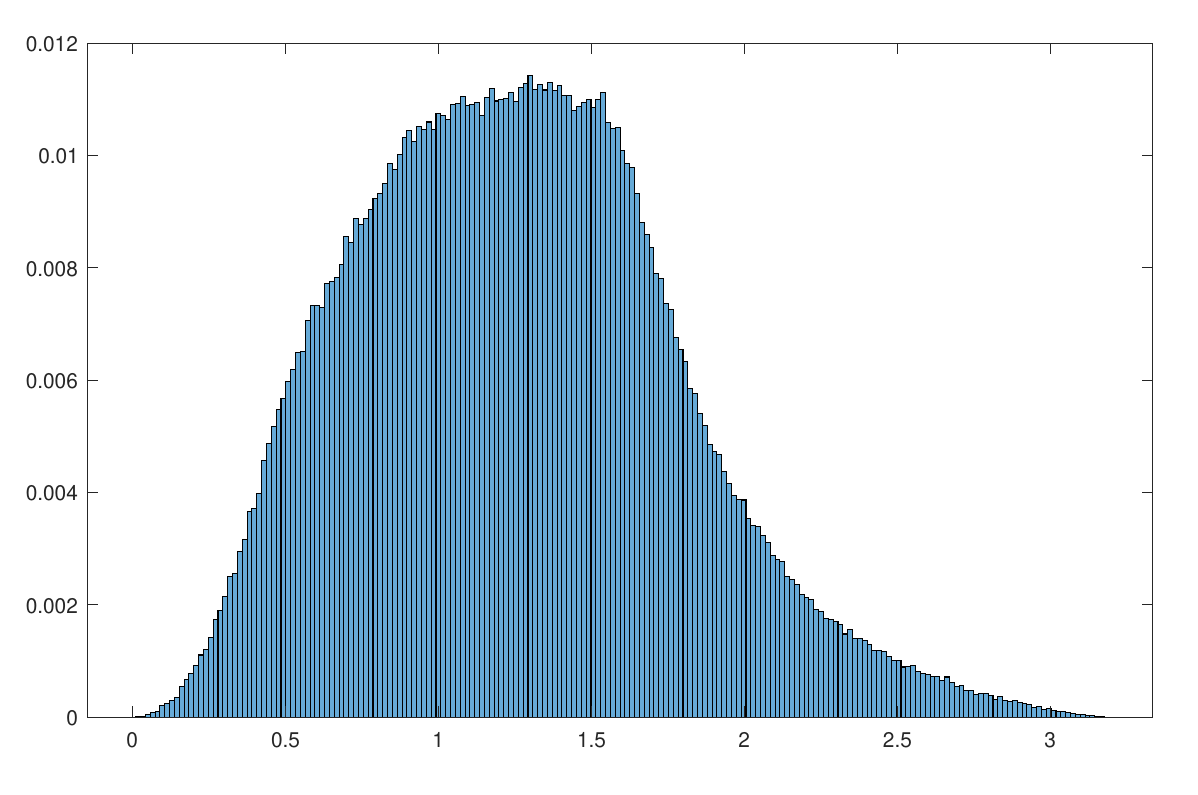}
	 \caption{(Left) $ \nu_{p,\textbf{1}} $ for $ p(\z)=\prod_{j=1}^{5}(z_{j}-1) $. (Right) $ \nu_{p,\textbf{1}} $ for $ p(\z)=\det(1-\diag(\z)U) $ for a fixed $ 5\times5 $ unitary matrix $ U $, chosen at random (Haar uniformly). Both calculated with $ 10^{4} $ random points from the torus.}\label{fig: nup1}
\end{figure}
       

The paper is organized as follows. The first two sections provides background and preliminary results,  \Cref{sec:backgroundCrystals} on crystalline measures and FQ's, and \Cref{sec:backgroundStable} on Lee-Yang polynomials and real stable polynomials. The torus zero sets of Lee-Yang polynomials are analyzed in \Cref{sec: torus zero set}. \Cref{thm: zeros density and upper bound}, the growth rate and upper bound on the gaps, is proved in \Cref{sec: proof of thm zeros density and upper bound}. In \Cref{sec: ergodicity} an ergodic dynamical system is defined on the torus zero set, which is being used in the subsequent sections. \Cref{thm: irreduicible}, decomposition and non-periodicity, is proven in \Cref{sec: proof of elaboration}. \Cref{thm: minimal gap}, minimal gap and genericity, is proved in \Cref{sec: minimal gap}. \Cref{sec: gap distributions} focus on gap distributions, in which \Cref{thm: existance of gaps distribution}, \Cref{thm: p dependence}, \Cref{thm: l dependenece }, and \Cref{thm: the l to 1 case}, are proved.
This is summarized as follows:

\setcounter{tocdepth}{1}
\tableofcontents

{\bf Acknowledgements.} The work began while both authors were visitors at the Institute for Advanced Studies. 
We would like to thank Peter Sarnak and Pavel Kurasov for many enlightening discussions.  
The first author is supported by the Simons Foundation Grant 601948, DJ.
The second author was partially supported by NSF grant No.~DMS-2153746 and a Sloan Research Fellowship. 
This material is based upon work directly supported by the National Science Foundation Grant No.~DMS-1926686, and indirectly supported by the National Science Foundation Grant No.~CCF-1900460. 

\section{Preliminaries on crystalline measures and FQ}\label{sec:backgroundCrystals}
A \emph{Schwartz function} on $\R$ is a smooth function $ f\in\C^{\infty}(\R,\C) $ that decays, as $|x|\to\infty  $, faster than any polynomial in $ |x| $, and so does any of its derivatives. The \emph{Schwartz space} $ \mathcal{S}(\R) $ is the infinite dimensional vector space of Schwartz functions. It can be defined in terms of the seminorms $ \|f\|_{n,m}:= \sup_{x\in\R}|x^{n}\left(\frac{d}{d x}\right)^{m}f(x)|$,  
\begin{equation*}
	\mathcal{S}(\R):=\{f\in\C^{\infty}(\R,\C)\ :\ \|f\|_{n,m}<\infty\quad \text{for all}\ m,n\in\Z_{\ge0}\},
\end{equation*}
and it is a complete metric space with respect to the metric $ d(f,g):=\sum_{n,m=0}^{\infty}\frac{\|f\|_{n,m}}{2^{n+m}(1+\|f\|_{n,m})} $. A ($\C$ valued) Borel measure $ \mu $ on $ \R $ is \emph{tempered} if $ \langle f,\mu\rangle:=\int f d\mu $ is finite for all $ f\in\mathcal{S}(\R) $. The vector space of tempered measures is the dual of $\S(\R)$ and is denoted by $\S^{'}(\R)$. The \emph{Fourier transform} $\mathcal{F}(f):=\hat{f}$, with $ 	\hat{f}(k):= \int_{-\infty}^{\infty}f(x)e^{-ikx}dx $, is a linear automorphism of $ \mathcal{S}(\R) $, and it defines an automorphism on the dual space. Given a measure $ \mu\in\S^{'}(\R) $, its Fourier transform is the measure $ \hat{\mu}\in \S^{'}(\R) $ defined by $ \langle f,\hat{\mu}\rangle:=\langle \hat{f},\mu\rangle$ for all $ f\in\S(\R) $.
Let $ \delta_{x}\in\S^{'}(\R)$ denote the atom at $ x\in\R $ (also known as Dirac delta at $ x $), which is defined by $ \langle f, \delta_{x}\rangle:=f(x)$. We say that a measure $ \mu $ is \emph{discrete} if it is supported on a discrete (locally finite) set, in which case it can be written as  
\begin{equation}\label{eq: discrete measusre}
	\mu=\sum_{x\in\Lambda}a_{x}\delta_{x}:=\lim_{T\to\infty}\sum_{x\in\Lambda\cap[-T,T]}a_{x}\delta_{x}
\end{equation}
with complex coefficients $a_{x}\in\C$ and discrete support $\Lambda\subset\R$. Whenever we write an infinite sum as in \eqref{eq: discrete measusre}, it should be understood as the $ T\to\infty $ limit of the $ [-T,T] $ truncated sum. One can check that a discrete measure $ \mu $ is tempered, i.e.  $ \mu\in\S^{'}(\R) $ , if and only if $ \mu([-T,T]) $ is bounded by some polynomial in $ T $, namely if there exist $ C>0 $ and $ m\in\N $ such that
\begin{equation*}
	\left|\sum_{x\in\Lambda\cap[-T,T]}a_{x}\right|\le C(1+T^{m}),\qquad \forall T>0.
\end{equation*}
If $ \mu $ is a complex valued measure given by \eqref{eq: discrete measusre}, then $ |\mu|:=\sum_{x\in\Lambda}|a_{x}|\delta_{x} $. 
\begin{Def}[crystalline measure and FQ]\cite{Meyer2016measures,LevOlev2016quasicrystals}
	A \emph{crystalline measure} is a discrete tempered distribution whose Fourier transform is also discrete\footnote{The Fourier transform is tempered by definition.}. A \emph{Fourier quasicrystal} (FQ) is a crystalline measure $ \mu $ with the further restriction that $ |\mu| $ and $ |\hat{\mu}| $ are also tempered. To write it explicitly, $ \mu $ is an FQ if there exists discrete (locally finite) sets $ \Lambda, S $, and complex coefficients $ (a_{x})_{x\in\Lambda}, (c_{k})_{k\in S} $, such that
	\begin{equation}\label{eq: crystalline measure}
		\mu=\sum_{x\in\Lambda}a_{x}\delta_{x},\quad \hat{\mu}=\sum_{k\in S}c_{k}\delta_{k},\quad \text{and}\quad \sum_{x\in\Lambda\cap[-T,T]}|a_{x}|+\sum_{k\in S\cap[-T,T]}|c_{k}|\le C(1+T^{m}),
	\end{equation}
	for some $ C>0,m\in\N $, for all $ T>0 $. 
\end{Def}
For example, the measure $ \mu=\sum_{x\in\Lambda}\delta_{x} $ for any periodic $ \Lambda $ is an FQ due to the Poisson summation formula.  


\section{Preliminaries on Lee-Yang polynomials}\label{sec:backgroundStable}
Let $\C[\z]$ denote the space $\C[z_1, \hdots, z_n]$ of polynomials in indeterminates 
$\z = (z_1, \hdots, z_n)$. 
For a nonnegative integer vector $\balpha = (\alpha_1, \hdots, \alpha_n)\in \Z_{\ge}^n$, 
we use $\z^{\balpha}$ to denote the monomial $\prod_{j=1}^n z_j^{\alpha_j}$. 
The degree of a polynomial $p = \sum_{\balpha}a_{\balpha}\z^{\balpha}$ in $\C[\z]$ 
in the variable $z_j$, denoted  $\deg_j(p)$, is the maximum value of $\alpha_j$ appearing in a monomial 
with nonzero coefficient $a_{\alpha}\neq 0$. 
For ${\bf d}=(d_{1},\ldots,d_{n})\in \N^{n}$, let $\C[\z]_{\le {\bf d}}$ 
denote the $\C$-vectorspace of polynomials with $\deg_j(p)\leq d_j$ in each variable $z_j$, i.e. 
$\C[\z]_{\leq \bd} = \left\{ \sum_{0\le\balpha\le\bd}a_{\balpha}\z^{\balpha} : a_{\balpha} \in \C\right\},$
where $\balpha \leq \bd$ is taken coordinate-wise. 

Given a circular region in the complex plane $C\subseteq \C$, we say that $p$ is \emph{stable}
with respect to $C$ if $p$ has no zeroes in $C^n$. For us, the circular regions  of interest
will be the upper half plane $\H_+ = \{z\in \C : {\rm Im}(z)>0\}$, lower half plane $\H_- = \{z\in \C : {\rm Im}(z)<0\}$, and the open unit disk 
$\D = \{ z \in \C : |z|<1\}$. Stability with respect to $\D$ is often known as \emph{Schur stability}. 
We use $\T$ to denote the unit circle $\{ z \in \C : |z|=1\}$ in $\C$ and $\overline{\D}$ for the closed 
unit disk $\D \cup \T$.  Of particular interest are polynomials stable with respect $\D$ and its inverse $\C\backslash \overline{\D}$. 

\begin{Def}
	We say that $p\in \C[\z]$ is a \emph{Lee-Yang polynomial} if it is stable with respect to both $\D$ and $\C\backslash \overline{\D}$
	and use  $\A_{\bd}$ to denote the set of Lee-Yang polynomials in $\C[\z]_{\le \bd}$ of multidegree equal to ${\bf d}$.
	That is, $\A_{\bd}$ is the set of polynomials $p = \sum_{0\le\balpha\le\bd}a_{\balpha}\z^{\balpha}$ 
	so that $\deg_{j}(p) =d_j$ for all $j$ with the property that $p(z_1, \hdots, z_n)\neq 0$ whenever
	$|z_j|<1$ for all $j$ or $|z_j|>1$ for all $j$. When $ n $ is not clear from the context, we will write $ \A_{\bd}(n) $. 
\end{Def}

One property of stability that we will often use is that the set of multivariate polynomials 
that is stable with respect to either an open disk or halfplane is closed in the Euclidean topology 
on $\C[{\bf z}]_{\leq \bd}$. This follows immediately from Hurwitz's Theorem: 

\newtheorem*{thmHurwitz*}{Hurwitz's Theorem}
\begin{thmHurwitz*}[Theorem 1.3.8 of \cite{MR1954841}] Let $\Omega \subseteq \C^m$ be a connected open set and $(f_n)_{n\in \N}$ a 
sequence of functions, each analytic and nonvanishing on $\Omega$, that converges to a limit $f$ uniformly on compact subsets of $\Omega$. Then $f$ is either nonvanishing on $\Omega$ or identically zero. 
\end{thmHurwitz*}

M\"obius transformations  map between circular regions in $\C$. 
Given a tuple of M\"obius transformations  $\phi = (\phi_j(z_j))_j$ where $\phi_j(z)= \frac{a_jz + b_j}{c_jz+d_j}$ and polynomial $p\in \C[\z]_{\leq \bd}$, define 
\[
\bphi\cdot p = \prod_{j=1}^n (c_jz_j + d_j)^{\deg_j(p)} \cdot p(\phi_1(z_1), \hdots, \phi_n(z_n)) \in \C[\z]_{\leq \bd}.
\]
We will sometimes abuse notation and, for a single M\"obius transformations  $\phi(z)= \frac{az + b}{cz+d}$, 
use $\phi\cdot p$ to denote $(\phi, \hdots, \phi)\cdot p$. 
Then $p$ is stable with respect to a region $C$ if and only if $\phi\cdot p$ 
is stable with respect to $\phi^{-1}(C)$.  See \cite[Lemma 1.8]{BB2}. 
Now we will often fix $\phi$ to be a M\"obius transformation taking $\H_+$ to $\D$.
Explicitly, for any fixed $\theta\in [0,2\pi)$, consider the pair
\begin{equation} \label{eq:Mobius}
	\phi(z) = \frac{e^{i\theta}(z-i)}{z+i} \ \ \text{ and } \ \ 
	\phi^{-1}(z) = \frac{-i (z+ e^{i\theta})}{z-e^{i\theta}} \ \ 
	\text{ with } 
	\rho(x) = \phi^{-1}(e^{ix}) = \cot\!\left(\frac{\theta-x}{2}\right).
\end{equation}
The derivative of $\rho$, $\rho'(x) = \frac{1}{2}\csc^2((\theta - x)/2)$, is strictly positive  everywhere it is defined, which is for $x\not\in \theta + 2\pi \Z$. 
In particular, we can always choose $\theta$ so that $\rho$ and its derivative 
are defined at any finite set $a_1, \hdots, a_n\in \R$. 

The following are straightforward from the definitions of stability: 

\begin{prop}\label{prop:TorusToReal}
	For $p\in \C[\z]_{\leq \bd}$ the following are equivalent: 
	\begin{itemize}
		\item[(a)] $p$ belongs to $\A_{\bd}$ 
		\item[(b)] for every $\ell = (\ell_1, \hdots, \ell_n) \in \R_+^n$ and $x\in \C$, 
		$p(\exp(ix\lv))=0$ implies $x\in \R$, and
		\item[(c)] for $\phi$ as in \eqref{eq:Mobius}, $\phi\cdot p$ is stable with respect to $\H_+$ and $\H_-$.   
	\end{itemize}
\end{prop}

In order to understand polynomials stable w.r.t. $\D$ and $\C\backslash \overline{\D}$, 
we first recall some useful facts about real polynomials stable w.r.t $\H_+$.

For any $w\in \R^n$ and polynomial $q = \sum_{\alpha}a_{\alpha}{\bf z}^{\alpha}$, we define the support of $q$ to be
the collection of exponents of monomials appearing in $q$, i.e.  $\supp(q) = \{ \alpha\in \Z_{\geq 0}^n : a_{\alpha}\neq 0\}$. 
For any vector $w = (w_1, \hdots, w_n)\in \R^n$, define the $w$-initial form of $q$ to be the sum over all terms  in $q$ maximizing 
$\langle w , \alpha\rangle$. That is, we can define 
\[
\deg_w(q) = \max_{\alpha \in {\rm supp}(q)} \langle w, \alpha \rangle 
\quad\text{ and }\quad 
{\rm in}_w(q) = (t^{\deg_w(q)}q (t^{-w_1}z_1, \hdots, t^{-w_n}z_n))|_{t=0} = \sum_{\alpha \in A}a_{\alpha}{\bf z}^{\alpha},  
\]
where  $A$ is the subset of  $\alpha \in \supp(p)$ maximizing $\langle w, \alpha \rangle $.

\begin{prop}\label{prop:realStable1}
	Let $q = \sum_{\alpha} a_{\alpha} {\bf z}^{\alpha} \in \C[z_1, \hdots, z_n]$ be stable w.r.t $\H_+$ with total degree $d$ and let ${\bf a}\in \R_+^n$ and ${\bf b}\in \R^n$.  
	Then 
	\begin{itemize}
		\item[(a)] for any $w\in \R^n$, ${\rm in}_w(q)$ is stable w.r.t. $\H_+$, 
		\item[(b)] for any ${\bf a}_1, \hdots, {\bf a}_m\in \R_{\ge 0}^n$, 
		$q({\bf b} + y_1{\bf a}_1 + \cdots y_m{\bf a}_m)\in \C[y_1, \hdots, y_m]$ is stable w.r.t $\H_+$,
		\item[(c)] if $q$ is homogeneous, then all its coefficients have the same phase, and
		\item[(d)]  if ${\bf b}\in \R^{n}$ is a real zero of $ q $ of multiplicity $ m $, namely $ q({\bf b})=0 $ and $ \partial^{\alpha}q({\bf b})=0 $ for all $ |\alpha|<m $, then the nonzero entries of $\{\partial^{\alpha}q({\bf b}) : |\alpha| = m\}$ 
		all have the same phase. 
			\end{itemize}
\end{prop}
\begin{proof}
	
	(a) Note that for any $t\in \R_{>0}$, the polynomial $t^{\deg_w(q)}q (t^{-w_1}z_1, \hdots, t^{-w_n}z_n)$ is stable w.r.t. $\H_+$. By Hurwitz's Theorem, the set of stable polynomials is closed in the Euclidean topology on $\C[{\bf z}]_{\leq {\bf d}}$, taking the limit as $t\to 0$ 
	shows that ${\rm in}_w(q)$ is stable w.r.t. $\H_+$. 
	
	(b) First, suppose ${\bf a}_1, \hdots, {\bf a}_m\in \R_+^n$. If ${\rm Im}(y_j)>0$ for all $j$, then  
	the imaginary part of ${\bf b}+ \sum_{i=1}^my_i{\bf a}_i$ belongs to $\R_+^n$ and so 
	$q({\bf b}+ \sum_{i=1}^my_i{\bf a}_i) \neq 0$. 
	Hurwitz's Theorem then shows that the polynomial $q({\bf b}+ \sum_{i=1}^my_i{\bf a}_i) \neq 0$ is stable for 
	any ${\bf a}_1, \hdots, {\bf a}_m\in \R_{\geq 0}^n$.	

        (c) This is the content of \cite[Theorem~6.1]{COSW}. 
	
	(d)  Let $m$ denote the multiplicity of $q$ at $\z = {\bf b}$. 
	Note that by replacing $q({\bf z})$ with $ q({\bf z} + {\bf b})$ it suffices to 
	address the case ${\bf b} = (0,\hdots, 0)$. The notation $ \alpha!=\prod_{j=1}^{n}\alpha_{j}! $ allows to write $\partial^{\alpha}q(0) = \alpha! \cdot a_{\alpha}$ and so it is enough to prove that all non zero $ a_{\alpha} $ with $ |\alpha|=m $ share the same phase. Fix $w = (-1, \hdots, -1)$. Because $\alpha! \cdot a_{\alpha}= \partial^{\alpha}q(0) = 0$ for all $|\alpha|< m$, then the $\alpha \in \supp(p)$ that maximize $\langle w, \alpha \rangle=-|\alpha| $ are those with $ |\alpha|=m $, and in particular ${\rm in}_w(q) = \sum_{|\alpha|=m}a_{\alpha}{\bf z}^{\alpha}$.
	By parts (a) and (b), this polynomial is stable and so all of its nonzero coefficients have the same phase, which proves the claim. 
\end{proof}

We translate this statement for derivatives of trigonometric polynomials of the form $F(\xv)=p(\exp{(i\xv)})$
where $p\in\A_{\bd}$. 

First, we need a technical lemma on derivatives of compositions: 
\begin{prop}[Multivariate Chain Rule]\label{prop:chain rule}
	Let $\varphi:\C\to \C$ be a meromorphic function such that $\varphi'(x)$ is nonzero wherever defined.  Consider $f({\bf x}) = g(\varphi({\bf x}))$, where $\varphi$ is applied coordinate-wise. For any ${\bf a}\in \C^n$ at which $\varphi$ is defined, 
	the multiplicity $m$ of 
	$f$ at ${\bf a}$ equals the multiplicity of $g$ at ${\bf b} = \varphi({\bf a})$
	and for any $\alpha\in \N^n$ with $|\alpha| = m$, 
	\[
	\partial^{\alpha}f({\bf a}) = \partial^{\alpha}g({\bf b}) \cdot 
	\prod_{j=1}^n\varphi'(a_j)^{\alpha_j}. 
	\]
\end{prop}
\begin{proof}
	The symbolic expansion of 
	$\partial^{\alpha}g(\varphi({\bf a}))$ using the chain rule will be a sum 
	of products of factors $\partial^{\beta}g({\bf b})$ 
	and $\rho^{(k)}(a_j)$ for some $|\beta|\leq |\alpha|$ and  $k\leq \alpha_j$. 
	The unique such term involving $\partial^{\alpha}g$
	is  $\partial^{\alpha}g({\bf b})\prod_{j=1}^n\varphi'(a_j)^{\alpha_j}$
	and all others have a factor of $\partial^{\beta}g({\bf b})$ with $|\beta|<|\alpha|$.  
	If $m$ is the multiplicity of $g$ at ${\bf b}$, then $\partial^{\beta}g({\bf b}) = 0$ for all 
	$|\beta|<m$ and $\partial^{\alpha}g({\bf b}) \neq 0$ for some $|\alpha| = m$. 
	The calculation above shows that $\partial^{\beta}f({\bf a}) = 0$ for all 
	$|\beta|<m$ and $\partial^{\alpha}f({\bf a}) \neq 0$. 
\end{proof}

\begin{prop}\label{prop: positive derivatives any degree}
	Let $p\in\A_{\bd}$ and define $F:\C^n \to \C$ by   $F(\xv)=p(\exp{(i\xv)})$. 
	If ${\bf a}\in \R^n$ is a zero of $F$ of multiplicity $m$, 
	then nonzero elements of $\{\partial^{\alpha}F({\bf a}) : |\alpha| = m\}$ have the same phase.  
\end{prop}
\begin{proof}	Let $\phi$ be a M\"obius transformation taking $\H_+$ to $\D$, as in \eqref{eq:Mobius}, with $ \theta $ such that $ e^{i\theta}\ne e^{ia_{j}} $ for all the coordinates 
	$e^{ia_1}, \hdots, e^{ia_n}$ of $\exp(i{\bf a})$.  
	By Proposition~\ref{prop:TorusToReal}, $q({\bf z}) = \phi\cdot p(\z)$ 
	is stable w.r.t. $\H_+$ and $\H_-$.
	Then 
	\[
	p({\bf z})= \phi^{-1}\cdot q({\bf z}) =  r({\bf z})\cdot q(\phi^{-1}({\z})) \ \ 
	\text{ and } 
	F({\bf x}) = p(\exp(i{\bf x})) = r(\exp(i{\bf x}))\cdot q\left(\rho({\bf x})\right),
	\]
	where $ \rho(\xv)=(\cot(\frac{a_{1}-x_{1}}{2}),\ldots,\cot(\frac{a_{n}-x_{n}}{2})) $ and $r(\exp(i\xv)) = \prod_{j=1}^n(e^{i x_j} - e^{ia_{j}})^{d_j}$. In particular $ r(\exp(i{\bf a}))\ne0 $, so $ q(\rho(\textbf{a}))=0 $, hence $ \rho(\textbf{a}) $ is a zero of $ q $. An induction argument shows that ${\bf a}$ must be a zero of $q(\rho({\bf x}))$ of multiplicity $m$. That is, $\partial^{\alpha}q(\rho({\bf x}))$ is zero at ${\bf x} = {\bf a}$ for all $|\alpha|<m$ and nonzero for some $|\alpha| = m$. To do so, suppose that $\partial^{\alpha}q(\rho({\bf x}))|_{\xv=\textbf{a}}=0$ for all $|\alpha|\le m'-1$ for $ m'<m $. Then, for any $\alpha\in \Z_{\ge0}^n$ with $ |\alpha|=m' $, 
	\begin{equation}\label{eq:prodRule}
		\partial^{\alpha} F({\bf x})=
		\sum_{\beta + \gamma =\alpha}\partial^{\beta}r(\exp(i\xv))\partial^{\gamma}q(\rho({\bf x})), \end{equation}
	so 
	\[0=\partial^{\alpha} F({\bf x})|_{\xv=\textbf{a}}=r(\exp(i\textbf{a}))\partial^{\alpha}q(\rho({\bf x}))|_{\xv=\textbf{a}}.\]
	Since $ r(\exp(i\textbf{a}))\ne 0 $, then $ \partial^{\alpha}q(\rho({\bf x}))|_{\xv=\textbf{a}}=0 $ for every $ |\alpha|=m' $, and by induction for any $ |\alpha|<m $. Together with Proposition~\ref{prop:chain rule}, this gives that 
	\[
	\partial^{\alpha} F({\bf a}) 
	= r(\exp(i{\bf a}))\partial^{\alpha}q(\rho({\bf x}))|_{{\bf x} = {\bf a}} = 
	r(\exp(i{\bf a})) (\partial^{\alpha}q)|_{{\bf z} = \rho({\bf a})} \cdot 
	\prod_{j=1}^n\rho'(a_j)^{\alpha_j}. 
	\]
	The phase of the non-zero factor $ r(\exp(i{\bf a}))\prod_{j=1}^n\rho'(a_j)^{\alpha_j} $ is independent of $ \alpha $, since $ \rho'(a_j) $ is positive for all $ j $, so the nonzero elements of $\{\partial^{\alpha}F({\bf a}): |\alpha| = m\}$ have the same phase because the nonzero 
	elements of $\{\partial^{\alpha}q|_{{\bf z} = \rho({\bf a})}: |\alpha| = m\}$ have the same phase, by Proposition~\ref{prop:realStable1}(d). 
\end{proof}	
\begin{lem}\label{lem: all multiplicities agree}
	For $t\in \R$, $\ell\in \R_{\ge 0}^n$ and $p\in \A_{\bd}(n)$, the following coincide: 
	\begin{itemize}
		\item[(a)] the multiplicity of $t\in \R$ as a zero of the function $f(t) = p(\exp(it\lv))$,
		\item[(b)] the multiplicity of $\xv=t\lv$ as a zero of $F({\bf x}) = p(\exp(i\xv))$, 
		\item[(c)] the multiplicity of $\z=\exp{(it\ell)}\in \T^n$ as a zero of $p(\z)$, and 
		\item[(d)] the multiplicity of $1$ as a root of the univariate polynomial $q(s)=p(s\exp(it\ell))$.
	\end{itemize}
\end{lem}
\begin{proof} 
	Note that by replacing $p({\bf z})$ with $p(e^{it\ell_1}z_1, \hdots,e^{it\ell_n}z_n)$
	it suffices to consider $t=0$ for this equivalence. 
	
	(a)$=$(b) Let $D_{\ell}$ denote the differential operator $\sum_{j=1}^n \ell_j \frac{\partial}{\partial x_{j}}$. Then for any $m\in \N$, $D_{\ell}^m = \sum_{|\alpha| = m} \binom{m}{\alpha}\ell^{\alpha}\partial^{\alpha}$, where $\binom{m}{\alpha} = \frac{m!}{\alpha_1!\cdots \alpha_n!}$, and 
	\[
	f^{(m)}(0) = D_{\ell}^m F|_{{\bf x} = (0,\hdots,0)} = \sum_{|\alpha| = m} 
	\binom{m}{\alpha}\ell^{\alpha}\partial^{\alpha}F(0).
	\]
We see that the multiplicity of $(0,\hdots, 0)$ as a zero of $F({\bf x}) = p(\exp(i\xv))$ lower bounds on the multiplicity of  $0$ as a zero of the function $f$. 
	Moreover, by Proposition~\ref{prop: positive derivatives any degree}, when $m$ 
	is the multiplicity of  $(0,\hdots, 0)$ as a zero of $F({\bf x})$, 
	the nonzero values of $\{ \partial^{\alpha}F(0) : |\alpha| =m\}$ have the same phase. 
	By assumption, at least one of these is nonzero, ensuring that their sum, $f^{(m)}(0)$, is non-zero and that $f$ has multiplicity $m$ at $t=0$. 
	
	(b)$=$(c) Follows from Proposition~\ref{prop:chain rule} with $\varphi(x) = e^{ix}$. 
	
	(b)$=$(d) Consider $q(s) = p(s,s,\ldots,s)$ 
	and $h(t) = q(e^{i t}) = p(e^{it},\ldots,e^{i t}) = F(t,t,\ldots,t) $.
	
	By Proposition~\ref{prop:chain rule}, the multiplicity of $q$ at $s=1$ 
	equals the multiplicity of $h$ at $t=0$. By the equivalence 
	(a)$=$(b) with $\ell = {\bf 1}$, this equals the multiplicity 
	of $F({\bf x})$ at ${\bf x} = {\bf 0}$.
\end{proof}

\subsection{Connectivity of $\A_{\bd}$ and perturbations}
One of the early results on real hyperbolic polynomials is Nuij's result 
that space of hyperbolic polynomials of a given degree is simply connected \cite{Nuij}.  
Here we adapt these techniques to better understand $\A_{\bd}$.

Nuij's proof relies on the following operators on univariate polynomial that 
preserve real rootedness.  For $\lambda\in \R_{>0}$, define $\mathcal{D}_{\lambda}:\C[z]\to \C[z]$ by $\mathcal{D}_{\lambda}(q) = q + \lambda q'$. Nuij shows that 
if $q$ is real rooted then $\mathcal{D}_{\lambda}(q)$ is real rooted, $\mathcal{D}_{\lambda}$ decreases the multiplicity of roots of $q$ by $1$ for $\lambda\neq 0$, and all new roots of $\mathcal{D}_{\lambda}(q)$ are simple.  In particular, for any real rooted polynomial $q\in \R[z]$ of degree $d$, the multiplicity of any root of $q$ is at most $d$ and so applying $\mathcal{D}_{\lambda}$ $d$ times to $q$ results in a real rooted univariate polynomial with $d$ simple roots.  The roots of $\mathcal{D}_{\lambda}(q)$ interlace those of $q$ in the following sense: if $a_1, \hdots, a_d$ are the roots of $q$ and $b_1, \hdots, b_d$ are the roots of $\mathcal{D}_{\lambda}(q)$, then $b_j\leq a_j\leq b_{j+1}$ for all $j$.

Let $\C[{\bf y}, {\bf z}]_{\bf d}$ denote the set of polynomials in $y_1, \hdots, y_n$ and $z_1, \hdots, z_n$ that are homogeneous of degree $d_j$ in each set of variables $(y_j, z_j)$.
The zero-set of such polynomials are well-defined subsets of $(\P^1(\C))^n$. 
Here we use $\P^1(K)$ to denote the projective line over a field $K$, which is 
$K^2\backslash \{(0,0)\}$ modulo the equivalence $(a,b) \sim (\lambda a, \lambda b)$ for $\lambda \neq 0$. 
For any polynomial $p\in \C[{\bf y}, {\bf z}]_{\bf d}$ and $\lambda \in (\C^*)^n$, 
$p(\lambda_1y_1, \hdots, \lambda_n y_n, \lambda_1z_1, \hdots, \lambda_nz_n) = \lambda^{\bf d}p(y_1, \hdots, y_n, z_1, \hdots, z_n)$. 
In a slight abuse of notation, we will use $[{\bf a}:{\bf b}]$ to denote a point 
$([a_i:b_i])_{i\in [n]} \in (\P^1(K))^n$, where ${\bf a} = (a_1, \hdots, a_n)$ and ${\bf b} = (b_1, \hdots, b_n)$. Similarly, for a subset $I\subseteq [n]$, 
we use $[{\bf a}_I:{\bf b}_I]$ to denote the point $([a_i:b_i])_{i\in I} \in (\P^1(K))^I$. 

To understand the zero set of $p$ on $(\P^1(K))^n$ we restrict to various affine charts. 
We can partition points $[{\bf a}:{\bf b}] \in (\P^1(K))^n$ by the set $I = \{i\in [n]:a_i\neq 0\}$. 
The affine chart of points $[{\bf a}:{\bf b}] \in (\P^1(K))^n$ with $a_i\neq 0$ for all $i$ is 
isomorphic to $K^n$ via  the coordinate-wise correspondence $[a_i:b_i] \leftrightarrow b_i/a_i$. For $i\in I$, $b_i\neq 0$ and for $j\not\in I$, $a_j\neq 0$, 
and so after rescaling we may take $b_i=1$ and $a_j = 1$.

On this vector space of polynomials define the linear operator \[
\mathcal{D}_{\lambda}:\C[{\bf y},{\bf z}]_{\bf d}\to \C[ {\bf y},{\bf z}]_{\bf d} \ \ \text{ by } \ \ 
\mathcal{D}_{\lambda}(q) = q + \lambda \sum_{j=1}^ny_j\partial_{z_j}q. 
\]
Let $\mathcal{D}^{|\bd|}_{\lambda}$ denote the operator obtained from $\mathcal{D}_{\lambda}$ by applying it $|{\bf d}| = \sum_{j=1}^nd_j$ times.

For each ${\bf d}\in \Z_{\geq 0}^n$, consider the following sets of polynomials: 
\begin{align*}
	\mathfrak{S}_{\bf d} & = 
	\{q\in \C[{\bf y},{\bf z}]_{\bf d} : {\rm coeff}(q, {\bf z}^{\bf d})=1 \text{ and $q({\bf 1},{\bf z})$ is 
		stable w.r.t. $\H_{+}$ and $\H_{-}$}\}\\
	\mathfrak{S}^{\circ}_{\bf d} & = 
	\{q\in \mathfrak{S}_{\bf d}  :  q \text{ and }\nabla q\ \text{have no common zeros in $(\P^1(\R))^n$}\}.
\end{align*}

\begin{prop}$\mathfrak{S}_{\bf d} \subseteq \R[{\bf y},{\bf z}]_{\bf d}$. 
\end{prop}
\begin{proof}
One can check directly from the definition that the stability of $q({\bf 1},{\bf z})$ implies that for all  
${\bf a}\in \R^n$, the polynomial $q({\bf 1}, {\bf a} + t\textbf{1})\in \C[t]$ has real roots, say $ r_{j}\in\R $ for $ j=1,\ldots,|\bd| $. 
By assumption, the coefficient of $\z^{\bf d}$ in $q$ is $1$, which implies that the coefficient of 
$t^{|{\bf d}|}$ in $q({\bf 1}, {\bf a} + t\textbf{1})$ is $1$, so $ q({\bf 1}, {\bf a} + t\textbf{1})=\prod_{j=1}^{|\bd|}(t-r_{j}) $ and therefore all of its coefficients must be real. If $q = g+ih$ where $g,h\in \R[{\bf y}, {\bf z}]_{\bf d}$, then 
we have shown that $h({\bf 1}, {\bf a} + t\textbf{1})\in \R[t]$ is the zero polynomial for all ${\bf a}\in \R^n$. In particular, $h({\bf 1}, {\bf a})=0 $ for all ${\bf a}\in \R^n$, which implies that $h$ is identically zero. 
 \end{proof}

\begin{prop}\label{prop:interlacer}
If $q\in \C[{\bf z}]_{\leq d}$ is stable with respect to $\H_+$ and $\H_-$, then so
is the polynomial $q + z_{n+1}\sum_{j=1}^n\partial_{z_j}q$ in $\C[z_1, \hdots, z_{n+1}]$.  
Moreover, for any ${\bf a}\in \R^n$ and ${\bf b}\in \R_{\geq 0}^n$, the roots of 
$q + \sum_{j=1}^n\partial_{z_j}q$ interlace those of $q$ when restricted to ${\bf z} = {\bf a} + t {\bf b}$. 
That is, $\lambda_1\leq \mu_1\leq \lambda_2\leq \hdots \leq \lambda_{|{\bf d}|}\leq \mu_{|{\bf d}|}$, where $\{\lambda_j\}_j$ and  $\{\mu_j\}_j$ are the roots of the restrictions of $q + \sum_{j=1}^n\partial_{z_j}q$ and $q$, respectively. 
\end{prop}

\begin{proof}
We use the theory of stability preservers by Borcea and Br\"and\'en, \cite[Theorem 1.3]{BB1}. 
The \emph{symbol} of the operator  $\mathcal{D}(q) =  q + z_{n+1}\sum_{j=1}^n\partial_{z_j}q$ is 

\[
\mathcal{D}((\z+\textbf{w})^{\bf d}) =(\z+\textbf{w})^{\bf d}\left(1+ z_{n+1}\sum_{j=1}^n d_{j}(z_{j}+w_{j})^{-1}\right).
\]
We can see by inspection that this polynomial is stable. If ${\rm Im}(z_j)>0$ and ${\rm Im}(w_j)>0$ for all $j$, then ${\rm Im}(-(z_{n+1})^{-1})>0$ and
${\rm Im}((z_j+w_j)^{-1})<0$ for all $ j $, so $\sum_{j=1}^n d_{j}(z_{j}+w_{j})^{-1}\ne -(z_{n+1})^{-1}$, and therefore $ \mathcal{D}((\z+\textbf{w})^{\bf d})\ne0 $. This shows that 
the symbol $\mathcal{D}(({\bf z}+{\bf w})^{\bf d}) \in \R[z_1, \hdots, z_n, z_{n+1}, w_1, \hdots, w_n]$ is stable with respect to $\H_+$ and by the same argument it is also stable with respect to $\H_-$. Then, 
by \cite[Theorem 1.3]{BB1}, the linear operation $\mathcal{D}$ preserves stability. 

The statement of interlacing then follows from \cite[Lemma 1.8]{BB1}.
\end{proof}

\begin{lem} \label{lem:nonsingRestrict}
Let $q\in \mathfrak{S}_{\bf d}$, $I\subseteq [n]$, and let $q_I$ denote 
restriction of $q$ to $y_j=0$ and $z_j=1$ for all $j\not\in I$. Then $q_I\in \mathfrak{S}_{{\bf d}_I}$, i.e. $q_I$ is nonzero and $q_I(\textbf{1}_I, {\bf z}_I)$ is stable w.r.t. $\mathcal{H}_+$ and $\mathcal{H}_-$.
If additionally $q\in \mathfrak{S}_{\bf d}^{\circ}$, then $q_I\in \mathfrak{S}_{{\bf d}_I}^{\circ}$, i.e. $q_I$ and $\nabla q_I$ have no common zeros in $(\P^1(\R))^I$. 
\end{lem}
\begin{proof}
Note that $1 = {\rm coeff}(q,{\bf z}^{\bf d})  = q({\bf 0}, {\bf 1}) = q_I({\bf 0}_I, {\bf 1}_I)$, showing that $q_I$ is nonzero and has $1 = {\rm coeff}(q_{I},\prod_{i\in I} z_i^{d_i})$. Note that $\prod_{j\not\in I}z_j^{d_j} \cdot q_I({\bf 1}_I, {\bf z}_I)$ 
is the initial form of $q(1,{\bf z})$ with respect to the vector $({\bf 0}_I, {\bf 1}_{[n]\backslash I})$ and so $q_I({\bf 1}_I, {\bf z}_I)$ is stable by \Cref{prop:realStable1} (a). 

Suppose that $q_I$ is zero at a point $[{\bf a}_I:{\bf b}_I]\in (\P^1(\R))^I$. We will show that for some $i\in I$, 
$\partial_{y_i}q$ or $\partial_{z_i}q_I$ is nonzero at this point. Note that if $a_k = 0$ for some $k\in I$, then we can replace $I$ with $I' = \{i\in I : a_i \neq 0\}$, which is non-empty by assumption. If there is some $i\in I'$ for which $\partial_{y_i}q_{I'}$ or $\partial_{z_i}q_{I'}$ is non-zero at $[{\bf a}_{I'}:{\bf b}_{I'}]$, this proves the claim. Therefore we may assume that for all $i\in I$, $a_i\neq 0$ and take $a_i = 1$.  
Moreover, by replacing $q({\bf y}, {\bf z})$ with its substitution of $z_i\mapsto z_i + b_i y_i$ for all $i \in I$, 
we can assume that $b_i = 0$ for all $i\in I$. 

Since $q\in \mathfrak{S}_{\bf d}^{\circ}$, 
there is some derivative $\partial_{y_j}q$ or $\partial_{z_j}q$ that is nonzero at $[{\bf a}:{\bf b}]$ where ${\bf a} = ({\bf 1}_I, {\bf 0}_{[n]\backslash I})$ 
and ${\bf b} = ({\bf 0}_I, {\bf 1}_{[n]\backslash I})$. If $j\in I$, we are done so take $j\not\in I$ and assume by contradiction that all derivatives with respect to variables labeled by $ I $ are zero.  
Since $q$ is homogeneous of degree $d_j$ in $(y_j, z_j)$, then 
$ y_j \partial_{y_j}q+z_j \partial_{z_j}q=d_j q$. Since $q$ and $y_j$ both vanish at this point and $z_j$ does not, we see that it must be $\partial_{y_j}q$ that 
is nonzero at $[{\bf a}: {\bf b}]$. 

Consider the polynomial $\tilde{q}(s,t) = q({\bf a} - te_j,   {\bf b}+s{\bf 1}_{I})\in \R[s,t]$. We claim that this polynomial is stable. 
To see this, note the upper halfplane is invariant under $\varphi(z)=-1/z$. Let $ \varphi_{I}(\z) $ be the vector with $ i $-th entry $ \varphi(z_{i}) $ if $ i\notin I $ or $ z_{i} $ otherwise, so that $\prod_{j\not\in I}z_{j}^{d_j}\cdot p(\varphi_{I}(\z)) $ is stable with respect to upper and lower halfplanes if and only if $ p(\z) $ is. 
The polynomial 
\[\prod_{j\not\in I}(-z_j)^{d_j}\cdot q({\bf 1}, \varphi_{I}(\z)) =
q(({\bf 1}_I, -{\bf z}_{[n]\backslash I}), ({\bf z}_I, {\bf 1}_{[n]\backslash I})) \in \R[\z]\] 
is stable. We then obtain the polynomial $\tilde{q}$ as a further restriction of $z_j = t$, $z_k=0$ for $k\in [n]\backslash (I \cup\{j\})$ and $z_i = s$ for all $i\in I$. It follows that $\tilde{q}$ is stable by \Cref{prop:realStable1} (b). 

Note that $\tilde{q}(0,0) = q({\bf a}, {\bf b}) = 0$. Moreover, we have that 
\[
\partial_s\tilde{q}|_{(s,t) = (0,0)} = \left(\sum_{i\in I}\partial_{z_i}q\right)|_{({\bf y}, {\bf z}) = ({\bf a}, {\bf b})} = 0 \text{ and } 
 \partial_t\tilde{q}|_{(s,t) = (0,0)} = \left(-\partial_{y_j}q\right)|_{({\bf y}, {\bf z}) = ({\bf a}, {\bf b})} \neq  0.
 \]
The polynomial $\tilde{q}(s,0) = q({\bf a},  {\bf b}+s{\bf 1}_{I})$ has leading term $s^{\sum_{i\in I}d_i}$, since $ q(\textbf{0},\textbf{1})=1 $, and so is nonzero.
Let $k$ be the smallest integer for which $\partial_s^k\tilde{q}|_{(s,t) = (0,0)}$ is nonzero. 
By the arguments above, $k$ exists and $k\geq 2$. This means that all monomials $ s^{\alpha}t^{\beta} $ appearing in $ \tilde{q} $ with non-zero coefficients either have $ \beta\ge 1 $ or $ \alpha\ge k\ge 2 $. In particular $ \langle (-1,-k),(\alpha,\beta)\rangle\le -k $, with equality if and only if $ (\alpha,\beta)=(k,0) $ or $ (\alpha,\beta)=(0,1) $. We conclude that the initial form ${\rm in}_{(-1,-k)}\tilde{q}=as^k + bt$ for some non-zero coefficients $a,b\in \R^*$. By \Cref{prop:realStable1} (a), it is stable with respect to both upper and lower halfplanes. However, since $k\geq 2$, there is some $ c\in\mathcal{H}_{+} $ such that $c^{k}= -\frac{b}{a}i $, and so $ (s,t)=(c,i)\in\mathcal{H}_{+}^{2} $ is a root, contradicting stability. 
\end{proof}

\begin{prop}\label{prop:perturbStable}
 For any $q\in\mathfrak{S}_{\bf d}$ and $\lambda >0$, $\mathcal{D}_{\lambda}(q)\in \mathfrak{S}_{\bf d}$ and 
 $\mathcal{D}^{|{\bf d}|}_{\lambda}(q) \in \mathfrak{S}^{\circ}_{\bf d}$. \end{prop} 
\begin{proof}
By Proposition~\ref{prop:interlacer}, the operation $q\mapsto q + \lambda \sum_{i=1}^n \partial_{z_i}q$ preserves stability of polynomials in $\R[\z]$. We need to show that  $\mathcal{D}^{|{\bf d}|}_{\lambda}q $ has no common zeros with its gradient on $(\P^1(\R))^n$. By the univariate case discussed above, $\mathcal{D}^{|{\bf d}|}_{\lambda}q({\bf 1}, {\bf b} + t{\bf 1}) \in \R[t]$ has simple roots for all ${\bf b}\in \R^n$. It follows that if $\mathcal{D}^{|{\bf d}|}_{\lambda}q({\bf 1}, {\bf z})$ vanishes at $ \z\in\R^{n} $ then its gradient does not. 
This shows that $\mathcal{D}^{|{\bf d}|}_{\lambda}q({\bf 1}, {\bf z})$ and its gradient 
have no common zeros of the form $[{\bf a}:{\bf b}]$ where $a_j\neq 0$ for all $j$. Assume by contradiction that $\mathcal{D}^{|{\bf d}|}_{\lambda}q({\bf 1}, {\bf z})$ and its gradient 
have some common zero $[{\bf a}:{\bf b}]\in (\P^1(\R))^n$ and let $I = \{i : a_i \neq 0\}$, so $ I\ne[n] $.   
Note that we can assume $b_j=1$ for all $j\not\in I$. If $I = \emptyset$, then $[{\bf a}:{\bf b}] = [{\bf 0}: {\bf 1}]$, at which $\mathcal{D}^{|{\bf d}|}_{\lambda}q({\bf 0}, {\bf 1})={\rm coeff}(\mathcal{D}^{|{\bf d}|}_{\lambda}q,{\bf z}^{\bf d})\ne0 $. Therefore, $\emptyset \subsetneq I \subsetneq [n]$. 

Let $q_I\in \R[y_i, z_i :i\in  I]$ denote the restriction of $q$ to $y_j = 0$ and $z_j=1$ for $j\not\in I$. Note that the operator $\mathcal{D}_{\lambda}$ commutes with the 
	restriction to $y_j = 0$ and $z_j=1$. That is, 
	\[
	(\mathcal{D}_{\lambda}q)|_{\{y_j = 0, z_j=1: j\notin I\}}  \ =  \ \left(q + \lambda \sum_{i=1}^{n}y_i\partial_{z_i}q\right)|_{\{y_j = 0, z_j=1: j\notin I\}}
	\ = \  q_I + \lambda \sum_{i\in I}y_i\partial_{z_i}q_I \ = \ \mathcal{D}_{\lambda}q_I.
	\]
In particular, $ \mathcal{D}_{\lambda}q_I(\textbf{a}_{I},\textbf{1}_{I})=\mathcal{D}_{\lambda}q_I(\textbf{1}_{I},\textbf{a}_{I}^{-1})=0 $. Since $ q_I\in\mathfrak{S}_{\bf d_I} $, by \Cref{lem:nonsingRestrict}, and it has total degree $|\bd_I|\le |\bd| $, then the argument above shows that $ \mathcal{D}_{\lambda}q_I\in\mathfrak{S}_{\bf d_I} $ and that the gradient of $\mathcal{D}_{\lambda}^{|{\bf d}|}(q_I)(1,\z_{I})$ cannot vanish at the zero $ (\textbf{1}_{I},\textbf{a}_{I}^{-1}) $. Hence, there must be some nonzero derivative of $\mathcal{D}_{\lambda}^{|{\bf d}|}(q_I)$ at $[\textbf{1}_{I}:\textbf{a}_{I}^{-1}]=[{\bf a}_I:{\bf b}_I]$, which gives a nonzero derivative of $\mathcal{D}_{\lambda}^{|{\bf d}|}(q)$ at $[{\bf a}:{\bf b}]$.
\end{proof}

\begin{prop} \label{prop:Nuij}
	Both $\mathfrak{S}_{\bf d}$ and $\mathfrak{S}^{\circ}_{\bf d}$
	are contractible and $\mathfrak{S}_{\bf d}$ equals the closure (in the Euclidean topology on $\R[{\bf y},{\bf z}]_{\bf d})$ of $\mathfrak{S}^{\circ}_{\bf d}(\R)$. 
\end{prop} 
\begin{proof}The proof follows the proof of the main theorem in \cite{Nuij}. 
	For $\mu\in \R$, consider the linear operator $G_{\mu}$ on $\R[{\bf y},{\bf z}]_{\bf d}$
	defined by $G_{\mu}(q) = q( {\mu}{\bf y},{\bf z})$. 
	This operator preserves both stability and the coefficient of ${\bf z}^{\bf d}$. 
	
	For $\lambda \in [0,1]$, consider the map $\mathcal{D}_{1-\lambda}^{|{\bf d}|}G_{\lambda}$. 
	This map preserves stability and, for $\lambda \neq 1$, the image of $\mathfrak{S}_{\bf d}$ 
	under this map belongs to $\mathfrak{S}^{\circ}_{\bf d}$. 
	For $\lambda = 1$, we get the identity map $\mathcal{D}_{0}^{|{\bf d}|}G_{1}(q) = q$ and 
	for $\lambda = 0$ we get $\mathcal{D}_{1}^{|{\bf d}|}({\bf z}^{\bf d})\in \mathfrak{S}^{\circ}_{\bf d}(\R)$. 
	Therefore this gives a deformation retraction of both 
	$\mathfrak{S}_{\bf d}$ and $\mathfrak{S}^{\circ}_{\bf d}$ onto the point  $\mathcal{D}_{1}^{|{\bf d}|}({\bf z}^{\bf d})$. 
\end{proof}

\begin{prop}\label{prop:stableInterior}
The interior of $\mathfrak{S}_{\bf d}$ in $\{q\in \R[{\bf y}, {\bf z}]_{\bf d} : {\rm coeff}(q,{\bf z}^{\bf d}) = 1\}$ is not empty, and in particular, it contains $\mathfrak{S}_{\bf d}^{\circ}$.
\end{prop}
\begin{proof}
Suppose that $q\in \mathfrak{S}_{\bf d}^{\circ}$, so that $ q $ and its gradient have no common zeros in $(\P^1(\R))^n$. Let $(S^1)^n=\{({\bf y}, {\bf z})\in \R^{2n} : y_j^2+z_j^2 = 1 \forall j\}$ and let $ V=\set{(\yv,\z)\in(S^1)^n :  \mathcal{D}_1q(\yv,\z)=0} $, where $ \mathcal{D}_{1}q $ is $ \mathcal{D}_{\lambda} q $ at $ \lambda=1 $. Consider the set of polynomials 
\[
U = \{g\in \R[{\bf y}, {\bf z}]_{{\bf d}}: {\rm coeff}(g, {\bf z}^{\bf d}) = 1, g(\yv,\z)q(\yv,\z)>0 \text{ for all } (\yv,\z)\in V\}.
\] 
The set $ V $ is compact, since $(S^1)^n$ is compact, and so $ \min_{(\yv,\z)\in V}g(\yv,\z)q(\yv,\z) $ is continuous in the coefficients of $ g $, which means that $U$ is open in $\{g\in \R[{\bf y}, {\bf z}]_{\bf d} : {\rm coeff}(g,{\bf z}^{\bf d}) = 1\}$. 
We claim that $q\in U$ and $U\subseteq \mathfrak{S}_{\bf d}^{\circ}$. 

To see that $q\in U$, it suffices to show that $q$ and $\mathcal{D}_1q$ have no common zeros in $(\P^1(\R))^n$. 
We first check this for the points in the affine chart ${\bf y} = {\bf 1}$.
Suppose that $q({\bf 1}, {\bf b})=0$ for $[{\bf 1}: {\bf b}]\in(\P^{1}(\R))^{n} $, so by assumption, there is some $j$ for which $\partial_{z_j}q({\bf 1}, {\bf b})$ is nonzero. By Proposition~\ref{prop:realStable1}, all of the the nonzero entries of $\{\partial_{z_i}q({\bf 1}, {\bf b}): i=1, \hdots, n\}$ have the same phase, which implies that $\sum_{i=1}^n \partial_{z_i}q({\bf 1}, {\bf b})$ is nonzero. 
Since $q({\bf 1}, {\bf b})=0$, it follows that $\mathcal{D}_1q = q + \sum_{i=1}^n \partial_{z_i}q$ is nonzero at $[{\bf 1}: {\bf b}]$.

For any arbitrary point $[{\bf a}: {\bf b}]\in (\P^1(\R))^n$, let $I = \{i\in [n]: a_i\neq 0\}$, which by assumption is non-empty. 
By \Cref{lem:nonsingRestrict}, $q_I$ is stable and has no common zeros with its gradient on  $(\P^1(\R))^I$. The argument above shows that $q_I$ and $\mathcal{D}_1q_I$ cannot both be zero at $[{\bf a}_{I}: {\bf b}_{I}]$, and 
so $q$ and $\mathcal{D}_1q$ cannot both be zero at $[{\bf a}:{\bf b}]$.

To see that $U\subseteq \mathfrak{S}_{\bf d}^{\circ}$, let ${\bf a}\in \R^n$, ${\bf b}\in \R_+^n$ and let $\{\lambda_j\}_j$ denote the roots of $\mathcal{D}_1q({\bf 1}, {\bf a}+{\bf b}t)$ and $\{\mu_j\}_j$ denote the roots of $q({\bf 1}, {\bf a}+{\bf b}t)$. These roots are distinct by the argument above. 

By Proposition~\ref{prop:interlacer},  $\lambda_1< \mu_1< \lambda_2< \hdots < \lambda_{|{\bf d}|}< \mu_{|{\bf d}|}$.  In particular, $q$ must alternate signs on the roots of  $\mathcal{D}_1q({\bf 1}, {\bf a}+{\bf b}t)$. 
If $g\in U$, then $g({\bf 1}, {\bf a}+{\bf b}t)\in \R[t]$ has degree $ |\bd| $ with positive leading coefficient $\textbf{b}^{\bd}$, and it alternates signs on the roots of $\mathcal{D}_1q({\bf 1}, {\bf a}+{\bf b}t)$. Hence it has $|{\bf d}|$ distinct real roots. As this holds for any ${\bf a}\in \R^n$, ${\bf b}\in \R_+^n$, then $ g\in\mathfrak{S}_{\bf d}^{\circ} $. See, for example, \cite[2.3,2.4]{Wagner}.
\end{proof}
We can modify this using the M\"obius transformations $\phi$ from 
\eqref{eq:Mobius} to translate these results to $\A_{\bd}$. 
For any $\xv\in [0,2\pi)^n$ define 
\[
\A_{\bd}(\xv)  = 
\{p\in \A_{\bd} : p(\exp({i{\xv}}))\neq 0\}  
\ \ \text{ and }\]
 \[\A_{\bd}^{\circ}({\xv})  = 
\{p\in\A_{\bd}({\xv})  : p \text{ and }\nabla p\ \text{have no common zeros in } \T^{n}\}.
\]
One can check that $q\in \C[{\bf y}, {\bf z}]_{\bf d}$ belongs to $\mathfrak{S}_{\bf d}$ (respectively $\mathfrak{S}_{\bf d}^{\circ}$) if and only if $\phi^{-1}\cdot q({\bf 1},{\bf z})$ 
belongs to $\A_{\bd}(\xv)$ (respectively $\A_{\bf d}^{\circ}(\xv)$) for $ \phi $ defined using the angles $ e^{ix_{1}},\ldots,e^{ix_{n}} $. 

\begin{Def}\label{def: inv}
	Define an involution of polynomials in $\C[ {\bf z}]_{\leq {\bf d}}$ by
	\begin{align*}
		p(z_{1},\ldots,z_{n})\mapsto & p^{\dagger}(z_{1},\ldots,z_{n})  = \z^{\bd}\overline{p\left(\bar{z_{1}}^{-1},\ldots,\bar{z_{n}}^{-1}\right)}, \ \text{ namely}\\
	\sum_{\alpha} a_{\alpha} {\bf z}^{\alpha} \mapsto & \sum_{\alpha} {\overline{a_{\alpha}}} {\bf z}^{{\bf d} - \alpha}  =\sum_{\alpha} {\overline{a_{{\bf d} - \alpha}}} {\bf z}^{\alpha},		
	\end{align*}  
and define the set 
of polynomials in $\C[ {\bf z}]_{\leq {\bf d}}$ that are invariant under the involution
\[\C[ {\bf z}]_{\leq {\bf d}}^{\rm in}:=\set{p\in\C[ {\bf z}]_{\leq {\bf d}}\ :\ p=p\inv}=\left\{\sum_{\alpha\le\bd} a_{\alpha} {\bf z}^{\alpha}\ :\ a_{\alpha} = {\overline{a_{{\bf d} - \alpha}}}\ \text{ for all }\alpha\right\},\]
and the set of polynomials for which $ p\inv$ is a scalar multiple of $p $ by
\[\C\cdot \C[ {\bf z}]_{\leq {\bf d}}^{\rm in}:=\set{cp\ :\ c\in\C,p\in\C[ {\bf z}]_{\leq {\bf d}}^{\rm in}}=\set{p\in\C[ {\bf z}]_{\leq {\bf d}}\ :\ p\inv=cp\ \text{ for some $ c $ with } |c|=1 }.\]
	\end{Def}
The next lemma is straight forward.
\begin{lem}
	The set $ \C[ {\bf z}]_{\leq {\bf d}}^{\rm in} $ is a real vectorspace of dimension $ \prod_{j=1}^n(d_j+1) $, spanned by $({\bf z}^{\alpha} + {\bf z}^{{\bf d}-\alpha})$ and $i({\bf z}^{\alpha} - {\bf z}^{{\bf d}-\alpha})$ for $\alpha\leq {\bf d}$. The set $ \C\cdot \C[ {\bf z}]_{\leq {\bf d}}^{\rm in} $ is a semialgebraic set of dimension $1+ \dim(\C[ {\bf z}]_{\leq {\bf d}}^{\rm in})$ in the $\left(2\prod_{j=1}^n(d_j+1)\right)$-dimensional real vectorspace $\C[{\bf z}]_{\leq {\bf d}}$, from which it inherits the Euclidean topology.
	\end{lem}
\begin{remark}
	Note that from the polynomial $q = cp$ with $c= e^{ix}$ for $x\in [0,\pi)$ and 
	$p\in \C[ {\bf z}]_{\leq {\bf d}}^{\rm in}$, 
	we can recover $c/\overline{c} = e^{2ix}$ as $q({\bf z})/q^{\dagger}({\bf z})$ and $c = (c/\overline{c})^{1/2} =  e^{ix}$. 
\end{remark}
The image of $\C\cdot \R[{\bf y}, {\bf z}]_{\bf d}$ under the map $q\mapsto \phi^{-1}\cdot q({\bf 1},{\bf z})$ 
coincides with $\C\cdot\C[ {\bf z}]_{\leq {\bf d}}^{\rm in}$:
\[\C\cdot \C[ {\bf z}]_{\leq {\bf d}}^{\rm in}=\left\{c \sum_{\alpha\le\bd}a_{\alpha}(-i)^{|\alpha|}(\z+\exp(i{\xv}))^{\alpha}(\z-\exp(i{\xv}))^{\bd-\alpha}\ :\ a_{\alpha}\in\R, c\in \C\right\}.\]
Note that for $p({\bf z}) = \sum_{\alpha\le\bd}a_{\alpha}(-i)^{|\alpha|}(\z+\exp(i{\xv}))^{\alpha}(\z-\exp(i{\xv}))^{\bd-\alpha}$ with $a_{\alpha}\in \R$, using the notation $ \z^{-1}=(\frac{1}{z_{1}},\ldots,\frac{1}{z_{n}}) $, we have 
\begin{align*}
p\inv(\z) & = {\bf z}^{\bd}\sum_{\alpha\le\bd}a_{\alpha}(i)^{|\alpha|}(\z^{-1}+\exp(-i{\xv}))^{\alpha}(\z^{-1}-\exp(-i{\xv}))^{\bd-\alpha}\\
& = (\exp(-i{\xv}))^{\bd}\sum_{\alpha\le\bd}a_{\alpha}(i)^{|\alpha|}(\exp(i{\xv})+\z)^{\alpha}(\exp(i{\xv})-\z)^{\bd-\alpha}=  (-\exp(-i{\xv}))^{\bd}p(\z),
\end{align*}
and so $ cp\in \C[ {\bf z}]_{\leq {\bf d}}^{\rm in}$ for $ c=(i\exp(-i{\xv}/2))^{\bd} $.

\begin{Def}\label{def: Tlambda}
	Let $\mathcal{D}_{\lambda,\xv}:\C[{\bf z}]_{\bf d}\to \C[ {\bf z}]_{\bf d}$ denote the linear operator 
	corresponding to $\mathcal{D}_{\lambda}$ and tuple of Mobius transformations $\phi = (\phi_1, \hdots, \phi_n)$ where $\phi_j$ is defined as in 	\eqref{eq:Mobius} with $\theta = x_j$ for $\xv = (x_1, \hdots, x_n)\in \R^n$. 
Namely $p\mapsto (\phi^{-1}\circ \mathcal{D}_{\lambda}\circ \phi \cdot p^{\hom})|_{{\bf y}={\bf 1}}$, where $p^{\rm hom} = {\bf y}^{\bf d}p(z_1/y_1, \hdots, z_n/y_n)$.  
Explicitly, for $p({\bf z}) = \sum_{\alpha\le\bd}a_{\alpha}(-i)^{|\alpha|}(\z+\exp(i{\xv}))^{\alpha}(\z-\exp(i{\xv}))^{\bd-\alpha}$, 
	\[\mathcal{D}_{\lambda,\xv}p(\z)=p(\z)+\lambda\sum_{j=1}^{n}\sum_{\alpha\le\bd}\alpha_{j}a_{\alpha}(-i)^{|\alpha|}(\z+\exp(i{\xv}))^{\alpha-e_{j}}(\z-\exp(i{\xv}))^{\bd-\alpha}.\]
\end{Def}
\begin{cor} \label{cor:LY_top}
For any $p\in \A_{\bf d}({\xv})$ and $\lambda>0$ and $\phi$ defined as above, 
	$\mathcal{D}_{\lambda,\xv}(p) \in \A_{\bf d}({\xv})$ and $(\mathcal{D}_{\lambda,\xv})^{|{\bf d}|}(p) \in \A^{\circ}_{\bf d}({\xv})$.
	The interior of $\A_{\bf d}({\xv})$  in the Euclidean topology on $\C\cdot\C[{\bf z}]_{\leq \bf d}^{\rm in}$ is nonempty and contains $\A^{\circ}_{\bf d}({\xv})$.
	Moreover, $\A_{\bf d}({\xv})$ is contained in the closure of $\A^{\circ}_{\bf d}({\xv})$.
	\end{cor}
\begin{proof}
Note that $p\in \A_{\bf d}({\xv})$, resp. $\A^{\circ}_{\bf d}({\xv})$, if and only if 
the homogenezation of $\phi\cdot p$ belongs to $\C^*\mathfrak{S}_{\bf d}$, resp. $\C^*\mathfrak{S}^{\circ}_{\bf d}({\xv})$. 
The result then follows from Propositions~\ref{prop:perturbStable}, \ref{prop:Nuij}, and \ref{prop:stableInterior}.
\end{proof}

\begin{remark}
The set of Lee-Yang polynomials is connected but not contractible, even in $\P(\C[{\bf z}]_{\leq \bf d}^{\rm in})$. 
For example, the set of univariate Lee-Yang polynomials $p$ of degree-one, modulo global scaling, 
is parametrized by $z - e^{i\theta}$, for $\theta\in [0,2\pi]$, showing this set to be a circle. 
\end{remark}

\begin{remark}
The proof of Proposition~\ref{prop:Nuij} gives an explict contraction of $\A_{\bd}({\xv})$ (modulo scaling) to a polynomial $p^*\in \A_{\bd}^{\circ}({\xv})$, namely $\phi^{-1}\circ \mathcal{D}_{1}^{|{\bf d}|}{\bf z}^{\bf d}$, which we can explicitly compute. 
The space of real stable polynomials is contracted to 
\[
\mathcal{D}_{1}^{|{\bf d}|}{\bf z}^{\bf d} = (1 + \sum_{j=1}^n y_j\partial_{z_j})^{|{\bf d}|} \cdot {\bf z}^{\bf d}
= \sum_{\alpha}\binom{|{\bf d}|}{\alpha} {\bf y}^{\alpha} \partial_{{\bf z}}^{\alpha} {\bf z}^{\bf d}
= \sum_{\alpha\leq {\bf d}}\binom{|{\bf d}|}{\alpha} \frac{{\bf d}!}{\alpha!} {\bf y}^{\alpha} {\bf z}^{{\bf d} - \alpha}
\]
where the sum in the third term is taken over all $\alpha\in \Z_{\geq 0}^n$ with $|\alpha|\leq |{\bf d}|$, 
$\binom{|{\bf d}|}{\alpha} = \frac{|{\bf d}|!}{(|{\bf d}|-|\alpha|)!\alpha_1!\cdots \alpha_n!}$, 
and $\frac{{\bf d}!}{\alpha!} = \prod_{j=1}^n \left(\frac{d_j!}{\alpha_j!}\right)$.   
Taking $\phi$ as in \eqref{eq:Mobius}, we find that $\P(\A_{\bd}({\xv}))$ is contracted to 
\[
p^*(\z) = \phi^{-1}\cdot \mathcal{D}_{1}^{|{\bf d}|}{\bf z}^{\bf d} = 
\sum_{\alpha\leq {\bf d}}\binom{|{\bf d}|}{\alpha} \frac{{\bf d}!}{\alpha!}(\z-\exp(i\xv))^{\alpha}(-i(\z+\exp(i\xv)))^{\bd-\alpha} 
.
\]
\end{remark}

As above let $\C[{\bf z}]_{\leq {\bf d}}^{\rm in}$ denote
the real vectorspace of polynomials in $\C[{\bf z}]_{\leq {\bf d}}$ that are invariant under the involution
$\sum_{\alpha} a_{\alpha}{\bf z}^{\alpha} \mapsto \sum_{\alpha} {\overline{a_{{\bf d} - \alpha}}} {\bf z}^{\alpha}$.

\begin{thm} \label{thm: LYd full dimension}
For any ${\bf d}\in \Z_{\geq 0}^n$, the set of Lee-Yang polynomials $\A_{\bf d}$ is a full-dimensional semialgebraic subset of $\C\cdot \C[\bf z]_{\leq {\bf d}}^{\rm in}$. That is, $ \dim(\A_{\bf d})=\prod_{j=1}^n(d_j+1)+1 $.
Its interior in $\C\cdot \C[\bf z]_{\leq {\bf d}}^{\rm in}$ is nonempty and contains
\[
\A_{\bf d}^{\circ} = \{p\in \A_{\bf d} :p \text{ and }\nabla p\ \text{have no common zeros in } \T^{n}\}
\]
and $\A_{\bf d}$ is contained in the closure of $\A_{\bf d}^{\circ}$. 
\end{thm}
\begin{proof}
Note that the set 
\[
\{(p,{\bf a}, {\bf b}) \in \C\cdot \C[{\bf z}]_{\leq {\bf d}}^{\rm in} \times \R^n\times \R^n : p({\bf a} + i {\bf b}) = 0 \text{ and }(( a_j^2+b_j^2 < 1 \forall j) \text{ or }(
a_j^2+b_j^2 > 1 \forall j)) \}
 \]
is semialgebraic. By the Tarski-Seidenberg theorem its projection on to $\C\cdot \C[{\bf z}]_{\leq {\bf d}}^{\rm in}$ is also semialgebraic, as 
is the complement of the image of this projection, $\A_{\bf d}$.

Suppose that $p\in \A_{\bf d}$ and fix $\xv\in [0,2\pi)^n$ with $p(\exp(i\xv))\neq 0$. 
Then $p\in \A_{\bf d}(\xv)$ and we invoke Corollary~\ref{cor:LY_top}.
If $p$ and $\nabla p$ have no common zeros in  $\T^{n}$, then $p$ belongs to 
$\A_{\bf d}^{\circ}(\xv)$, which is contained in the interior of $ \A_{\bf d}(\xv) \subseteq  \A_{\bf d}$
in $\C\cdot \C[\bf z]_{\leq {\bf d}}^{\rm in}$.  Otherwise, $p$ is contained in the closure of $\A_{\bf d}^{\circ}(\xv)\subseteq \A_{\bf d}^{\circ}$. \end{proof}

\section{The torus zero set $ \Sigma_{p} $}\label{sec: torus zero set}

\begin{figure}
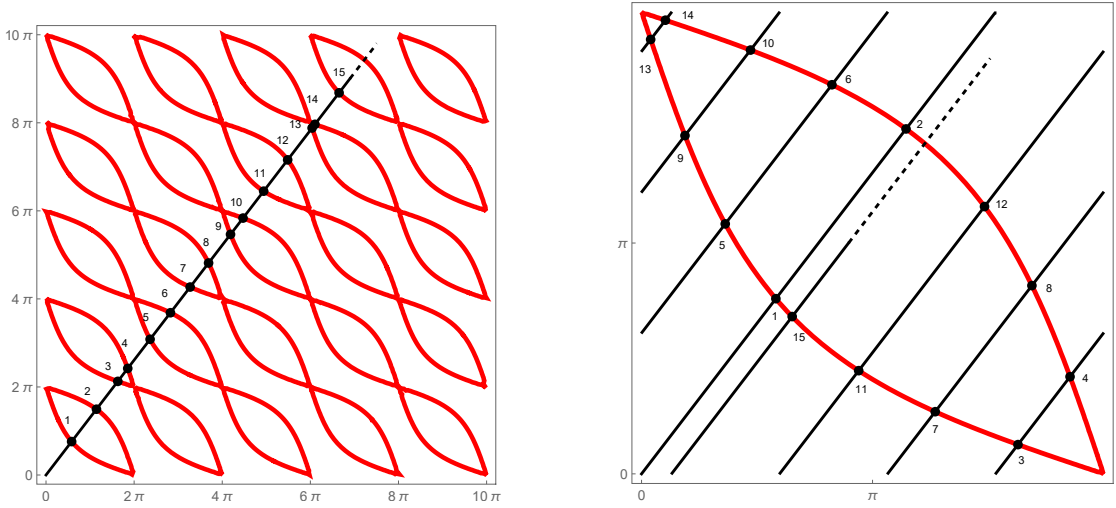

	\includegraphics[width=0.4\textwidth]{unfolded.pdf}  \hspace{.5in}	\includegraphics[width=0.4\textwidth]{folded.pdf} 
	\caption{
		(Left) The zero set of $p(e^{ix},e^{iy})$ (in red) and the line $(x,y)=t\lv$ (in black) as in \Cref{fig: flow on secman}.
		(Right) $ \Sigma_{p} $, the zero set modulo $ 2\pi $, and the line $(x,y)=t\lv\mod 2\pi$.}
	\label{fig: mod 2p}
\end{figure}

It is a simple, yet fruitful, observation that the zeros of $ x\mapsto p(\exp(ix\lv)) $ correspond to intersection points of the line $ \set{x\lv\mod{2\pi}\ :\ x\in\R}\subset \R^{n}/2\pi\Z^{n}  $ with the zero set 
 \[\Sigma_{p}:=\set{\xv\in\R^{n}/2\pi\Z^{n}\ :\ p(\exp(i\xv))=0 }.\]
See \Cref{fig: mod 2p}. In particular, certain properties of $ \mu_{p,\lv} $ are determined by the structure of $ \Sigma_{p} $, regardless of the choice of $ \lv\in\R_{+}^{n} $ with $ \Q $-independent entries. 
\begin{lem}[dimension and singularity]\label{lem: dimension and singularity}
	Given $ p\in\A_{\bd}(n) $, its torus zero set $ \Sigma_{p}\subset \R^{n}/2\pi\Z^{n}$ is a real analytic variety of dimension $ n-1 $, and
	\begin{enumerate}
		\item The set of singular points, $ \sing(\Sigma_{p}) $, is a subvariety of dimension at most $ n-2 $. If $ p $ has no square factors (i.e., square free), then $\xv\in\Sigma_{p}$ is singular if and only if $ \nabla p|_{\z=\exp(i\xv)}=0 $, or equivalently, its multiplicity is $ \mult(\xv)>1 $.
		\item Every irreducible factor of $ p $ is a Lee-Yang polynomial. If $ p $ is irreducible, then $ \Sigma_{p} $ is irreducible in the following sense: the zero set of $ q(\exp(i\xv)) $ in $ \Sigma_{p} $, for any polynomial $ q\in\C[z_{1},\ldots,z_{n}] $, is a subvariety of smaller dimension (at most $ n-2 $ dimensional), unless $ p $ is a factor of $ q $ in which case $ q(\exp(i\xv)) $ vanishes on $ \Sigma_{p} $. 
	\end{enumerate}
\end{lem}
\begin{proof}
	Since the real and imaginary parts of $ F(\xv)= p(\exp(i\xv)) $ are real analytic, then $ \Sigma_{p} $ is a real analytic variety. As such, its singular set $ \sing(\Sigma_{p}) $ is subvariety of lower dimension. To see why $ \Sigma_{p} $ is $ n-1 $ dimensional, let $ p\in\A_{\bd}(n) $ and let $ Z_{p} $ denote its zero set in $ \C^{n} $. As seen in \cite{Agler06toral}, if $ p $ is Lee-Yang, then $ Z_{p}\cap\T^{n} $ has real dimension $ n-1 $ and therefore $ \dim(\Sigma_{p})=n-1 $ by the homeomorphism $ \xv\mapsto\exp(i\xv) $ between them. Moreover, $ Z_{p}\cap\T^{n} $ is Zariski dense in $ Z_{p} $, according to \cite{Agler06toral}, which proves Part (2). 
	
	For part (1) suppose that $ p $ is square-free, so that the singular points of $ Z_{p} $ are exactly the points in $ Z_{p} $ where $ \nabla p=0 $ (if $ p $ has square factors this criteria fails at zeros of any multiple factor), or equivalently with multiplicity $ >1 $. Due to \Cref{prop:chain rule} with $ \varphi(x)=e^{ix} $, $ \xv\in\sing(\Sigma_{p}) $ if and only if $ \z=\exp(i\xv)\in\sing(Z_{p}) $. 
\end{proof}
\subsection{The layers structure of $ \Sigma_{p} $} 
It was shown in \cite[Lemma 4.14]{AloBanBer_conj}, for Lee-Yang polynomials arising from quantum graphs, that $ \Sigma_{p} $ is the union of $ 2n $ layers, each homeomorphic to $ (0,2\pi]^{n-1} $. These special polynomials are square free and have $ \bd=(2,2,\ldots,2) $ so $ 2n=|\bd| $. In this chapter, it is shown that for any $ p\in\A_{\bd} $, $ \Sigma_{p} $ is the union of $ |\bd| $ such layers, and in the case of polynomials with square factors multiplicities should be taken into account. 
\begin{prop}[Layers structure]\label{prop: Layer structure}
	Given $ p\in\A_{\bd}(n) $, $ \Sigma_{p} $ is the union of $ |\bd| $ layers, 
	\[\Sigma_{p}=\bigcup_{j=1}^{|\bd|}\Sigma_{p,j},\]
	each layer homeomorphic to $ (0,2\pi]^{n-1} $ through the parameterization $ \varphi_{j}:(0,2\pi]^{n-1}\to\Sigma_{p,j} $, 
	\[\varphi_{j}(\yv):=(\yv,0)+\theta_{j}(\yv,0)\textbf{1}\mod{2\pi},\]
	where $ \textbf{1}=(1,1,\ldots,1) $ and $ \theta_{j}:\R^{n}\to\R $ is a continuous function. Each $ \varphi_{j} $ is real analytic on the open set $\varphi_{j}^{-1}\left(\reg(\Sigma_{p})\right)\subset(0,2\pi]^{n-1} $ which has full Lebesgue measure. The multiplicity of $ \xv $ as a zero of $ p(\exp(i\xv)) $ is equal to the number of layers  $ \Sigma_{p,j} $ containing $ \xv $. In particular,if $ p $ is square free, then 
	\[\sing(\Sigma_{p})=\bigcup_{1\le i< j\le |\bd|}\Sigma_{p,i}\cap\Sigma_{p,j}.\]
\end{prop}
See Figure \ref{fig: tilt} for example of the layers structure of $ \Sigma_{p} $ for $ p $ from \Cref{ex:running}.
\begin{remark}[Square factors, overlaps, and multiplicities]
	Suppose that $ p=\prod_{j=1}^{N}q_{j}^{c_{j}} $ where $ (q_{j})_{j=1}^{N} $ are the distinct irreducible factors, each raised to the power $ c_{j}\in\N $. Define the reduced polynomial $ p^{\mathrm{red}}:=\prod_{j=1}^{N}q_{j} $, so that it is  square free and has the same zero set as $ p $, so $ \Sigma_{p^{\mathrm{red}}}=\Sigma_{p} $, but the total degree of $ p^{\mathrm{red}} $ may be smaller, in which case $ \Sigma_{p^{\mathrm{red}}} $ would have fewer layers than $ \Sigma_{p} $. This means that the layers coming from $ p $ must overlap, resulting in multiplicity. 
	Note that a given layer $\Sigma_{p,j}$ might comprise of pieces of the varieties of several different irreducible factors of $p$, each coming with their own multiplicities, which can differ. 
\end{remark}

 To prove \Cref{prop: Layer structure} the continuous \emph{phase functions} $ \theta_{j}:\R^{n}\to\R $, for $ j=1,\ldots,|\bd| $, are introduced in the next proposition.
\begin{figure}[h!]
	\includegraphics[width=0.9\textwidth]{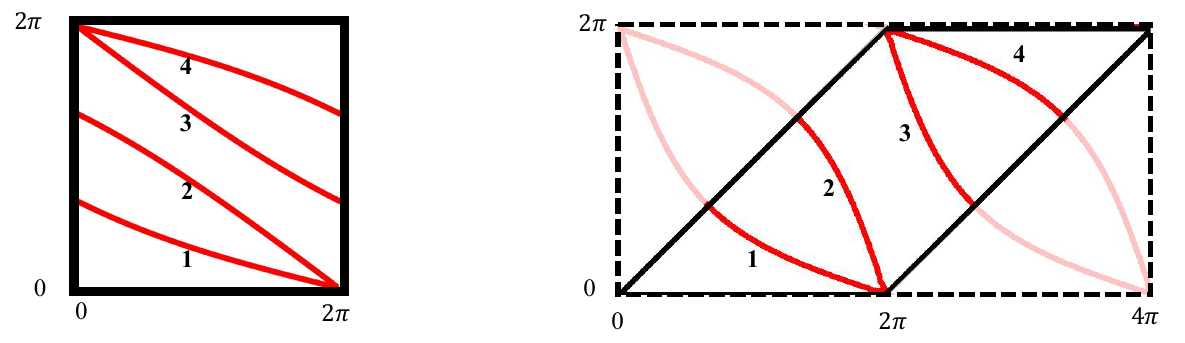}
	\caption{(right) The four layers of $ \Sigma_{p} $, presented in the tilted fundamental domain, for $ p\in\A_{(2,2)} $ given in \Cref{ex:running}. (left) The graphs of $\theta_j(y,0)$ for $y\in (0,2\pi]$.}
	\label{fig: tilt}
\end{figure}
\begin{Def}\label{def: px and m}
	Given $p = \sum_{\alpha}a_{\alpha}\z^{\alpha}\in \A_{\bd}(n)$ and $\xv \in\R^{n}$, define the univariate polynomial $ p_{{\xv}}(s)\in\C[s] $ by
	\begin{equation}\label{eq: def of pk}
		p_{{\xv}}(s)=p(se^{ix_{1}},\ldots,se^{ix_{n}} )=\sum_{j=0}^{|\bd|}\left(\sum_{|\alpha|=j}a_{\alpha}e^{i\langle \xv,\alpha\rangle}\right)s^{j} .
	\end{equation} 
	$ p_{\xv} $ has degree $ |\bd| $, with leading term $ a_{\bd}e^{i\langle \xv,\bd\rangle}s^{|\bd|} $, and all roots on the unit circle, say $ (e^{i\theta_{j}(\xv)})_{j=1}^{|\bd|} $. Let $ m(\theta_{j},\xv)$ denote the multiplicity of $ e^{i\theta_{j}(\xv)} $ as a root of $ p_{\xv} $, which agree with the multiplicity of $ \xv$ as a zero of $ p(\exp(i\xv)) $ when $ e^{i\theta_{j}(\xv)}=1 $, by \Cref{lem: all multiplicities agree}.
\end{Def}

\begin{prop}[Phase functions]\label{prop: phase functions} Given $ p\in\A_{\bd}(n) $, its \emph{phase functions} are $ |\bd| $ continuous functions $ \theta_{j}:\R^{n}\to\R $, such that $ (e^{i\theta_{1}(\xv)},\ldots,e^{i\theta_{|\bd|}(\xv)})$ are the roots of $ p_{\xv} $, ordered as follows $ \theta_{1}(\xv)\le\ldots\le\theta_{|\bd|}(\xv)\le\theta_{1}(\xv)+2\pi $, for all $ \xv\in\R^{n} $. Let $ \widehat{\Sigma}_{p} $ denote the lift of $ \Sigma_{p} $ to $ \R^{n} $, so that  
\begin{equation}\label{eq: p with thetaj}
	p(\exp(i\xv))=a_{\bd} e^{i\langle \bd, \xv\rangle} \prod_{j=1}^{|\bd|}\left(1-e^{i\theta_{j}(\xv)}\right),\ \text{  and  }\ \widehat{\Sigma}_{p}=\bigcup_{j=1}^{|\bd|}\theta_{j}^{-1}(2\pi \Z).
\end{equation}
The phase functions enjoy the following properties:
	\begin{enumerate}
		\item Each $ \theta_{j} $ satisfy $ \theta_{j}(\xv+t\textbf{1})=\theta_{j}(\xv)-t $, for all $ \xv\in\R^{n}$ and $ t\in\R $. More generally, $ \theta_{j} $ is monotonically decreasing when restricted to lines in any non-negative direction $ \lv\in\R_{\ge0}^{n} $, with upper and lower bounds on the slope $ -t\lv_{\mathrm{max}}\le\theta_{j}(\xv+t\lv)-\theta_{j}(\xv)\le-t\lv_{\mathrm{min}} $,
		 where $ \lv_{\mathrm{min}} $ and $ \lv_{\mathrm{max}} $ are the minimal and maximal entries of $ \lv $. 
		\item Each $ \theta_{j} $ is real analytic on $ \reg(\widehat{\Sigma}_{p}) $. It is also real analytic around any $ \xv\in\R^{n} $ which is not a discontinuity point of $  m(\theta_{j},\xv) $, the multiplicity of $ e^{i\theta(\xv)} $ as a root of $ p_{\xv} $. The discontinuity set of $  \xv\mapsto m(\theta_{j},\xv) $, denoted by $M_j\subset\R^{n}$, is a closed set of dimension $ \dim(M_{j})\le n-1 $, and
		$ \sing(\widehat{\Sigma}_{p})=\bigcup_{j=1}^{|\bd|}(\theta_{j}^{-1}(2\pi \Z)\cap M_{j})$.
			\item The sum of the phase functions is linear in $ \xv\in\R^{n} $
		\begin{equation*}
			\sum_{j=1}^{|\bd|}\theta_{j}(\xv)= \langle\bd,\xv\rangle+\sum_{j=1}^{|\bd|}\theta_{j}(0).
		\end{equation*}
		\item Translations by the lattice $ 2\pi\Z^{n} $ acts on the ordered tuple $ (\theta_{1},\ldots,\theta_{|\bd|}) $ by
		\[\theta(\xv+2\pi\textbf{n})\equiv\sigma^{\langle\bd,\textbf{n}\rangle}\theta(\xv)\mod{2\pi},\]
		for all $ \textbf{n}\in\Z^{n} $, where $ \sigma $ is the permutation $ (1,2,\ldots,|\bd|)\mapsto (|\bd|,1,2,\ldots,|\bd|-1)$.
		\end{enumerate}
\end{prop} 
\begin{remark}\label{rem: prefixed ordering}
	The choice of such phase functions is not unique. However, given any $ \xv_{0}\in\R^{n} $ which is a zero of $ p(\exp(i\xv)) $ of multiplicity $ m<|\bd| $, there is a unique choice of phase functions as in \Cref{prop: phase functions}, such that
	\[0=\theta_{1}(\xv_{0})=\ldots=\theta_{m}(\xv_{0})<\theta_{m+1}(\xv_{0})\le\ldots\le\theta_{|\bd|}(\xv_{0})< 2\pi.\] 
\end{remark}
The proof of \Cref{prop: phase functions} includes a proof of \Cref{rem: prefixed ordering}.
\begin{proof}[Proof of \Cref{prop: phase functions}] Fix arbitrary $ \xv_{0}\in\R^{n} $ such that $ p(\exp(i\xv_{0}))=0 $ with multiplicity $ m<|\bd| $, so that $ s=1 $ is a root of $ p_{\xv_{0}}(s) $ of multiplicity $ m $, by \Cref{lem: all multiplicities agree}. Let $ (s_{j}(\xv_{0}))_{j=1}^{|\bd|} $ denote the roots of $ p_{\xv_{0}} $, so we can write $ s_{j}(\xv_0)=e^{i\theta_{j}(\xv_{0})} $ such that
	\[0=\theta_{1}(\xv_{0})=\ldots=\theta_{m}(\xv_{0})<\theta_{m+1}(\xv_{0})\le\ldots\le\theta_{|\bd|}(\xv_{0})< 2\pi.\]
	The roots of a uni-variate polynomial changes continuously with its coefficients, as a result of Rouché's Theorem. The coefficients of $ p_{\xv} $ are analytic in $ \xv\in\R^{n}$, so the roots of $ p_{\xv_{0}} $ can extend continuously to the roots 
$ (s_{j}(\xv))_{j=1}^{|\bd|} $ of $ p_{\xv} $ for any $ \xv\in\R^{n} $, since $ \R^{n} $ is simply connected, and we may do it while maintaining their counter-clockwise ordering. Each $ s_{j}:\R^{n}\to S^{1} $ can be lifted to a (unique) continuous function $ \theta_{j}:\R^{n}\to\R $ with $ \theta_{j}(\xv_{0}) $ as prescribed above. Since the roots were kept in a counterclockwise order throughout $ \R^{n} $, then the relation $ \theta_{1}\le\ldots\le\theta_{|\bd|}\le\theta_{1}+2\pi $ holds everywhere. Since the leading coefficient of $ p_{\xv} $ is $ a_{\bd} e^{i\langle \bd, \xv\rangle} $, as stated in \Cref{def: px and m}, we may write $ p_{\xv}(s)=a_{\bd} e^{i\langle \bd, \xv\rangle} \prod_{j=1}^{|\bd|}\left(s-e^{i\theta_{j}(\xv)}\right) $. In particular, 
\[p(\exp(i\xv))=p_{{\xv}}(1)=a_{\bd} e^{i\langle \bd, \xv\rangle} \prod_{j=1}^{|\bd|}\left(1-e^{i\theta_{j}(\xv)}\right).\]
Since $ \widehat{\Sigma}_{p} $ is the zero set of $ p(\exp(i\xv)) $ in $ \R^{n} $, then it is the union of $ \theta_{j}^{-1}(2\pi\Z) $.

The univariate polynomial changes along the line $ \{\xv+t\textbf{1}:t\in\R\}$, for  $ \xv\in\R^{n}$, by
\[p_{\xv+t\textbf{1}}(s):=p(se^{i(t+x_{1})},\ldots,se^{i(t+x_{n})})=p_{\xv}(se^{it}).\]
Together with the continuity and ordering of the phase functions, it gives 
\begin{equation}\label{eq: 111 shift}
	\theta_{j}(\xv+t\textbf{1})=\theta_{j}(\xv)-t.
\end{equation}

	(Proof of (2)) The function $ \xv\mapsto m(\theta_{j},\xv) $ is integer valued, so it is continuous at a point if it is constant in a neighborhood of that point. Therefore $ M_{j} $, its set of discontinuity points, is closed. Let $ p_{\xv}^{(k)}(s) $ denote the $ k $-th derivative (in $ s $) of $ p_{\xv}(s) $. Given a point $ \xv\in\R^{n}\setminus M_{j} $ with $ m=m(\theta_{j},\xv) $, every $ \xv' $ in is some small neighborhood of $ \xv $ satisfy $ p_{\xv'}^{(k)}(s_{j}(\xv'))=0 $ for all $ k<m  $ and $ p_{\xv'}^{(m)}(s_{j}(\xv'))\ne0 $. Then, $ s_{j} $ is analytic around $ \xv $, by the implicit function theorem for analytic functions, as the $ s(\xv') $ solution of $p_{\xv'}^{(m-1)}(s)=0 $ around the point $ (s,\xv')=(s_{j}(\xv),\xv) $. We conclude that $ s_{j} $ is analytic on $ \R^{n}\setminus M_{j} $, and therefore $ \theta_{j} $ is real analytic on the same domain. Since $ \theta_{j}(\xv+t\textbf{1})=\theta_{j}(\xv)-t$ holds for all $ j $ simultaneously, then $ m(\theta_{j},\xv+t\textbf{1})=m(\theta_{j},\xv) $ for all $ \xv\in\R^{n} $ and $ t\in\R $. In particular, $ m(\theta_{j},\xv') $ is locally constant around a point $ \xv\in\theta_{j}^{-1}(2\pi k)\subset\widehat{\Sigma}_{p} $, namely $ \xv\notin M_{j} $, if and only if it is constant in some neighborhood of $ \xv $ in the level set $ \theta_{j}^{-1}(2\pi k) $. Since the multiplicity $ m(\xv) $ of $ \xv $ as a zero of $p(\exp(i\xv)) $ agree with $ m(\theta_{j},\xv)$ for $ \xv\in\theta_{j}^{-1}(2\pi\Z)\subset\Sigma_{p} $, by \Cref{lem: all multiplicities agree}, and the discontinuity set of $ m(\xv) $ over $ \widehat{\Sigma}_{p} $ is exactly $\sing(\widehat{\Sigma}_{p})$, we conclude that $ \sing(\widehat{\Sigma}_{p})=\cup_{j=1}^{|\bd|}(\theta_{j}^{-1}(2\pi\Z)\cap M_{j})$.	Next we show that $ \cup_{j=1}^{|\bd|}M_{j} $ is the projection of an analytic variety of dimension $ n-1 $, from which it follows that $ \dim(M_{j})\le n-1 $ for each $ M_{j} $. By \eqref{eq: 111 shift}, as discussed above, each $ M_{j} $ is invariant under translations in direction $ \textbf{1} $. In particular, using \eqref{eq: 111 shift} again, $ \xv\in \cup_{j=1}^{|\bd|}M_{j}$ if and only if $\xv+t\textbf{1}\in\sing(\widehat{\Sigma}_{p})$ for some $ t\in\R $. According to \Cref{lem: dimension and singularity}, $ \sing(\widehat{\Sigma}_{p}) $ is an analytic variety of dimension at most $ n-2 $, so $ \{(\xv,t)\in\R^{n}\times\R:\xv+t\textbf{1}\in\sing(\widehat{\Sigma}_{p})\} $ is an analytic variety of dimension at most $ n-1 $, and $ \cup_{j=1}^{|\bd|}M_{j} $ is the projection of this variety to $ \R^{n} $ and it is closed since each $ M_{j} $ is. We conclude that $ \cup_{j=1}^{|\bd|}M_{j} $ is a closed subanalytic set with dimension at most $ n-1 $ (locally around any point), see \cite{BierMilm98} for the definitions. 
	
	(Proof of (1)) 
		We claim that $ \nabla\theta_{j}(\xv)\in\R_{\le0}^{n} $ for all $ j $ and all $ \xv\in\R^{n}\setminus\cup_{j=1}^{|\bd|}M_{j} $. To see that, let $ \xv\in\R^{n}\setminus\cup_{j=1}^{|\bd|}M_{j} $, and since $ \nabla\theta_{j}|_{\xv}=\nabla\theta_{j}|_{\xv+t\textbf{1}} $ by \eqref{eq: 111 shift}, we may assume that $ \theta_{j}(\xv)=0 $. In particular, $ \xv\in\reg(\widehat{\Sigma}_{p}) $. Note that 	$\widehat{\Sigma}_{p}$ can also be written as the zero set of 
$F^{\rm red}(\xv) = p^{\rm red}(\exp(i\xv))$ for the reduced polynomial $p^{\rm red}$ of $p$.  
Since $\xv \in {\rm reg}(\widehat{\Sigma}_{p})$, then it has multiplicity one as a zero of $ F^{\rm red}(\xv) $, and there is a well defined normal vector to $\widehat{\Sigma}_{p}$ at $\xv$, which is proportional to both 
$\nabla\theta_j(\xv)$ and $\nabla F^{\rm red}({\bf x})$.
According to \Cref{prop: positive derivatives any degree}, the nonzero 
coordinates of $\nabla F^{\rm red}({\bf x})$  have the same phase, 
and therefore the nonzero coordinates of $\nabla \theta_j(\xv)\in \R^n$ all have the same sign. Since \eqref{eq: 111 shift} gives $ \nabla\theta_{j}(\xv)\cdot \textbf{1}=-1 $,  we find that $ \nabla\theta_{j}(\xv)\in\R_{\le0}^{n} $.
				
		It follows that $ \theta_{j}(\xv_{1})\ge \theta_{j}(\xv_{2}) $ whenever $ \xv_{2}-\xv_{1}\in\R_{\ge 0}^{n} $. To see why, we may use continuity to assume that both $ \xv_{1} $ and $ \xv_{2} $ lie in the open dense set $ \R^{n}\setminus X $ for $X= \cup_{j=1}^{|\bd|}M_{j} $. Consider all possible smooth curves $ \varphi:[1,2]\to\R^{n} $ with $ \varphi(1)=\xv_{1},\varphi(2)=\xv_{2} $ and $ \varphi'(t)\in\R_{\ge 0}^{n} $ for all $ t $. For such $ \varphi $, the composition $ \theta_{j}\circ\varphi $ is continuous for all $ t $, and smooth with non-positive derivative as long as $ \varphi(t)\notin X $. Since $ X $ is a closed subanalytic set of dimension at most $ n-1 $, there exists such $ \varphi $ that either intersects $ X $ transversely in a discrete set of points, or doesn't intersect $ X $ at all, by \cite[Theorem 1.2]{BierMilm98} and dimension count. For such $ \varphi $, $ \theta(\varphi(2))\ge \theta(\varphi(1))$.
		
		Now let $ \xv\in\R^{n}, t\in\R $, and $ \lv\in\R_{\ge 0}^{n} $. Consider the three points $\xv_{1}=\xv+t\lv_{\mathrm{min}}\textbf{1},\ \xv_{2}=\xv+t\lv,$ and $\xv_{3}=\xv+t\lv_{\mathrm{max}}\textbf{1}$, so  $ \xv_{3}-\xv_{2}\in\R_{\ge 0}^{n} $ and $ \xv_{2}-\xv_{1}\in\R_{\ge 0}^{n} $, which gives 
	\[\theta_{j}(\xv+t\lv_{\mathrm{max}}\textbf{1})\le\theta_{j}(\xv+t\lv)\le\theta_{j}(\xv+t\lv_{\mathrm{min}}\textbf{1}),\]
	and therefore, using \eqref{eq: 111 shift}, 
	\[\theta_{j}(\xv)-t\lv_{\mathrm{max}}\le\theta_{j}(\xv+t\lv)\le\theta_{j}(\xv)-t\lv_{\mathrm{min}}.\]

(Proof of (3)) Recall that $ p_{\xv}(s) =a_{\bd}e^{i\langle \bd, \xv\rangle} \prod_{j=1}^{|\bd|}\left(s-e^{i\theta_j({\xv})}\right) $ and by substituting $ s=0 $ we get $ p(\textbf{0})=p_{\xv}(0) =a_{\bd} (-1)^{|\bd|}e^{i(\langle \bd, \xv\rangle+\sum_{j=1}^{|\bd|}\theta_{j}(\xv))}\ne0 $
for all $ \xv\in\R^{n} $. Since $ \langle \bd, \xv\rangle+\sum_{j=1}^{|\bd|}\theta_{j}(\xv) $ is continuous and $ e^{i(\langle \bd, \xv\rangle+\sum_{j=1}^{|\bd|}\theta_{j}(\xv))}=(-1)^{|\bd|}\frac{p(\textbf{0})}{a_{\bd}} $ is constant, then 
\[\langle \bd, \xv\rangle+\sum_{j=1}^{|\bd|}\theta_{j}(\xv)=\sum_{j=1}^{|\bd|}\theta_{j}(\textbf{0}),\ \text{  for all  }\xv\in\R^{n}.\]

(Proof of (4)) To prove that $ \theta(\xv+2\pi\textbf{n})\equiv\sigma^{\langle\bd,\textbf{n}\rangle}\theta(\xv)\mod{2\pi} $ holds for all $ \xv\in\R^{n} $ and $ \textbf{n}\in\Z^{n} $, where $ \sigma $ is the permutation $ (1,2,\ldots,|\bd|)\mapsto(|\bd|,1,2,\ldots,|\bd|-1) $ and $ \theta(\xv)=(\theta_{1}(\xv),\ldots,\theta_{|\bd|}(\xv)) $, it is enough to consider standard basis vectors, namely  $ \textbf{n}=e_{i}$. We only consider $\textbf{n}= e_{1} $ but the proof holds for every $ e_{i} $. For every $ \xv\in\R^{n} $, the polynomials $ p_{\xv} $ and $ p_{\xv+2\pi e_{1}} $ are equal by \Cref{def: px and m}, so their roots are equal as a set but may have different counterclockwise numbering, which means that $ \theta(\xv+2\pi e_{1})=\sigma^{r}\theta(\x)+2\pi \textbf{k} $ for some integer $ 0\le r\le |\bd| $ and $ \textbf{k}\in \Z^{n} $ that may a-priori depend on $ \xv $. Notice that if the roots of $ p_{\xv} $ are all simple, then $ r $ and $ \textbf{k} $ are uniquely determined, however if all the roots have multiplicity two for example, then there can be two choice $ r $ and $ r+1 $. Nevertheless, as the roots of $ p_{\xv} $ and $ p_{\xv+2\pi e_{1}} $ changes continuously in $ \xv $ in the same manner, then there is a continuous (hence constant) choice of $ r $ and $ \textbf{k} $. It is therefore enough to show that $ r=d_{1} $ for some point $ \xv_{0} $ that minimise $ \min_{j\le |\bd|}m(\theta_{j},\xv) $, and as this quantity is invariant to translations in direction $ \textbf{1} $ then we may take $ \xv_{0}\in\widehat{\Sigma}_{p} $. Let $ m=m(\xv_{0})<|\bd| $, and by \Cref{rem: prefixed ordering} we may assume that 
\[0=\theta_{1}(\xv_{0})=\ldots=\theta_{m}(\xv_{0})<\theta_{j+1}(\xv_{0})\le\ldots\le\theta_{|\bd|}(\xv_{0})<2\pi.\]
Let $ r $ and $ \textbf{k}=(k_{1},\ldots,k_{n}) $ such that $ \theta(\xv_{0}+e_{1})=\sigma^{r}\theta(\xv_{0})+2\pi \textbf{k} $. Using part (3) and the fact that the sum of $ \sigma^{r}\theta(\xv_{0}) $ and $ \theta(\xv_{0}) $ is the same, we get 
\[\sum_{j=1}^{\bd}k_{j}=\sum_{j=1}^{\bd}\theta(\xv_{0}+e_{1})-\sum_{j=1}^{\bd}\theta(\xv_{0})=-2\pi d_{1}. \]
By part (1), $ \theta_{j}(\xv_{0}+e_{1})-\theta_{j}(\xv_{0})\in[0,2\pi] $. Since $ 2\pi k_{j}=\theta_{j}(\xv_{0}+e_{1})-\theta_{j'}(\xv_{0}) $ for some $ j' $, and $ |\theta_{j'}(\xv_{0})-\theta_{j}(\xv_{0})|<2\pi $, then $ k_{j}\in\{0,-1\} $ for all $ j $. The equation for the sum  above implies that $ \textbf{k} $ has exactly $ d_{1} $ entries equal to $ -1 $ and the rest are zero.

Denote $ v:=\sigma^{r}\theta(\xv_{0}) $ so that 
\[v=(\theta_{|\bd|-r+1}(\xv_{0}),\theta_{|\bd|-r+2}(\xv_{0}),\ldots,\theta_{|\bd|}(\xv_{0}),\theta_{1}(\xv_{0}),\ldots,\theta_{|\bd|-r}(\xv_{0})),\]
then $v_{1}\le\ldots\le v_{r} $ and $ v_{r+1}\le\ldots\le v_{|\bd|} $ with $ v_{r+1}=0<v_{r}<2\pi $, while $ v+2\pi\textbf{k} $ is ordered increasingly. We conclude that $ k_{j}=-1 $ for $ j\le r $ and $ k_{j}=0 $ for $ j>r $, which means that $ r=d_{1} $. This proves part (4).

\end{proof}
We are now in position to prove \Cref{prop: Layer structure} using \Cref{prop: phase functions}.
\begin{proof}[Proof of \Cref{prop: Layer structure}]
	Consider the linear transformation $ \ \mathcal{T}(\yv,t)=(\yv,0)+t\textbf{1} $, and the quotient map $ \pi:\R^{n}\to\R^{n}/2\pi\Z^{n} $. Consider $ \Omega:=\mathcal{T}((0,2\pi]^{n}) $, which is a fundamental domain of $ 2\pi\Z^{n} $ so $\pi:\Omega\to\R^{n}/2\pi\Z^{n} $ is bijective. The map $ \yv\mapsto \mathcal{T}(\yv,\theta_{j}(\yv,0)) $ is continuous with a continuous inverse $ \xv\mapsto (x_{1}-x_{n},\ldots,x_{n-1}-x_{n})$, since $ \theta_{j} $ is continuous by \Cref{prop: phase functions}, and therefore $\varphi_{j}(\yv)= \pi(\mathcal{T}(\yv,\theta_{j}(\yv,0)))$ is a homeomorphism between $ (0,2\pi]^{n-1} $ and its image, which we denote by $ \Sigma_{p,j} $.
	
	Notice that $ \theta_{j}(\mathcal{T}(\yv,t))=\theta_{j}(\yv,0)-t $ by Part (1) of \Cref{prop: phase functions}, so 
	\begin{equation}\label{eq: varphi T}
		\theta_{j}(\mathcal{T}(\yv,t))\in 2\pi\Z\iff \varphi_{j}(\yv)=\pi(\mathcal{T}(\yv,t)).
	\end{equation}
	
	($ \Sigma_{p,j}\subset\Sigma_{p} $) Given $ \yv\in(0,2\pi]^{n-1} $ let $ t=\theta_{j}(\yv,0) $, so that $ \theta_{j}(\mathcal{T}(\yv,t))=\theta_{j}(\yv,0)-t=0 $. Therefore, $ \mathcal{T}(\yv,t)\in\widehat{\Sigma}_{p} $ which means that $ \pi(\mathcal{T}(\yv,t))=\varphi(\yv)\in\Sigma_{p} $ by \eqref{eq: varphi T}. 
	
	($ \Sigma_{p}\subset \cup_{j=1}^{|\bd|}\Sigma_{p,j} $) Consider $ \Omega:=\mathcal{T}((0,2\pi]^{n}) $, which is a fundamental domain of $ 2\pi\Z^{n} $ so $\pi:\Omega\to\R^{n}/2\pi\Z^{n} $ is bijective, and therefore any $ \xv\in\Sigma_{p}\subset\R^{n}/2\pi\Z^{n} $ has a unique point $ (\yv,t)\in(0,2\pi]^{n-1}\times(0,2\pi] $ for which $ \pi(\mathcal{T}(\yv,t))=\xv $. In such case, $ \mathcal{T}(\yv,t)\in\widehat{\Sigma}_{p} $ so $ \theta_{j}(\mathcal{T}(\yv,t))\in 2\pi\Z$ for some $ j $, and $ \varphi_{j}(\yv)=\pi(\mathcal{T}(\yv,t))=\xv $ by \eqref{eq: varphi T}.

	(Multiplicity and singularity) Let $ \xv=\pi(\mathcal{T}(\yv,t))\in \Sigma_{p} $ as above. The number of layers containing $ \xv $ is the number of $ j $'s for which $ \varphi_{j}(\yv)=\xv $, which are those $ j $'s for which $ e^{i\theta_{j}(\mathcal{T}(\yv,t))}=1$. This is exactly $ m(\xv) $, the multiplicity of $ \xv $ as a zero of $ p(\exp(i\xv)) $, by \Cref{lem: all multiplicities agree}. If $ p $ is square-free, then $ \xv\in\sing(\Sigma_{p})\iff m(\xv)>1\iff \xv\in\Sigma_{p,i}\cap\Sigma_{p,j}  $ for some $ i\ne j $, by \Cref{lem: dimension and singularity}.

	(Real analyticity) Suppose that $ \varphi_{j}(\yv_{0})=\pi(\mathcal{T}(\yv_{0},t))\in\reg(\Sigma_{p}) $, then $ \mathcal{T}(\yv_{0},t)\in\reg(\widehat{\Sigma}_{p}) $ which means that $ \theta_{j} $ is real analytic around $ \mathcal{T}(\yv_{0},t)=(\yv_{0},0)+t\textbf{1} $, according to \Cref{prop: phase functions}. Therefore, $ \theta_{j} $ is real analytic around $ (\yv_{0},0) $, due to the shift $ \theta_{j}((\yv,0)+v)=\theta_{j}(\mathcal{T}(\yv,t)+v)+t $ for all $ v\in\R^{n},\yv\in\R^{n} $. It follows that $ \yv\mapsto \mathcal{T}(\yv,\theta_{j}(\yv,0)) $, and hence $ \varphi_{j} $, are real analytic around $ \yv_{0} $. 
\end{proof}

%
%
%



\section{Zeros density - proof of Theorem \ref{thm: zeros density and upper bound}}\label{sec: proof of thm zeros density and upper bound}
If $p $ has the form $ p(\z)=\det(1-\diag(\z)U) $, then Theorem \ref{thm: zeros density and upper bound} holds for $ p $ by the proof of Weyl's Law for quantum graphs in \cite[Lemma 3.7.4]{BerKuc_graphs}. The proof for a general Lee-Yang polynomial $ p $ is similar, where the roots $ (e^{i\theta_{j}(\xv)})_{j=1}^{|\bd|} $ of the univariate polynomial $ p_{\xv}(s) $ replace the eigenvalues of $ \diag(\exp(i\xv))U $. 
\begin{proof}[Proof of Theorem \ref{thm: zeros density and upper bound}]
	Let $ p\in \A_{\bd} $ and consider the phase functions $ (\theta_{j})_{j=1}^{|\bd|} $ described in \Cref{prop: phase functions}. Given $ \lv\in\R_{+}^{n} $, a point $ x\in\R $ is a zero of $f(x)=p(\exp(ix\lv)) $ of multiplicity $ m $ if and only if exactly $ m $ of the phase functions satisfy $ \theta_{j}(x\lv)\in 2\pi\Z $, by \Cref{lem: all multiplicities agree} and \Cref{prop: phase functions}. The number of zeros of $ p(\exp{(ix\lv)}) $, counted with multiplicities, in an interval $ [a,b]\subset\R $ is therefore
	\[\mu_{p,\lv}([a,b])=  \sum_{j=1}^{|\bd|}|\set{x\in[a,b]~~:~~\theta_{j}(x\lv)\in 2\pi\Z}|.\] 
	According to \Cref{prop: phase functions} part (1), $ \theta_{j}(a\lv)>\theta_{j}(b\lv) $ for each $ j $, and the map $ x\mapsto\theta_{j}(x\lv) $ is a bijection between $ [a,b] $ and the interval $ [\theta_{j}(b\lv),\theta_{j}(a\lv)]\subset\R $ of length $ \theta_{j}(a\lv)-\theta_{j}(b\lv) $. Therefore,   
	\[|\set{x\in[a,b]:\theta_{j}(x\lv)\in 2\pi\Z}|=\left|[\theta_{j}(b\lv),\theta_{j}(a\lv)]\cap 2\pi\Z\right|=\frac{\theta_{j}(a\lv)-\theta_{j}(b\lv)}{2\pi}+\mathrm{err}_{j},\]
	with $ |\mathrm{err}_{j}|\le 1 $. Let $ \mathrm{err}:=\sum_{j=1}^{|\bd|}\mathrm{err}_{j} $, so it is bounded by $ |\mathrm{err}|\le|\bd| $, and 
\[
		\mu_{p,\lv}([a,b]) \ = \ \sum_{j=1}^{|\bd|}\frac{\theta_{j}(a\lv)-\theta_{j}(b\lv)}{2\pi} +\mathrm{err}
		 \  = \   \frac{\langle\bd,\lv\rangle}{2\pi}|b-a|+\mathrm{err}.
\]
	In the last equality we used part (3) of \Cref{prop: phase functions}. This proves part (1) of the theorem, by substituting $ [a,b]=[x,x+T] $ and $ \mathrm{err}(x,T)=\mathrm{err} $.\\
	
	For  part (2) of the theorem, let $ x_{j+1}>x_{j} $ be consecutive zeros of $ f(x) $ and consider an arbitrary interval $ I\subset(x_{j},x_{j+1}) $, so 
	\[0=\mu_{p,\lv}(I)\le \frac{\langle\bd,\lv\rangle}{2\pi}|I|+\mathrm{err}\quad\Rightarrow\quad|I|\le 2\pi\frac{|\mathrm{err}|}{\langle\bd,\lv\rangle}\le2\pi\frac{|\bd|}{\langle\bd,\lv\rangle}  ,\]
	and $ |I| $ can get arbitrarily close to $ x_{j+1}-x_{j} $.
\end{proof}

\section{Ergodic dynamics on $ \Sigma_{p} $}\label{sec: ergodicity}
To prove the existence of a gap distribution for the eigenvalues of a quantum graph, Barra and Gaspard introduced an $ \lv $-dependent ``first return'' dynamical system on $ \Sigma_{p} $, for the associated Lee-Yang polynomial $ p $, which is uniquely ergodic when $ \lv $ is $ \Q $-linearly independent \cite{BarGas_jsp00}. The same holds for any Lee-Yang $ p $, as shown in this section.

Given $ \lv\in\R_{+}^{n} $, consider the linear flow on $ \R^{n}/2\pi\Z^{n} $ induced by the constant vector field $ \lv $. That is, the flow at time $ t $ is a map $ \xv\mapsto\xv+t\lv\mod 2\pi$ from $ \R^{n}/2\pi\Z^{n} $ to itself. The minimal $ t>0 $ for which a point $ \xv\in\Sigma_{p} $ gets back to $ \Sigma_{p} $ is called the \emph{first return time} $ \tau_{\lv}(\xv) $, and $ \xv\mapsto\xv+\tau_{\lv}(\xv)\lv\mod{2\pi} $ is a map from $ \Sigma_{p} $ to itself that defines a dynamical system.
\begin{remark}
	Throughout this subsection we omit the ``$\text{mod }{2\pi} $'' when it is clear from the context.
\end{remark} 
\begin{figure}[h!]
	\includegraphics[width=0.9\textwidth]{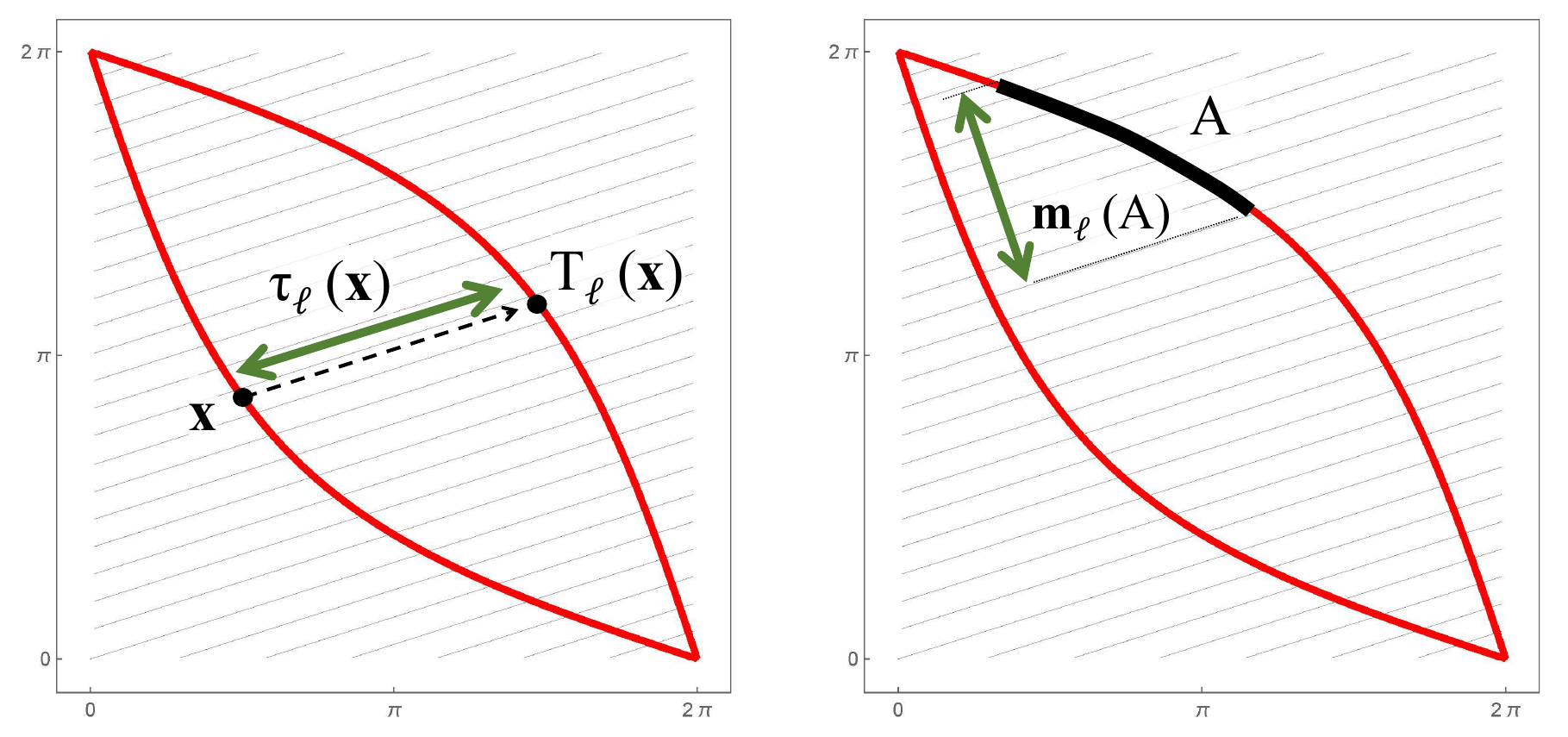}
	\caption{Illustration of $ T_{\lv},\tau_{\lv} $, and the measure $ \mm_{\lv} $, as in \Cref{def: Dynamical system}, for the Lee-Yang polynomial $ p $ from \Cref{ex:running} and $ \lv=(\pi,1) $. In the background the line $ (x,y)=t\lv\mod{2\pi} $ for $ t\in[0,44] $.}
	\label{fig: dynamics}
\end{figure}
\begin{Def}[Dynamical system on $ \Sigma_{p} $]\label{def: Dynamical system} Let $ p\in\A_{\bd}(n) $ and $ \lv\in\R_{+}^{n} $. The \emph{first-return time} $ \tau_{\lv}:\Sigma_{p}\to \R_{+} $ and the \emph{first-return map} $ T_{\lv}:\Sigma_{p}\to\Sigma_{p} $ are defined by,
	\[\tau_{\lv}(\xv):=\min\{t>0:\xv+t\lv\in\Sigma_{p}\},\ \text{ and }\ T_{\lv}(\xv):=\xv+\tau_{\lv}(\xv)\lv.\]
	The measure $ \mm_{\lv} $ is a Borel measure on $ \Sigma_{p} $ defined for any Borel subset $ A\subset\Sigma_{p} $ by
	\[\mm_{\lv}(A):=\lim_{\epsilon\to 0}\frac{\vol_{n}(A_{\epsilon\lv})}{2\epsilon},\quad\text{with}\quad A_{\epsilon\lv}:=\set{\xv+t\lv\ :\ \xv\in A,\ |t|<\epsilon },\]
	where $ \vol_{n} $ is $ n $-dimensional volume (Lebesgue measure) in $ \R^n/2\pi\Z^{n} $.   
\end{Def}

\begin{Def}
	A function $ h:\Sigma_{p}\to\C $ is called \emph{Riemann integrable} if the closer of its discontinuity set has zero volume in $ \Sigma_{p} $, with respect to the $ n-1 $ dimensional volume form induced by the $ n $-dimensional volume form on $ \R^{n}/2\pi\Z^{n} $.  
\end{Def}
Recall that if $ p $ has a decomposition into distinct irreducible factors $ p=\prod_{j=1}^{M}q^{c_{j}}_{j} $, then the reduced polynomial is $ p^{\mathrm{red}}:=\prod_{j=1}^{M}q_{j} $ and its multi-degree is denoted by $ \bd^{\mathrm{red}} $. Let $ \mult(\xv) $ denote the multiplicity of $ \exp(i\xv) $ as a zero of $p $. 
\begin{thm}[Unique Ergodicity]\label{thm: ergodicity}
	Let $ p\in\A_{\bd}(n) $, let $ \lv\in\R_{+}^{n} $ with $ \Q $-linearly independent entries, and fix an arbitrary point $ \xv_{0}\in\Sigma_{p} $. Let $ (x_{j})_{j\in\Z} $ denote the zeros of $f(x)=p(\exp(i(\xv_{0}+x\lv)))$, ordered increasingly with multiplicities, and consider $ (T_{\lv}^{j}(\xv_{0}))_{j\in\Z} $, the $ T_{\lv} $ orbit of $ \xv_{0} $.
	Then the averages of any bounded Riemann integrable $ h:\Sigma_{p}\to\C $ over the orbit $ (T_{\lv}^{j}(\xv_{0}))_{j\in\N} $, and over the sequence $ (\xv_{0}+x_{j}\lv)_{j\in\N} $, are independent of $ \xv_{0} $ and are given by
	\begin{align}\label{eq: ergodicity 1}
				\lim_{N\to\infty}\frac{1}{N}\sum_{j=1}^{N}h(T_{\lv}^j(\xv_{0})) & =\frac{1}{(2\pi)^{n-1}\langle\bd^{\mathrm{red}},\lv\rangle}\int_{\Sigma_{p}}h(\xv)d\mm_{\lv}(\xv)\\
		\lim_{N\to\infty}\frac{1}{N}\sum_{j=1}^{N}h(\xv_{0}+x_{j}\lv) & =\frac{1}{(2\pi)^{n-1}\langle\bd,\lv\rangle}\int_{\Sigma_{p}}\mult({\bf x})h(\xv)d\mm_{\lv}(\xv),\label{eq: ergodicity 2}		
	\end{align} 
where $ \mm_{\lv}(\Sigma_{p})=(2\pi)^{n-1}\langle\bd^{\mathrm{red}},\lv\rangle $ and $ \int_{\Sigma_{p}}\mult(\xv)\ d\mm_{\lv}(\xv)=(2\pi)^{n-1}\langle\bd,\lv\rangle  $. 
\end{thm}
For Lee-Yang polynomials associated to quantum graphs, this is shown in \cite{BarGas_jsp00,BerWin_tams10,CdV_ahp15}. A proof for any Lee-Yang polynomial is provided for completeness.
\begin{proof}	Let $ \{x_{i}\}_{i\in\N} $ denote the positive zeros of $ f(x)=p(\exp(i(\xv_{0}+x\lv))) $ ordered \emph{with multiplicity}, and let $ x_{0}=0 $, since $ \xv_{0}\in\Sigma_{p} $. Let $ \{k_{i}\}_{i\in\N} $ denote the \emph{distinct} zeros of $ f $, ordered \emph{without multiplicity}, with $ k_{0}=0 $, so that $ T_{\lv}^{i}(\xv_{0})=\xv_{0}+k_{i}\lv $ and $ m(T_{\lv}^{i}(\xv_{0})) $ is the multiplicity of $ k_{i} $ as a zero of $ f $ for all $ i\in\N $. The first step of the proof is showing that for any bounded Riemann integrable $ h:\Sigma_{p}\to\C $,
	\begin{equation}\label{eq: ergodic average of h}
	\lim_{R\to\infty}\frac{1}{R}\sum_{k_{i}\le R}h(T_{\lv}^{i}(\xv_{0}))=\frac{1}{(2\pi)^{n}}\int_{\Sigma_{p}}hd\mm_{\lv}.
\end{equation}
	Consider a layer $ \Sigma_{p,j} $ as in \Cref{prop: Layer structure} and let $ h=\chi_{A} $ be the indicator function of a Borel set $A\subset \Sigma_{p,j} $ with boundary of zero volume in $ \Sigma_{p} $. The set $ A_{\epsilon\lv}=\set{\xv+t\lv:(\xv,t)\in A\times[-\epsilon,\epsilon]}  $ is then a Borel set with boundary of zero volume in $ \R^{n}/2\pi\Z^{n} $. Since $ \lv $ has $ \Q $-linearly independent entries, the Kronecker-Weyl Theorem gives 
	\begin{equation}\label{eq: Kronecker}
		\frac{\vol_{n}(A_{\epsilon\lv})}{(2\pi)^{n}}=\lim_{R\to\infty}\frac{\mathrm{length}\left(\{t\in[0,R]:\xv_{0}+t\lv\in A_{\epsilon\lv}\}\right)}{R}.
	\end{equation}
Let $ \mathcal{A}=\{k_{i}:T_{\lv}^{i}(\xv_{0})\in A\}\subset\R $, so that $ \theta_{j}(\xv_{0}+k_{i}\lv)\in 2\pi\Z $ for all $ k_{i}\in\mathcal{A} $ since $ A\subset\Sigma_{p,j} $. The function $ t\mapsto\theta_{j}(\xv_{0}+t\lv) $ is strictly monotone with uniform upper and lower bounds on its slope, by \Cref{prop: phase functions} part (1), so $ \mathcal{A}$ is uniformly discrete, and therefore, for small enough $ \epsilon>0 $, the $ 2\epsilon $-intervals $ [k_{i}-\epsilon,k_{i}+\epsilon] $ for $ k_{i}\in\mathcal{A} $ are mutually disjoint. The set $ \{t\in[0,R]:\xv_{0}+t\lv\in A_{\epsilon\lv}\} $ is the intersection of these disjoint $ 2\epsilon $-intervals with $ [0,R] $, so up to an error of $ 2\epsilon $, its lengths is $ 2\epsilon \left|\mathcal{A}\cap[0,R]\right|=2\epsilon\sum_{k_{i}\le R}h(T_{\lv}^{i}(\xv_{0})) $. Substituting this estimate into \eqref{eq: Kronecker} gives
	\begin{equation}\label{eq: estimate O eps}
		\frac{\vol_{n}(A_{\epsilon\lv})}{(2\pi)^{n}}=\lim_{R\to\infty}\left(2\epsilon\frac{1}{R}\sum_{k_{i}\le R}h(T_{\lv}^{i}(\xv_{0}))+\frac{1}{R}O(\epsilon)\right)=2\epsilon\lim_{R\to\infty}\frac{1}{R}\sum_{k_{i}\le R}h(T_{\lv}^{i}(\xv_{0})). 
	\end{equation}
Dividing both sides by $ 2\epsilon $ and taking $ \epsilon\to 0 $ proves \eqref{eq: ergodic average of h} for the indicator function $ h=\chi_{A} $. Both sides of \eqref{eq: ergodic average of h} are linear in $ h $, so it holds for any step function $ \sum_{j=1}^{N}c_{j}\chi_{A_{j}} $ such that the sets $ A_{j}\subset\Sigma_{p} $ are Borel with boundary of zero volume in $ \Sigma_{p} $. Such functions can approximate (in the sup-norm) any non-negative bounded Riemann integrable function from below and above to any given precision, by taking the upper and lower Darboux sums as they converge to the Riemann integral of $ h $. We conclude that \eqref{eq: ergodic average of h} holds for any bounded Riemann integrable function $ h:\Sigma_{p}\to\C $, as it can be written as $ h=h_{1}-h_{2}+i(h_{3}-h_{4}) $ such that each $ h_{j} $ is real non-negative, bounded, and Riemann integrable and hence can be approximated by step functions for which \eqref{eq: ergodic average of h} holds.

The second step is calculating $ \mm_{\lv}(\Sigma_{p}) $ and $ \int_{\Sigma_{p}}m(\xv)d\mm_{\lv}(\xv) $. The sum of multiplicities of distinct zeros up to $ T $ is the number of repeated zeros up to $ R $, $ \sum_{k_{i}\le R}m(T_{\lv}^{i}(\xv_{0}))=|\{x_{i}<R\}| $ which equals to $ \frac{\langle\bd,\lv\rangle}{2\pi}R+O(1) $, by \Cref{thm: zeros density and upper bound}, and applying \eqref{eq: ergodic average of h} to $ h=m $ gives
\[\frac{1}{(2\pi)^{n}}\int_{\Sigma_{p}}m(\xv)d\mm_{\lv}(\xv)=\lim_{R\to\infty}\frac{1}{R}\sum_{k_{i}\le T}m(T_{\lv}^{i}(\xv_{0}))=\lim_{R\to\infty}\frac{|\{x_{i}<R\}|}{R}=\frac{\langle\bd,\lv\rangle}{2\pi}.\]
It follows that $ \int_{\Sigma_{p}}m(\xv)d\mm_{\lv}(\xv)=(2\pi)^{n-1}\langle\bd,\lv\rangle $, and by replacing $ p $ with $ p^{\mathrm{red}} $ we get that $ \mm_{\lv}(\Sigma_{p})=\int_{\Sigma_{p}}d\mm_{\lv}=(2\pi)^{n-1}\langle\bd^{\mathrm{red}},\lv\rangle $. To see why notice that the torus zero set of $ p^{\mathrm{red}} $ is equal to $ \Sigma_{p} $, with the same measure $ \mm_{\lv} $, but with multiplicity functions which is one for every $ \xv\in\reg(\Sigma_{p}) $. The complement has $ \mm_{\lv}(\sing(\Sigma_{p}))=0  $, since $ \dim(A)\le n-2 $ for $ A=\sing(\Sigma_{p}) $, which means that $ \dim(A_{\epsilon\lv})\le n-1 $ and so $ \vol_{n}(A_{\epsilon\lv})=0 $.

To prove \eqref{eq: ergodicity 1} apply \eqref{eq: ergodic average of h} twice and divide the two limits
\begin{equation*}
	\frac{\int_{\Sigma_{p}}hd\mm_{\lv}}{\mm_{\lv}(\Sigma_{p})}=\lim_{R\to\infty}\frac{ \sum_{k_{i}\le R}h(T_{\lv}^{i}(\xv_{0}))}{\left|\{k_{i}\le R\}\right|}=\lim_{N\to\infty}\frac{\sum_{i=1}^{N}h(T_{\lv}^{i}(\xv_{0}))}{N}.
\end{equation*}
Since $ \sum_{x_{i}\le R}h(\xv_{0}+x_{i}\lv)=\sum_{k_{i}\le R}m(T_{\lv}^{i}(\xv_{0}))h(T_{\lv}^{i}(\xv_{0}))  $, the same argument gives
\begin{equation*}
	\frac{\int_{\Sigma_{p}}m(\xv)h(\xv)d\mm_{\lv}(\xv)}{\int_{\Sigma_{p}}m(\xv)d\mm_{\lv}(\xv)}=\lim_{R\to\infty}\frac{ \sum_{x_{i}\le R}h(\xv_{0}+x_{i}\lv)}{\left|\{x_{i}\le R\}\right|}=\lim_{N\to\infty}\frac{\sum_{i=1}^{N}h(\xv_{0}+x_{i}\lv)}{N}.
\end{equation*}
\end{proof}
\subsection{Properties of $ \tau_{\lv} $ and $ \mm_{\lv} $}
The gap distributions in \Cref{sec: gap distributions} are defined in terms of $ \tau_{\lv} $ and $ \mm_{\lv} $. The needed properties of $ \tau_{\lv} $ and $ \mm_{\lv} $ are stated in the next two lemmas. 

In what follows, consider $ \reg(\Sigma_{p}) $ as a smooth Riemannian manifold with volume form $ d\sigma $, induced by $ d\vol_{n} $ in $ \R^{n}/2\pi\Z^{n} $, and the normal vector field $ \hat{n} $ with $ \hat{n}(\xv)\in\R_{\ge 0}^{n} $ for all $ \xv\in\reg(\Sigma_{p}) $, as guaranteed by \Cref{prop: positive derivatives any degree}. The  $ n-1 $ form with $ dx_{j} $ missing is denoted by $ dx_{1}\wedge dx_{2}\wedge\ldots\wedge \widehat{dx}_{j}\wedge\ldots\wedge dx_{n} $.

\begin{lem}\label{lem: BG measure} The measure $\mm_{\lv} $ is absolutely continuous with respect to $ d\sigma $, the volume form on $ \reg(\Sigma_{p}) $, with a strictly positive distribution
	\begin{equation}\label{eq: dm}
		d\mm_{\lv}=\langle\hat{n},\lv\rangle d\sigma= \sum_{j=1}^{n}\lv_{j}(-1)^{j+1}dx_{1}\wedge dx_{2}\wedge\ldots\wedge \widehat{dx}_{j}\wedge\ldots\wedge dx_{n}.
	\end{equation}
	For each layer $ \Sigma_{p,j} $, with parameterization $ \varphi_{j}:(0,2\pi]^{n-1}\to \Sigma_{p,j} $ as in \Cref{prop: Layer structure}, and for every measurable $ h:\Sigma_{p,j}\to\C $,
	\begin{equation}\label{eq: dm2}
		\int_{\Sigma_{p,j}}h(\xv) d\mm_{\lv}(\xv)=-\int_{(0,2\pi]^{n-1}}h(\varphi_{j}(\yv))\langle \nabla \theta_{j}(\yv,0),\lv\rangle d\yv,
	\end{equation}
	and in particular, for $ \lv=\textbf{1} $, 
	\begin{equation}\label{eq: dm3}
		\int_{\Sigma_{p,j}}h(\xv) d\mm_{\textbf{1}}(\xv)=\int_{(0,2\pi]^{n-1}}h(\varphi_{j}(\yv))d\yv.
	\end{equation}
\end{lem}
\begin{proof}[Proof of \Cref{lem: BG measure}] It was shown in the proof of \Cref{thm: ergodicity}
	that $ \mm_{\lv}(\sing(\Sigma_{p}))=0 $, so $ \mm_{\lv} $ is supported on $ \reg(\Sigma_{p}) $.  To show \eqref{eq: dm} it is enough to consider a small open set $ A\subset\reg(\Sigma_{p}) $. If $ A $ is sufficiently small, for $ \epsilon>0 $ sufficiently small, we can choose local coordinates $ \kappa=(\kappa_{1},...,\kappa_{n-1}) $ such that $ d\kappa=d\sigma $, which extend to local coordinates in a neighborhood of $ A_{\epsilon\lv} $ by adding a coordinate $ t $ in the normal direction $ \hat{n} $. The fact $ d\sigma $ is induced from $ d\vol_{n} $ gives $ d\vol_{n}=d\kappa dt $. Therefore, $ \vol_{n}(A_{\epsilon\lv})=\int_{A}|\langle\hat{n}(\kappa),2\epsilon\lv\rangle|d\kappa=2\epsilon\int_{A}\langle\hat{n}(\kappa),\lv\rangle d\kappa>0 $, using that $ \lv\in\R_{+}^{n} $ and $ \hat{n}(\kappa)\in\R_{\ge 0}^{n} $ for all $ \kappa\in A $. We conclude that 
	\[\mm_{\lv}(A):=\lim_{\epsilon\to 0}\frac{\vol_{n}(A_{\epsilon\lv})}{2\epsilon}=\int_{A}\langle\hat{n},\lv\rangle d\sigma.\]
	By definition, the form $ \langle\hat{n},\lv\rangle d\sigma $ agree with the $ n-1 $ form
	\[\omega=\sum_{j=1}^{n}\lv_{j}(-1)^{j+1}dx_{1}\wedge dx_{2}\wedge\ldots\wedge \widehat{dx}_{j}\wedge\ldots\wedge dx_{n},\] when restricted to $ \reg(\Sigma_{p}) $. 
	
	We are left with deducing \eqref{eq: dm2} from \eqref{eq: dm} by simple change of variables. Let $ \yv\in(0,2\pi]^{n-1} $ such that $ \varphi_{j}(\yv)\in\reg(\Sigma_{p})$, let $D=D\varphi_{j}|_{\yv}$ be the $n\times (n-1)$ matrix of derivatives whose $(s,i)$th entry is $ \frac{\partial (\varphi_{j})_{s} }{\partial y_{i}}|_{\yv} $. Then the change of variables formula for $ \xv=\varphi_{j}(\yv) $ is
	\[\sum_{k=1}^{n}\lv_{k}(-1)^{k+1}dx_{1}\wedge dx_{2}\wedge\ldots\wedge \widehat{dx}_{k}\wedge\ldots\wedge dx_{n}=\sum_{j=k}^{n}\lv_{k}(-1)^{k+1}D_{k}d\yv,\]
	where $ D_{j} $ denotes the $(n-1)\times (n-1) $ minor of $ D $ obtained by removing the $ j $-th row. Adding $ \lv $ as a column vector gives an $ n\times n $ matrix $ M=\begin{pmatrix}
		D & \lv
	\end{pmatrix} $, whose determinant is exactly $ \det(M)=\sum_{k=1}^{n}\lv_{k}(-1)^{k+1}D_{k} $, by expanding according to the column $ \lv $. We need to show that $ \det(M)=-\langle\nabla\theta_{j}(\yv,0),\lv\rangle $. Let $v =(\frac{\partial \theta_{j}(\yv,0)}{\partial y_{1}},\ldots,\frac{\partial \theta_{j}(\yv,0)}{\partial y_{n-1}})\in\R^{n-1} $, so that the entries of $ D $ are $ D_{s,i}=v_{i} $ if $ i\ne s $, and $ D_{i,i}=v_{i}+1 $, since $ \varphi_{j}(\yv)= (\yv,0)+\theta_{j}(\yv,0)\textbf{1} $. Subtracting the last row of $ M $ from all other gives the matrix
\[\tilde{M}=\begin{pmatrix}
	\mathrm{id}_{n-1} & \tilde{\lv}\\
	v & \lv_{n}
\end{pmatrix},\ \text{  with  }\  \tilde{\lv}=(\lv_{1}-\lv_{n},\lv_{2}-\lv_{n},\ldots,\lv_{n-1}-\lv_{n}),\] so that $ \det(M)=\det(\tilde{M})=\lv_{n}-\langle v,\tilde{\lv}\rangle, $
using Schur complement in the last equality. Notice that $ \langle v,\tilde{\lv}\rangle=\langle \nabla\theta_{j}(\yv,0),\lv-\lv_{n}\textbf{1}\rangle=\langle \nabla\theta_{j}(\yv,0),\lv\rangle-\lv_{n}$ since $ \langle \nabla\theta_{j}(\yv,0),\textbf{1}\rangle=-1 $ by \Cref{prop: phase functions} part (1). We conclude that $ \det(M)=-\langle \nabla\theta_{j}(\yv,0),\lv\rangle $ which proves \eqref{eq: dm2}, and \eqref{eq: dm3} follows from $ \langle \nabla\theta_{j}(\yv,0),\textbf{1}\rangle=-1 $ again.
\end{proof}

\begin{lem}\label{lem: tau}
	Let $ p\in\A_{\bd}(n)$, $  \lv\in\R_{+}^{n}  $, and let $ \{x_{j}\}_{j\in\Z} $ denote the ordered zeros of $ f(x)=p(\exp(ix\lv)) $. Then $ \tau_{\lv}(x_{j}\lv)=x_{j+1}-x_{j} $ whenever $ x_{j+1}>x_{j} $, where the function $ \tau_{\lv}:\Sigma_{p}\to\R_{+} $, introduced in \Cref{def: Dynamical system}, is bounded by $ \frac{2\pi |\bd|}{\langle\bd,\lv\rangle} $ and satisfies  
		\begin{align*}
			\{\tau_{\lv}(\xv):\xv\in\sing(\Sigma_{p})\} & \subset \overline{\{\tau_{\lv}(\xv):\xv\in\reg(\Sigma_{p})\} },\ \text{and}\\
		\mathrm{inf}\{\tau_{\lv}(\xv):\xv\in\reg(\Sigma_{p})\}& =0  \iff \sing(\Sigma_{p})\neq\emptyset.
	\end{align*}
Furthermore, the map $(\xv,\lv)\mapsto  \tau_{\lv}(\xv,\lv) $ is continuous on $ \reg(\Sigma_{p})\times\R_{+}^{n} $, and real analytic on   $\{(\xv,\lv)\in\reg(\Sigma_{p})\times\R_{+}^{n}:T_{\lv}(\xv)\in\reg(\Sigma_{p})\} $, which is an open subset of $ \reg(\Sigma_{p})\times\R_{+}^{n} $.	 
\end{lem}
\begin{proof} Since $ p $ and the reduced polynomial $ p^{\mathrm{red}} $ share the same torus zeros set, then they share the same $ \tau_{\lv} $, and so we may assume that $ p $ is square free. Consider $ \widehat{\Sigma_{p}} $,the lift of $ \Sigma_{p} $ from $ \R^{n}/2\pi\Z^{n} $ to $ \R^{n} $, and define $ \tau(\xv,\lv):=\tau_{\lv}(\xv\mod{2\pi}) $ for $ (\xv,\lv)\in\widehat{\Sigma_{p}}\times\R_{+}^{n} $. Namely, $ \tau(\xv,\lv)=\min\{t>0:p(\exp(i(\xv+t\lv)))=0\} $. By definition, $ \tau(x_{j}\lv,\lv)=x_{j+1}-x_{j} $ whenever $ x_{j+1} $ and $ x_{j} $ are distinct consecutive zeros of $ p(\exp(ix\lv)) $. 
	
	Fix an arbitrary point $ (\xv_{0},\lv)\in\widehat{\Sigma_{p}}\times\R_{+}^{n}  $. By replacing $ p(\z) $ with $ p(e^{ix_{1}'}z_{1},\ldots,e^{ix_{n}'}z_{n}) $ if needed, we may assume that $ \xv_{0}=\textbf{0}\in \widehat{\Sigma_{p}}$. For the bound, assume by contradiction that $ \tau(\textbf{0},\lv)\ge c> \frac{2\pi |\bd|}{\langle\bd,\lv\rangle}$, so $ p(\exp(it\lv))\ne 0 $ for all $ t\in(0,c) $, which contradicts Theorem \ref{thm: zeros density and upper bound}. 
	
	The two cases of $ \textbf{0} $ being a regular or a singular point of $ \widehat{\Sigma_{p}} $ are treated separately.

	($ \textbf{0}\in\reg(\widehat{\Sigma_{p}}) $) If $ \textbf{0} $ is a regular point of $ \widehat{\Sigma_{p}} $, then it is a zero of $ p(\exp(i\xv)) $ of multiplicity one, since $ p $ is square free. The phase functions (defined in \Cref{prop: phase functions}) can be chosen according to \Cref{rem: prefixed ordering}, such that 
\[0=\theta_{1}(\textbf{0})<\theta_{2}(\textbf{0})\le\ldots\le \theta_{|\bd|}(\textbf{0})<2\pi,\]
Taking $ U\subset\widehat{\Sigma}_{p} $ a small enough neighborhood of $ \textbf{0} $, we can ensure that $ \theta_{1}(\xv)=0 $ and $ \theta_{j}(\xv)\in(0,2\pi) $ for all $ \xv\in U $ and $ j\ge2 $. In particular, the minimal $ t>0 $ for which $ p(\exp(i(\xv+t\lv))=0 $ must satisfy $ \theta_{2}(\xv+t\lv)=0 $, for any $ (\xv,\lv)\in U\times\R_{+}^{n} $, by the ordering and strict monotonicity of the phase functions as shown in \Cref{prop: phase functions}. 
%
	In such case, $ \tau=\tau(\xv,\lv) $ is the unique solution 
	to $ \theta_{2}(\xv+\tau\lv)=0 $ and is therefore continuous in $ (\xv,\lv) $, by the continuity of $ (\xv,\lv,t)\mapsto \theta_{2}(\xv+t\lv) $ and the implicit function theorem for monotone continuous functions. As a result $ (\xv,\lv)\mapsto T_{\lv}(\xv)=\xv+\tau(\xv,\lv)\lv $ is also continuous in $ U\times\R_{+}^{n} $, and therefore the set $ \{(\xv,\lv)\in U\times\R_{+}^{n}:T_{\lv}(\xv)\in\reg(\widehat{\Sigma}_{p})\} $ is an open subset of $ U\times\R_{+}^{n} $. If $ (\xv',\lv')\in\Omega $, then $ \theta_{2}(\xv+t\lv) $ is real analytic in $ (\xv,\lv,t) $ around $ (\xv',\lv',\tau(\xv',\lv')) $, by \Cref{prop: phase functions}, and so $ \tau(\xv,\lv) $ is real analytic around $ (\xv',\lv') $ by the implicit function theorem for real analytic functions. We conclude that $ \tau(\xv,\lv) $ is continuous on $ \reg(\widehat{\Sigma}_{p})\times\R_{+}^{n} $, and real analytic on   $\{(\xv,\lv)\in\reg(\widehat{\Sigma}_{p})\times\R_{+}^{n}:T_{\lv}(\xv)\in\reg(\widehat{\Sigma}_{p})\} $, which is an open subset of $ \reg(\widehat{\Sigma}_{p})\times\R_{+}^{n} $. 
	
	It follows that if $ \sing(\Sigma_{p})=\emptyset $, then $ \tau_{\lv}:\Sigma_{p}\to \R_{+} $ is continuous, for every $ \lv\in\R_{+}^{n} $. Since $ \Sigma_{p} $ is compact, $ \tau_{\lv} $ obtains its minimum at some point $ \tilde{\xv} $, so $ \mathrm{inf}\{\tau_{\lv}(\xv):\xv\in\Sigma_{p}\}=\tau_{\lv}(\tilde{\xv})>0 $.

	($ \textbf{0}\in\sing(\widehat{\Sigma_{p}}) $) If $ \textbf{0} $ is a singular point of $ \widehat{\Sigma}_{p} $, then it has multiplicity $ m=m(\textbf{0}) $ as a zero of $ p(\exp(i\xv)) $. Choose the phase functions according to \Cref{rem: prefixed ordering} such that \[0=\theta_{1}(\textbf{0})=\ldots=\theta_{m}(\textbf{0})<\theta_{m+1}(\textbf{0})\le\ldots\le \theta_{|\bd|}(\textbf{0})<2\pi.\]
	For small enough $ \epsilon>0 $, the neighborhood of $ \textbf{0} $, $ U=\set{\xv\in\widehat{\Sigma}_{p}: \|\xv\|<\epsilon} $, has the form
	\[U=\cup_{j=1}^{m}U_{j},\quad \text{ with }\ U_{j}:=\{\xv\in\R^{n}:\|\xv\|<\epsilon,\ \theta_{j}(\xv)=0\},\]
	 since $ \widehat{\Sigma}_{p}=\cup_{j=1}^{|\bd|}\theta_{j}^{-1}(2\pi\Z) $ and the phase function are continuous. Define $ t_{j}(\xv,\lv) $ as the unique $ t$-solution to $ \theta_{j}(\xv+t\lv)=0 $. As before, $ t_{j} $ is continuous on $ U\times\R_{+}^{n} $, and $ \tau(\textbf{0},\lv)=t_{m+1}(\textbf{0},\lv) $ (where $ \theta_{|\bd|+1}=\theta_{1}+2\pi $ if $ m=|\bd| $). Furthermore, for any $ j\le m $ and $ \xv\in U_{j} \cap \reg(\widehat{\Sigma}_{p})$, $ \theta_{j+1}(\xv)>0 $, and so $ \tau(\xv,\lv)=t_{j+1}(\xv,\lv) $. Consider a converging sequence $ \xv_{n}\to\textbf{0} $, with $ \xv_{n}\in\reg(\widehat{\Sigma}_{p}) $ for all $ n $, and by taking a subsequence if needed, we may assume $ \xv_{n}\in\reg(\widehat{\Sigma}_{p})\cap U_{j} $ for all $ n $, for some specific $ j $. So $\tau(\xv_{n},\lv)=t_{j+1}(\xv_{n},\lv) $ for all $ n $, and 
	\[\lim_{n\to\infty}\tau(\xv_{n},\lv)=t_{j+1}(\textbf{0},\lv)=\begin{cases}
		\tau(\textbf{0},\lv) & \mbox{if}\ j=m\\
		0 & \mbox{if}\ 1\le j<m-1
	\end{cases},\]
by continuity of $ t_{j+1} $, using that $ \theta_{j+1}(\textbf{0})=0 $ when $ j+1\le m $. It follows that $ \tau_{\lv} $ is upper semi continuous, that $ \mathrm{inf}\{\tau_{\lv}(\xv):\xv\in\reg(\Sigma_{p})\}=0 $, and that $ 	\{\tau_{\lv}(\xv):\xv\in\sing(\Sigma_{p})\} \subset \overline{\{\tau_{\lv}(\xv):\xv\in\reg(\Sigma_{p})\} }$.    
 
\end{proof}


\section{Proof of Theorem \ref{thm: irreduicible}}\label{sec: proof of elaboration}
Let $ p\in\A_{\bd}(n) $ with decomposition $ p=\prod_{j=1}^{N}q_{j}^{c_{j}} $ into distinct irreducible polynomials, and let $ \lv\in\R_{+}^{n} $ with $ \Q $-linearly independent entries. Each factor $ q_{j}$ is a Lee-Yang polynomial by definition. Let $ m_{p}(x) $ denote the multiplicity of $ x $ as a zero of $ f_{p}(x)=p(\exp(ix\lv)) $, with $ m_{p}(x)=0 $ if $ f_{p}(x)\ne0 $, and similarly let $ m_{j}(x) $ denote the multiplicity with respect to $ f_{j}(x)=q_{j}(\exp(ix\lv)) $. Since $ f(x)=\prod_{j=1}^{N}\left(f_{j}(x)\right)^{c_{j}} $ and multiplicity of zeros is additive under multiplication of functions, then $ m(x)=\sum_{j=1}^{N}c_{j}m_{j}(x) $. As a result 
\[\mu_{p,\lv}=\sum_{x\in\Lambda}m_{p}(x)\delta_{x}=\sum_{j=1}^{N}c_{j}\sum_{x\in\Lambda_{j}}m_{j}(x)\delta_{x}=\sum_{j=1}^{N}c_{j}\mu_{q_{j},\lv},\]
where $ \Lambda $ denotes the zero set of $ f $ and $ \Lambda_{j} $ the zero set of $ f_{j} $. Clearly, $ \Lambda=\bigcup_{j=1}^{N}\Lambda_{j}  $. The proof of \Cref{thm: irreduicible} follows from the next lemma and proposition, considering the case of $ p $ being irreducible and either binomial or not.
\begin{lem}[Binomial]\label{lem: binoimal}
	If $ p\in\A_{\bd}(n) $ is binomial, normalized such that $ p({\bf 0})=1 $, then $ p(\z)=1-e^{-i\varphi}\z^{\bd} $ for some $ \varphi\in\R $. In such case, for any $ \lv_{+}^{n} $, the zeros of $ f(x)=p(\exp(ix\lv)) $ are simple and form an infinite arithmetic progression $ \{\varphi+\frac{2\pi}{\langle \bd,\lv\rangle }k\ :\ k\in\Z \} $. 
\end{lem}  
\begin{proof}
	If $ p\in\A_{\bd}(n) $ then $ p({\bf 0})\ne0$ and the coefficient of $ \z^{\bd} $ is non-zero. If it has only two monomials and $ p({\bf 0})=1$, then $ p(\z)=1+a\z^{\bd} $. Assume by contradiction that $ |a|\ne 1 $, then for any $|\bd|$-th root $\omega\in \C$ of $a$, the 
	point ${\bf z} = (\omega, \hdots, \omega)$ will be a root of $p$ in $\D^n$ or $(\C\backslash \overline{\D})^n$, in contradiction to  $ p\in\A_{\bd}(n) $. Therefore $ p(\z)=1-e^{-i\varphi}\z^{\bd} $, and so $ f(x)=1-e^{i(\langle \bd,\lv\rangle x -\varphi)} $, for some $ \varphi\in\R $. Hence, $ f(x)=0\iff x-\varphi\in\frac{2\pi}{\langle \bd,\lv\rangle }\Z  $ in which case $ f'(x)\ne0 $.   
\end{proof}
 \begin{prop}[Non-binomial]\label{prop: non-binomial}
 	Let $ p\in\A_{\bd}(n) $ be irreducible and non-binomial, $ \lv\in\R_{+}^{n} $ be $ \Q $-linearly independent, and $ f(x)=p(\exp(ix\lv)) $ with zero set $ \Lambda $ and multiplicities $ (m(x))_{x\in\Lambda} $. Then,
 	\begin{enumerate}
 		\item $ m(x)\le|\bd| $ for all $ x\in\Lambda $ and $ \lim_{R\to\infty}\frac{\left|\{|x|<R\ :\ x\in\Lambda,\ m(x)=1\}\right|}{\left|\{|x|<R\ :\ x\in\Lambda\}\right|}=1 $.
 		\item For any $ N\in\N $ and any set $ \Gamma\in\R $ with $ \dim_{\Q}(\Gamma)=N $, $ |\Lambda\cap\Gamma|\le c $, with uniform bound $ c=c(|\bd|,N) $ that only depends on $ |\bd| $ and $ N $. In particular, $ \dim_{\Q}(\Lambda)=\infty $.  
 	\end{enumerate} 
 \end{prop}
\begin{proof}[Proof of \Cref{prop: non-binomial} part (1)]
	The bound $ m(x)\le|\bd| $ follows from \Cref{thm: zeros density and upper bound} part (1). Numbering the distinct zeros of $ f(x) $ by $ (k_{j})_{j\in\Z} $, with $ k_{j}>0 $ for $ j>0 $ and $ k_{j}<0 $ for $ j<0 $. We need to show that 
	\[\lim_{N\to\infty}\frac{\left|\{-N\le j\le N\ : \ m(k_{j})>1\}\right|}{2N}=0.\]
	Let $ p\inv\in\A_{\bd} $ as in \Cref{def: inv}, so that $ p\inv $ is also irreducible, non-binomial, and has $ p\inv(\exp(ix\lv))=0\iff p(\exp(-ix\lv))=0 $ with the same multiplicities, so it is enough to prove the one sided limit 	
	\[\lim_{N\to\infty}\frac{\left|\{1\le j\le N\ : \ m(k_{j})>1\}\right|}{N}=0.\]
	By \Cref{lem: all multiplicities agree} and since $ p $ is irreducible, $ m(k_{j})>1 $ if and only if $ k_{j}\lv\in\sing(\Sigma_{p}) $. Notice that $ k_{j}\lv=T_{\lv}^{j}(k_{0}\lv) $, using the fact that the $ k_{j} $s are the distinct zeros. Let $ h $ be the indicator function of $ \sing(\Sigma_{p}) $, so that $ \left|\{1\le j\le N\ : \ m(k_{j})>1\}\right|=\sum_{j=1}^{N}h(T_{\lv}^{j}(k_{0}\lv)) $. Then, $ h $ is bounded Riemann integrable and \Cref{thm: ergodicity} gives  
		\[\lim_{N\to\infty}\frac{1}{N}\sum_{j=1}^{N}h(T_{\lv}^{j}(k_{0}\lv))\propto\int h(\xv)d\mm_{\lv}(\xv)=\mm_{\lv}(\sing(\Sigma_{p}))=0.\]
\end{proof}

\begin{remark}
The proof of \Cref{prop: non-binomial} part (2) is a consequence of \cite[Theorem 1.2]{Evertse00},  often known as Lang's $ G_{m} $ Theorem. 
\end{remark}

To state we consider $ (\C^{*})^{n} $ as multiplicative group, and it will be convenient to define the notions of \emph{rank, division group,} and \emph{algebraic torus cosets} in terms of the exponent map $ \exp:\C^n\to(\C^{*})^{n} $. 
\begin{Def}
	A subgroup $ G\subset(\C^{*})^{n} $ has \emph{rank} $ N $, if $ N $ is the minimal integer for which $ G=\{\exp(A\textbf{k}):\textbf{k}\in\Z^{N}\} $ for some matrix $ A\in\C^{n\times N} $. Its \emph{division group} is defined by $ \overline{G}=\{\exp(A\textbf{k}):\textbf{k}\in\Q^{N}\} $ for the same $ A $. An \emph{algebraic torus} of dimension $ d $ \emph{torus} in $ (\C^{*})^{n} $ has the form $ H=\set{\exp{(B\yv)}:\yv\in\C^{d}} $ for some integer matrix $ B\in\Z^{n\times d}  $ of rank $ d $. The \emph{algebraic torus coset} $ \z H $ for $ \z=\exp(\xv) $ is the set $ \z H=\set{\exp{(\xv+A\yv)}~~|~~\yv\in\C^{d}} $, for the same matrix $ B $. It also has dimension $ d $.
\end{Def}
\begin{thm*}[Lang's $ G_{m} $ Theorem]\cite[Theorem 1.2]{Evertse00}
	Let $V\subset (\C^{*})^n$ be an algebraic variety of dimension $ N $ and degree $ D $, and let $ G $ be a subgroup of $ (\C^{*})^n $, of rank $ N $, with division group $ \overline{G}$. Then $ \overline{G}\cap V $ is contained in a union of at most $ r $ algebraic torus cosets $ \z_{j}H_{j}\subset V $ for $ r\le e^{(N+1)\left(6D\binom{n+D}{D}\right)^{\left(5D\binom{n+D}{D}\right)}} $. 
\end{thm*}
In what follows, recall our notation $ \dim_{\Q}(A) $, for a set $ A\subset\R $, that stands for the dimension (as a $ \Q $-vector space) of the $ \Q $-linear span of the elements of $ A $, and that $ \dim_{\Q}(\textbf{a}) $ for a real vector $ \textbf{a}=(a_{1},\ldots,a_{N}) $ means $ \dim_{\Q}(\{a_{1},\ldots,a_{N}\}) $. 
\begin{lem}\label{lem: torus cosets}
	Suppose that $ \z H\subset(\C^{*})^{n} $ is an algebraic torus coset of dimension $d\le n-2$, and that $\lv\in\R^{n}$ has $\Q$-linearly independent entries. Then there is at most one $k\in\R$ such that $\exp{(ik\ell)}\in \z H$.    
\end{lem}
\begin{proof}
	Let $ B\in\Z^{n\times d}  $ of rank $ d $ such that $ H=\set{\exp{(B\yv)}:\yv\in\C^{d}} $, and suppose that both $ \exp{(ik\ell)} $ and $ \exp{(ik'\ell)} $ lie in $ \z H $. Then, $ \exp{(i(k-k')\ell)}\in H $ and therefore $ (k-k')\ell=B \yv+2\pi \textbf{k} $ for some $\textbf{k}\in\Z^{n}, \yv\in \C^d $. The left kernel of $ B $ in $ \C^{n} $ contains an $(n-d)$-dimensional $ \Q $-linear vector space of vectors orthogonal to $ B\yv $ so $ \dim_{\Q}(B\yv)\le d $,	and therefore,
	\[\dim_{\Q}((k-k')\ell)=\dim_{\Q}(B\yv+2\pi \textbf{k})\le d+1<n.\]
	However, if $ k-k'\ne 0 $ then $\dim_{\Q}((k-k')\ell)= \dim_{\Q}(\lv)=n $, a contradiction.
\end{proof}

\begin{proof}[\Cref{prop: non-binomial} part (2)] Let $ p,\lv $, and $ \Lambda $ as in \Cref{prop: non-binomial}. Let $V\subset(\C^{*})^{n}$ be the zero set of $ p $ in $ (\C^{*})^{n} $. The degree of $ V $ is finite and only depends on $ |\bd| $. Given $ N\in\N $, let $ \Gamma\subset\R $ of $ \dim_{\Q}(\Gamma)=N $, so $ \Gamma=\set{\langle \textbf{a},\textbf{k}\rangle\ : \ \textbf{k}\in\Q^{N}} $ for some $ \textbf{a}\in\R^{N} $. Define the matrix $ A\in\C^{n\times N} $ whose $ j $-th row is the vector $ i\lv_{j}\textbf{a}\in\C^{N} $, and let $ G=\{\exp(A\textbf{k}):\textbf{k}\in\Z^{N}\} $ so that its division group is $ \overline{G}=\{\exp(A\textbf{k}):\textbf{k}\in\Q^{N}\}=\set{\exp(it\lv)\ :\ t\in\Gamma}$. In particular,
	\[x\in\Lambda\cap\Gamma\iff \exp(ix\lv)\in \overline{G}\cap V.\]
	Since $ G $ has rank at most $ N $, Lang's $ G_{m} $ Theorem says that there are at most $ r=r(|\bd|,N) $ algebraic torus cosets $\z_{i}H_{i}\subset V$ such that $ \overline{G}\cap V\subset\z_{1}H_{1}\cup\ldots\cup \z_{r}H_{r} $. In particular, any $ x\in \Lambda\cap\Gamma $ satisfies $ \exp(ix\lv)\in \z_{i}H_{i}  $ for some $ i $. An algebraic torus coset of dimension $ n-1 $ is the zero set of a binomial polynomial, and since $ p $ is irreducible and not binomial, then $ \dim(\z_{i}H_{i})\le\dim(V)-1=n-2 $ for every $ i=1,\ldots,r $. By Lemma \ref{lem: torus cosets}, each $ \z_{i}H_{i} $ contains at most one point $ \exp(ix\lv) $ for $ x\in\R $. We conclude that $ \Lambda\cap \Gamma$ contains at most $ r $ points.   
\end{proof}


\section{Proof of Theorem \ref{thm: minimal gap}}\label{sec: minimal gap}
\begin{proof}[Proof of Theorem \ref{thm: minimal gap}]
Suppose that $ n\ge2 $. Say that $ p\in \A_{\bd}(n)$ satisfies (i) if $ p $ and $ \nabla p $ have no common zeros in $ \T^{n} $, and satisfies (ii) if $ p $ has a non-binomial factor. Say that $ \mu_{p,\lv} $ satisfies (*) if it is non-periodic, with unit coefficients and has a uniformly discrete support. The proof of \Cref{thm: minimal gap} consists of three parts.

((i)+(ii)$\Rightarrow$(*)) It follows from \Cref{thm: irreduicible} that $ \mu_{p,\lv} $ is non-periodic when $ \lv $ is $ \Q $-linearly independent and $ p $ satisfies (ii). It is left to show that if $ p $ satisfies (i), then $ \mu_{p,\lv} $ has unit coefficients and uniformly discrete support for any $ \lv\in\R_{+}^{n} $. Assume that $ p $ satisfies (i) and $ \lv\in\R_{+}^{n} $. Property (i) is equivalent to $ \sing(\Sigma_{p})=\emptyset $ and $ \mult(\xv)\equiv 1 $ for all $ \xv\in\Sigma_{p} $. According to \Cref{lem: all multiplicities agree}, this means that the multiplicities of the zeros of $ f(x)=p(\exp(ix\lv)) $, which are the coefficients in $ \mu_{p,\lv} $, are all equal to one. According to \Cref{lem: tau}, $ \sing(\Sigma_{p})=\emptyset $ implies that $ r=\mathrm{inf}\{\tau_{\lv}(\xv):\xv\in\Sigma_{p}\}>0 $. The zeros of $ f $ are distinct, so their gaps are given by $ \tau_{\lv} $, as seen in \Cref{lem: tau}, providing uniform lower bound $ x_{j+1}-x_{j}=\tau_{\lv}(x_{j}\lv)\ge r>0 $.

((*)$\Rightarrow$(i)+(ii)) Let $ p\in\A_{\bd}(n) $ with $\Q  $-linearly independent $ \lv\in\R_{+}^{n} $, and assume that $ \mu_{p,\lv} $ satisfies (*). Let $ \Lambda $ be the support of $ \mu_{p,\lv} $, so it is non-periodic and uniformly discrete. If $ p $ had only binomial factors, then $ \Lambda $ would be a union of infinite arithmetic progressions, by \Cref{thm: irreduicible}, and such a union is either periodic or it has gaps as small as we wish. We conclude that $ p $ satisfies (ii), and it is left to show (i), namely that $ \sing(\Sigma_{p})=\emptyset $ and $ m(\xv)\equiv 1 $. Let $ (x_{j})_{j\in\Z} $ be the zeros of $ f(x)=p(\exp(ix\lv)) $, ordered increasingly, so by (*) they are all simple and $ \tau_{\lv}(x_{j}\lv)=x_{j+1}-x_{j}\ge r>0 $ uniformly for some given $ r>0 $. 
Note that $ x_{j}\lv\in\reg(\Sigma_{p}) $ with $ m(x_{j}\lv)=1 $ for all $ j\in\Z $, since every $ x_{j} $ has multiplicity one. The sequence $ \{x_{j}\lv\}_{j\in\Z} $ is dense in $ \reg(\Sigma_{p}) $ since $ \lv $ is $ \Q $-linearly independent, so $ m(\xv)=1 $ for all $ \xv\in\reg(\Sigma_{p}) $ and  $\mathrm{inf}\{\tau_{\lv}(\xv):\xv\in\reg(\Sigma_{p})\}=\mathrm{inf}\{\tau_{\lv}(x_{j}\lv):j\in\Z\}\ge r>0  $, by continuity of $ \tau_{\lv} $ and $ m $ on $ \reg(\Sigma_{p}) $. Then $ \sing(\Sigma_{p})=\emptyset $, by \Cref{lem: tau}, which means that $ \mult(\xv)\equiv1 $.

(Genericity)
By \Cref{thm: LYd full dimension}, 
For any ${\bf d}\in \Z_{> 0}^n$, the subset $ \A_{\bf d}^{\circ}\subset \A_{\bf d} $ of $p\in \A_{\bf d}(n)$ that satisfy (i), is a semialgebraic open, dense subset of $\A_{\bf d}(n)$. Furthermore, for any nonzero $p\in \A_{\bf d}(n)$, we can chose ${\bf x}\in [0,2\pi)^n$ for which $p(\exp(i{\bf x}))\neq 0$. By \Cref{cor:LY_top}, for any $\lambda>0$, the 
polynomial $(\mathcal{D}_{\lambda,\xv})^{|{\bf d}|}p$ satisfies (i). As seen in \Cref{def: Tlambda}, every application of $\mathcal{D}_{\lambda,\xv}$ 
contributes one to the degree of $\lambda$ and so the result, $(\mathcal{D}_{\lambda,\xv})^{|{\bf d}|}p$ can be expressed 
as a polynomial of degree $ |\bd| $ in $ \lambda $.

For (ii), consider the set $B_{\alpha}$ of polynomials $p\in \A_{\bf d}(n)$ that has a binomial factor of multi-degree $ \alpha\le \bd,\alpha\ne \bd $. We will see that $B_{\alpha}$ is a semialgebraic subset of $ \A_{\bf d}(n)$ of positive codimension. By \Cref{lem: binoimal}, the binomial factor of $p$ has the form  
$(1+a{\bf z}^{\alpha})$ for some $a\in \C^*$ with $|a| = 1$. Therefore $B_{\alpha}=\set{(1+a{\bf z}^{\bf \alpha})q({\bf z})\ :\ |a|=1,\ q\in \A_{\bd-\alpha}}.$
From this and \Cref{thm: LYd full dimension}, we see that $ B_{\alpha} $ is semi-algebraic of dimension 
\[\dim(B_{\alpha}) = 1+\dim(\A_{\bd-\alpha})= 2+\prod_{j=1}^n(d_j-\alpha_j+1).\] 
Since $\alpha\neq 0$, there is some $\alpha_i\geq 1$. We then calculate that
%
\[\prod_{j=1}^n(d_j-\alpha_j+1)\le(d_i-\alpha_i+1)\prod_{j\ne i}(d_j+1)
=\prod_{j=1}^n(d_j+1)-\alpha_i\prod_{j\ne i}(d_j+1)
 <\prod_{j=1}^n(d_j+1)-1,\]
using that $ \alpha_i\prod_{j\ne i}(d_j+1)\ge 2^{n-1} > 1$ since $ n\ge 2 $ and $ d_{j}+1\ge 2 $ for all $ j $. 
This shows that $\dim(B_{\alpha}) <\dim(\A_{\bd})$ for any $0\lneq \alpha \leq {\bf d}$. 

Together these show that the set of polynomials in $\A_{\bf d}(n)$ satisfying (i) and (ii) is a semialgebraic, open dense subset of $\A_{\bf d}(n)$.\end{proof}

\section{Gap distributions}\label{sec: gap distributions}

	The existence of a gap distribution $ \rho_{p,\lv} $ was previously known for specific type of Lee-Yang polynomials, those for which the zeros of $ p(\exp(ix\lv)) $ are the square-root eigenvalues of a quantum graph that has $ n $ edges of lengths $ \lv=(\lv_{1},\ldots,\lv_{n}) $, assuming these lengths are $ \Q $-linearly independent \cite{BarGas_jsp00,BerWin_tams10,CdV_ahp15}. The existence of a gap distribution of $ \mu_{p,\lv} $ for any choice of Lee-Yang $ p $ and positive $ \lv $ is proven in this chapter. In particular, this includes the case of quantum graphs with edge lengths that are not $ \Q $-linearly independent.


Recall that if $ p $ has multi-degree $ \bd $ and it decomposes as $ p=\prod_{j=1}^{N}q_{j}^{c_{j}} $ into distinct irreducible $ q_{j} $'s, then $ p^{\mathrm{red}}=\prod_{j=1}^{N}q_{j} $ is the reduced square-free polynomial and we denote its multi-degree by $ \bd^{\mathrm{red}} $. In particular, $ \bd^{\mathrm{red}}\le\bd $ element-wise, with equality if and only if $ p $ is square free. As seen in \Cref{lem: tau}, if we number the zeros of $ f(x)=p(\exp(ix\lv)) $ increasingly with multiplicity, then the positive gaps are described by $ \tau_{\lv}:\Sigma_{p}\to\R_{+} $,
\begin{equation}\label{eq: xj+1-xj}
	x_{j+1}-x_{j}=\tau_{\lv}(x_{j}\lv)\quad\text{whenever}\ x_{j+1}\ne x_{j},
\end{equation}
as can be seen in \Cref{fig: mod 2p}. To prove \Cref{thm: existance of gaps distribution}, let us define the measure $ \nu_{p,\lv} $.
\begin{Def}\label{defn: nu}
	Let $ p\in\A_{\bd}(n) $, $ \lv\in\R_{+}^{n} $, and let $ (\tau_{\lv})_{*}\mm_{\lv} $ denote the push-forward of $ \mm_{\lv} $ by $ \tau_{\lv} $. Define the measure $ \nu_{p,\lv} $ on $ \R_{\ge0} $ by
	\begin{equation}\label{eq: nu(p,l)}
		\nu_{p,\lv}:= c_{0}\delta_{0}+c_{\tau}(\tau_{\lv})_{*}\mm_{\lv},\quad \text{ with  }\ 
		c_{0}:=\frac{\langle\bd-\bd^{\mathrm{red}},\lv\rangle}{\langle\bd,\lv\rangle},\ c_{\tau}:=\frac{1}{(2\pi)^{n-1}\langle\bd,\lv\rangle}.
	\end{equation}
That is, for any continuous $ f:\R_{\ge0}\to\C $,
\begin{equation}\label{eq: int f}
	\int f d\nu_{p,\lv}:=c_{0}f(0)+c_{\tau}\int_{\Sigma_{p}}f(\tau_{\lv}(\xv))\ d\mm_{\lv}(\xv).
\end{equation}
\end{Def}
\begin{remark}
	The measure $ \nu_{p,\lv} $ is normalized, $ \int d\nu_{p,\lv}=1 $, since $ \int d(\tau_{\lv})_{*}\mm_{\lv}=\mm_{\lv}(\Sigma_{p}) $,
	\[ c_{\tau}=\frac{1}{\int_{\Sigma_{p}}\mult(\xv)d\mm_{\lv}(\xv)}, \text{ and }\ c_{0}=\frac{\int_{\Sigma_{p}}(\mult(\xv)-1)d\mm_{\lv}(\xv)}{\int_{\Sigma_{p}}\mult(\xv)d\mm_{\lv}(\xv)}=1-c_{\tau}\mm_{\lv}(\Sigma_{p}).\]
\end{remark}

\begin{proof}[Proof of \Cref{thm: existance of gaps distribution} and \Cref{thm: p dependence}]
	Fix $ \mu $, an $ \N $-FQ, and let $ n\in\N,p\in\A_{\bd}(n), $ and $ \lv\in\R_{+}^{n} $ with $ \Q $-linearly independent entries, such that $ \mu=\mu_{p,\lv} $, as guaranteed by \cite{AloCohVin}. Consider the decomposition $ p=\prod_{j=1}^{N}q_{j}^{c_{j}} $ into distinct irreducible Lee-Yang polynomials. Let $ (x_{j})_{j\in\Z} $ be the zeros of $ p(\exp(ix\lv)) $, numbered increasingly with multiplicity.	
	
	The proofs of \Cref{thm: existance of gaps distribution} and \Cref{thm: p dependence} interlace according to the following sequence of lemmas, which will be proven afterwards. 
	For each, we take the assumptions listed above. 

	\begin{lem}\label{lem:step1}
	The gap distribution $\rho= \rho_{p,\lv} $ exists and is equal to $ \nu_{p,\lv} $. 
	That is, for any continuous $ f:\R\to\C $, 
		\[\lim_{N\to\infty}\frac{1}{N}\sum_{n=1}^{N}f(x_{j+1}-x_{j})=\int f\nu_{p,\lv}.\]
		Moreover, $ \nu_{p,\lv}=\nu_{q,\lv} $ when $ q(\z):=p(\exp(i\xv_{0})\z) $ for any fixed $ \xv_{0}\in\R^{n} $ (\Cref{thm: p dependence}(1)).
	\end{lem}

	\begin{lem}[\Cref{thm: p dependence}(2),(3)]\label{lem:step2}
The distribution  $ \nu_{p,\lv} $ has an atom at $ \Delta=0 $ if and only if $ p $ is not square free. 
It  has an atom at $ \Delta>0 $ if and only if some (not necessarily distinct) pair of factors, $ q_{i} $ and $ q_{j} $ are related by $q_{j}(\z)= q_{i}(\exp(i\Delta \lv)\z) $ for all $ \z $. Moreover, if this holds and $ q_{i}=q_{j} $ then $ q_{i} $ is binomial. 
	\end{lem}
		
	\begin{lem}\label{lem:step3}
$ \nu_{p,\lv} $ has no singular continuous part. 
	\end{lem}
		
\Cref{lem:step2} and \Cref{lem:step3} then give the following. 
		\begin{cor}[\Cref{thm: existance of gaps distribution}(1)]
The distribution $ \nu_{p,\lv} $ has finitely many atoms and no singular continuous part.   
		\end{cor}

	\begin{lem}[\Cref{thm: existance of gaps distribution}(3)]\label{lem:step4}
 For any $ \Delta=x_{j+1}-x_{j}$ and any open interval $ I $ that contains $ \Delta $, $ \nu_{p,\lv}(I)>0 $.  
	\end{lem}
 Together with \Cref{thm: irreduicible}, this gives the following:
		\begin{cor}[\Cref{thm: existance of gaps distribution}(2)]
		If $ \mu_{p,\lv} $ is periodic then $ \nu_{p,\lv} $ is purely atomic. Conversely, if $ \mu_{p,\lv} $ is not periodic, with support $ \Lambda $, then at least one of the following holds:
		\begin{enumerate}
			\item $ \Lambda $ contains two arithmetic progressions with periods $ \Delta_{1},\Delta_{2} $ such that $ \frac{\Delta_{1}}{\Delta_{2}}\notin\Q $.
			\item $ \dim_{\Q}(\Lambda)=\infty $. 
			\end{enumerate}
Each one of these ensures that there are infinitely many gap values, hence the support of $ \nu_{p,\lv} $ is not finite. In particular $ \nu_{p,\lv} $ must have an absolutely continuous part. 
		\end{cor}
Once proven, these statements complete the proof of  \Cref{thm: existance of gaps distribution} and \Cref{thm: p dependence}.
	\end{proof}

\begin{proof}[Proof of \Cref{lem:step1}]
	Let $ f:\R\to\C $ be continuous, so the composition $ f\circ\tau_{\lv} $ is Riemann integrable, since $ \tau_{\lv} $ is bounded and continuous on an open full measure set $ \reg(\Sigma_{p}) $, see \Cref{lem: tau}. Therefore, the function
	\[h(\xv):=\frac{\mult(\xv)-1}{\mult(\xv)}f(0)+\frac{1}{\mult(\xv)}f(\tau_{\lv}(\xv)),\]
	is bounded and Riemann integrable. By \Cref{thm: ergodicity} we get, 
	\[\lim_{N\to\infty}\frac{1}{N}\sum_{j=1}^{N}h(x_{j}\lv)=  \frac{\langle\bd-\bd^{\mathrm{red}},\lv\rangle}{\langle\bd,\lv\rangle}f(0)+\frac{1}{(2\pi)^{n-1}\langle\bd,\lv\rangle}\int_{\Sigma_{p}}f(\tau_{\lv}(\xv))\ d\mm_{\lv}(\xv)\\
	= \int f d\nu_{p,\lv}.\]
	Whenever $ x_{j-1}<x_{j}=x_{j+1}=\ldots=x_{j+(m-1)}<x_{j+m} $,
	we have $ \tau_{\lv}(x_{i}\lv)=x_{j+m}-x_{j+(m-1)} $ and $ \mult(x_{i}\lv)=m $ for all $ i\in\{j,\ldots,j+m-1\} $, so 
	\[\sum_{i=j}^{j+m-1}h(x_{i}\lv)=(m-1)f(0)+f(x_{j+m}-x_{j+(m-1)})=\sum_{i=j}^{j+m-1}f(x_{i+1}-x_{i}).\]
	Therefore, given any $ N\in\N $ such that $ x_{N}<x_{N+1} $,
	\begin{equation}\label{eq: sum h equal sum f}
		\frac{1}{N}\sum_{j=1}^{N}h(x_{j}\lv)=\frac{1}{N}\sum_{j=1}^{N}f(x_{j+1}-x_{j}).
	\end{equation}
	The left-hand-side of \eqref{eq: sum h equal sum f} converges to $ \int f d\nu_{p,\lv} $ as $ N\to \infty $. The equality in \eqref{eq: sum h equal sum f} holds for infinitely many $ N $ values (those for which $ x_{N}<x_{N+1} $) whose spacing is bounded by the maximum multiplicity $ |\bd| $, so according to Lemma \ref{lem: partial sum},
	\begin{equation*}
		\lim_{N\to\infty}\frac{1}{N}\sum_{j=1}^{N}f(x_{j+1}-x_{j})=\int f d\nu_{p,\lv}.
	\end{equation*}
	Given any fixed $ \xv_{0}\in\R^{n}/2\pi\Z^{n} $, let $ q(\z):=p(\exp(i\xv_{0})\z) $ and let $ (t_{j})_{j\in\Z} $ denote the repeated ordered zeros of $ t\mapsto q(\exp(it\lv))=p(\exp(i(\xv_{0}+t\lv))) $, then according to \Cref{thm: ergodicity}, 
	\[\lim_{N\to\infty}\frac{1}{N}\sum_{j=1}^{N}f(x_{j+1}-x_{j})=\lim_{N\to\infty}\frac{1}{N}\sum_{j=1}^{N}f(t_{j+1}-t_{j}),\]
	namely $ \nu_{p,\lv}=\nu_{q,\lv} $.
\end{proof}	
		
\begin{proof}[Proof of \Cref{lem:step2}]
	By \Cref{defn: nu}, $ \nu_{p,\lv} $ has an atom at $ \Delta=0 $ if and only if the multidegrees of $p$ and $p^{\rm red}$ differ, which occurs if and only if $ p $ is not square free. 

	Suppose that $ \nu_{p,\lv} $ has an atom at $ \Delta>0$. Then, $ (\tau_{\lv})_{*}\mm_{\lv} $ has an atom at $ \Delta $, which means that the level set $ \tau_{\lv}^{-1}(\Delta) $ has positive measure $ \mm_{\lv}(\tau_{\lv}^{-1}(\Delta))>0 $. The set $ A:=\reg(\Sigma_{p})\cap T_{\lv}^{-1}(\reg(\Sigma_{p})) $ is an open subset of $\reg(\Sigma_{p}) $ of full $ \mm_{\lv} $ measure and $ \tau_{\lv} $ is real analytic on $ A $ by \Cref{lem: tau}. Since $ \mm_{\lv} $ is absolutly continuous with respect to the volume measure on $ A $, then $ A\cap\tau_{\lv}^{-1}(\Delta) $ has positive volume, and therefore $ \tau_{\lv} $ is identically $ \Delta $ on some open set $ U\subset A  $. By taking $ U $ sufficiently small, there are two (not necessarily distinct) irreducible factors of $ p $, say $ q_{1} $ and $ q_{2} $, such that $ q_{1}(\exp(i\xv))=0 $ and $ q_{2}(\exp(i(\xv+\Delta\lv)))=0 $ for all $ \xv\in U $. It follows that $ q_{1}(\z)=q_{2}(\exp(i\Delta\lv)\z) $ for all $ \z\in\C^{n} $ by \Cref{lem: dimension and singularity} part (3), since $ q_{1} $ and $ q_{2} $ are irreducible Lee-Yang polynomials. 
	
	Now, suppose that  $ q_{1}=q_{2} $, so $ q_{1}(\z)=q_{1}(\exp(i\Delta\lv)\z) $ for all $ \z\in\C^{n} $. In particular, if $ \Lambda $ is the zero set of $ x\mapsto q_{1}(\exp(ix\lv)) $, then for any $ x\in\Lambda $ we have $ x+\Delta\in\Lambda $, and as a result, $ x+j\Delta\in\Lambda $ for any $ j\in\N $. Since $ q_{1} $ is irreducible Lee-Yang polynomial and $ \lv\in\R_{+}^{n} $ has $ \Q $-linearly independent entries, then $ q_{1} $ must be binomial, by \Cref{thm: irreduicible}. \end{proof}

\begin{proof}[Proof of \Cref{lem:step3}]
It follows from \Cref{lem:step2} that $ \nu_{p,\lv} $ has finitely many atoms, say $ (t_{i})_{i=1}^{N} $, so that $ \nu_{p,\lv}=\sum_{j=1}^{n}c_{j}\delta_{t_{j}}+\rho_{ac} $, with $ \rho_{ac} $ being a continuous measure (no atoms). We now show that $ \rho_{ac} $ is absolutely continuous with respect to Lebesgue measure. Let $ A:=\reg(\Sigma_{p})\cap T_{\lv}^{-1}(\reg(\Sigma_{p})) $.
We use $A_{nc}$ to denote the union of the connected components of $A$ on which $\tau_{\lv}$ is not constant. 
Then $ \rho_{ac} $ is ($ c_{\tau} $ times) the push-forward of $ \mm_{\lv} $ by the restriction of $ \tau_{\lv} $ to  $ A_{nc}$, using that $ A $ has full measure. It is left to show that for any set $ E\subset\R $ of Lebesgue measure zero, the set $ A_{nc}\cap\tau_{\lv}^{-1}(E) $ has zero $ \mm_{\lv} $ measure, or equivalently, due to $\text{\Cref{lem: BG measure}}$, zero volume in $ \reg(\Sigma_{p}) $.

Since $ \tau_{\lv} $ is real analytic on $ A $ and is not constant on any open set in $ A_{nc}$, then the set $ \Omega=\set{\xv\in A_{nc}\ :\ \nabla\tau_{\lv}(\xv)\ne 0  } $ is open in $A_{nc}$ and its complement in $ A_{nc} $ has zero volume. By the definition of $ \Omega $, $ \tau_{\lv} $ has no critical points in $ \Omega $, which means that for any compact connected $ K\subset\Omega $, the image of $ \tau_{\lv} $ over $ K $ is an interval $ [a,b] $ and the level sets $ K\cap\tau_{\lv}^{-1}(t) $ for $ t\in[a,b] $ are homotopic to one another. In particular, if we let $ \mathrm{area}(K,t)=\sigma_{n-2}(K\cap\tau_{\lv}^{-1}(t)) $ denote the $(n-2)$-dimensional volume of the level set, induced by the volume form $ d\sigma $ on $ \reg(\Sigma_{p}) $, then $ t\mapsto \mathrm{area}(K,t) $ is continuous in $ t\in[a,b] $, and so it is bounded by some constant. Let $ C $ be the maximum of $ \mathrm{area}(K,t) $ for $ t\in[a,b] $, and $ |\nabla \tau_{\lv}(\xv)|^{-1} $ for $ \xv\in K $. Then, 
\[\int_{K\cap\tau_{\lv}^{-1}(E)}d\sigma\le C\int_{K\cap\tau_{\lv}^{-1}(E)}\|\nabla\tau_{\lv}(\xv)\|d\sigma(\xv)=M\int_{t\in E}\mathrm{area}(K,t)dt\le C^2\int_{t\in E}dt=0,\]
using the co-area formula (or Disintegration Theorem) in the middle equality. It follows that $ \mm_{\lv}(K\cap\tau_{\lv}^{-1}(E))=0 $ for any compact connected $ K\subset\Omega $, hence $ \rho_{ac}(E)\propto\mm_{\lv}(\Omega\cap\tau_{\lv}^{-1}(E))=0 $. As it holds for any $ E $ of zero Lebesgue measure, $ \rho_{ac} $ is absolutely continuous.
\end{proof}
 
\begin{proof}[Proof of \Cref{lem:step4}]
Let $ \Delta=x_{j+1}-x_{j} $ for some arbitrary fixed choice of $ j $, let $ I\subset \R $ be any open interval with $ \Delta\in I $, and consider the open set $U:=\set{\xv\in\reg(\Sigma_{p}):\tau_{\lv}(\xv)\in I} $. It is enough to show that $ U\ne\emptyset $ to conclude that $ \mm_{\lv}(U)>0 $, by \Cref{lem: BG measure}, and so
	\[\nu_{p,\lv}(I)\ge c_{\tau}\mm_{\lv}(U)>0. \]
	Consider two cases, according to whether $ \Delta>0 $ or $ \Delta=0 $.
	
	($ \Delta>0 $) Suppose $ x_{j+1}>x_{j} $ and let $ \xv=x_{j}\lv\mod{2\pi} $, so $ \tau_{\lv}(\xv)=\Delta $. If $ \xv\in\reg(\Sigma_{p}) $, then $ \xv\in U $. Otherwise, if $ \xv\in\sing(\Sigma_{p}) $, then $ \Delta\in\overline{\{\tau_{\lv}(\xv):\xv\in\reg(\Sigma_{p})\}} $, by \Cref{lem: tau}, which means that $ U\ne\emptyset $.

	($ \Delta=0 $) Suppose $ x_{j+1}=x_{j} $. If $ \sing(\Sigma_{p})\ne\emptyset $ then $ \Delta=\inf_{\xv\in\reg(\Sigma_{p})}\tau_{\lv}(\xv)$ by \Cref{lem: tau}, and so $ U\ne\emptyset $. Otherwise, if $ \sing(\Sigma_{p})=\emptyset $, having $ x_{j+1}=x_{j} $ means that $ p $ has a square factor, and so
	\[\nu_{p,\lv}([\Delta-\epsilon,\Delta+\epsilon])\ge\nu_{p,\lv}(\{0\})=c_{0}>0. \]
	\end{proof}

\begin{proof}[Proof of \Cref{cor: atoms}] 
(2) and (3) were already discussed in \Cref{thm: irreduicible} and \Cref{lem:step2}, respectively. 
 For (1), if $ p $ is irreducible and not binomial, then its gap distribution cannot have any atoms by \Cref{thm: p dependence} part (3), and so it is absolutely continuous by \Cref{thm: existance of gaps distribution} part (1). Part (4) is a counting argument. Suppose that $ p $ has $ N+M $ distinct irreducible factors, $ M $ of which are binomial. There can be three types of atoms according to \Cref{lem:step2},  
 	\begin{enumerate}
		\item[(a)] an atom at zero.
	\item[(b)] an atom at positive $ \Delta>0 $ coming from a pair of distinct non-binomial factors related by $ q_{i}(\z)=q_{j}(\exp(i\Delta\lv)\z) $. 
	\item[(c)] an atom at a positive $ \Delta>0 $, coming from a pair of (not necessarily distinct) binomial factors related by $ q_{i}(\z)=q_{j}(\exp(i\Delta\lv)\z) $. 
	\end{enumerate} 
Notice that if $ q_{i} $ is binomial and $ q_{j} $ is non-binomial, then they cannot satisfy a relation of the form $ q_{i}(\z)=q_{j}(\exp(i\Delta\lv)\z) $, as such a relation means that the torus zero set $ \Sigma_{q_{i}} $, which is a torus, is a translation of the torus zero set $ \Sigma_{q_{j}} $, which is not a torus. It is left to bound the number of atoms of each type. There can be at most one (a) atom. 

For atoms of type (b), notice that a pair of non-binoimal factors cannot satisfy the relation $ q_{i}(\z)=q_{j}(\exp(i\Delta\lv)\z) $ for two different values of $ \Delta>0 $, say $ \Delta_{1}\ne\Delta_{2} $. Otherwise, we get $ q_{j}(\z)=q_{j}(\exp(i(\Delta_{1}-\Delta_{2})\lv)\z) $ in contradiction to $ q_{j} $ being non-binomial. Therefore, there are at most $ {N \choose 2}$ atoms of type (b), one for each possible pair.

To bound the number of type (c) atoms, consider a pair of (not necessarily distinct) binomial factors related by $ q_{i}(\z)=q_{j}(\exp(i\Delta\lv)\z) $. In particular, $q_{i}$ and $q_{j} $ share the same multi-degree, say $ \alpha $. According to \Cref{lem: binoimal}, the zero sets of $ f_{i}(x)=q_{i}(\exp(ix\lv)) $ and $ f_{j}(x)=q_{j}(\exp(ix\lv)) $ are arithmetic progressions of the same step size, say $ \Lambda_{i}=\{a+\frac{2\pi}{\langle \alpha,\lv\rangle}k\}_{k\in\Z} $ and $ \Lambda_{j}=\{a+\Delta+\frac{2\pi}{\langle \alpha,\lv\rangle}k\}_{k\in\Z} $ for some $ a\in\R $. Suppose that $ p $ has exactly $ M_{\alpha} $ binomial factors with multi-degree $ \alpha $, and let $ \Lambda_{\alpha} $ denote the union of their arithmetic progressions defined above. Then $ \Lambda_{\alpha} $ is $ \frac{2\pi}{\langle \alpha,\lv\rangle} $ periodic with $ M_{\alpha} $ points in a period, and therefore at most $ M_{\alpha} $ gap values between consecutive points. By partitioning the $ M $ binomial factors according to their multi-degrees we see that there are at most $ M $ atoms of type (c). To conclude, there are at most $ {N \choose 2}+M+1$ atoms
\end{proof}
\begin{proof}[Proof of \Cref{thm: l dependenece }]
By \Cref{lem:step1}, if $ p\in\A_{\bd}(n) $ and $ \lv\in\R_{+}^{n} $ has $ \Q $-linearly independent entries, then $\rho_{p,\lv}=\nu_{p,\lv}  $. It is left to show that $ \nu_{p,\lv} $ is weakly continuous in $ \lv $, namely, that for any fixed continuous $ f:\R\to\C $, the following integral is continuous in $\lv\in \R_{+}^{n} $,  
	\[\int f d\nu_{p,\lv}:=c_{0}f(0)+c_{\tau}\int_{\Sigma_{p}}f(\tau_{\lv}(\xv))\ d\mm_{\lv}(\xv).\]
	The weights $ c_{0} $ and $ c_{\tau} $, given in \Cref{defn: nu}, are continuous in $\lv\in \R_{+}^{n} $, and the remaining integral can be written as 
	\[\int_{\Sigma_{p}}f(\tau_{\lv}(\xv))\ d\mm_{\lv}(\xv)=\int_{\reg(\Sigma_{p})}f(\tau_{\lv}(\xv))\ d\mm_{\lv}(\xv)=\sum_{j=1}^{n}\lv_{j}\int_{\reg(\Sigma_{p})}f(\tau_{\lv}(\xv))\ d\mm_{e_{j}}(\xv),\]
	using that $ \mm_{\lv}(\sing(\Sigma_{p}))=0 $ in the first equality, and the linearity of $ \mm_{\lv}$ in $ 
	\lv $, by \Cref{lem: BG measure}, in the second equality. The integral $ \int_{\reg(\Sigma_{p})}f(\tau_{\lv}(\xv))\ d\mm_{e_{j}}(\xv) $ is continuous in $ \lv $ because $(\xv,\lv)\mapsto f(\tau_{\lv}(\xv)) $ is continuous over $ \reg(\Sigma_{p})\times\R_{+}^{n} $, by continuity of $ f $ and Lemma \ref{lem: tau}. 
\end{proof}
Let us now prove Theorem \ref{thm: the l to 1 case}.
\begin{proof}[Proof of Theorem \ref{thm: the l to 1 case}]
	Fix $ p\in \A_{\bd}(n) $, and for any $ \xv\in\R^{n} $ let $ p_{\xv}\in\A_{|\bd|}(1) $ be the univariate polynomial $ p_{\xv}(s)=p(se^{ix_{1}},se^{ix_{2}},\ldots,se^{ix_{n}}) $ whose degree is $ |\bd| $ and its roots lie on the unit circle. Let $ \theta_{j}:\R^{n}\to\R $, for $ j=1,2\ldots,|\bd| $, be the continuous phase functions given in \Cref{prop: phase functions}, so that $ (e^{i\theta_{j}(\xv)})_{j=1}^{|\bd|} $ are the roots of $ p_{\xv} $ numbered (counter-clockwise) increasingly including multiplicity, and let $ \theta_{|\bd|+1}=\theta_{1}+2\pi $. 
	We need to prove that for any continuous $ f:\R\to\C $,
	\[\int f d\nu_{p,\textbf{1}}=\frac{1}{(2\pi)^{n}}\int_{\xv\in[0,2\pi]^{n}}\left[\frac{1}{|\bd|}\sum_{j=1}^{|\bd|}f(\theta_{j+1}(\xv)-\theta_{j}(\xv))\right]d\xv,\]
	where $ \textbf{1}=(1,1,\ldots,1) $. Fix a continuous $ f:\R\to\C $ and define 
	\[h(\xv):=\frac{\mult(\xv)-1}{\mult(\xv)}f(0)+\frac{1}{\mult(\xv)}f(\tau_{\lv}(\xv)).\]
	Consider the layers $ \Sigma_{p,j} $ and their parameterizations $ \varphi_{j} $ as defined in \Cref{prop: Layer structure}, so that the multiplicity $ m(\xv) $ counts the number of layers containing $ \xv $, so that 
		\[\int_{\Sigma_{p}}\mult(\xv)h(\xv)d\mm_{\textbf{1}}(\xv)=\sum_{j=1}^{|\bd|}\int_{\Sigma_{p,j}}\mult(\xv)h(\xv)d\mm_{\textbf{1}}(\xv)=\sum_{j=1}^{|\bd|}\int_{(0,2\pi]^{n-1}}h(\varphi_{j}(\yv))d\yv,\]
		using \eqref{eq: dm3} from \Cref{lem: BG measure} in the last equality.   
	As in the proof of \Cref{lem:step1}, this gives
	\[\int f d\nu_{p,\textbf{1}}=  \frac{1}{(2\pi)^{n-1}|\bd|}\int_{\Sigma_{p}}\mult(\xv)h(\xv)\ d\mm_{\textbf{1}}(\xv)= \frac{1}{(2\pi)^{n-1}|\bd|}\sum_{j=1}^{|\bd|}\int_{(0,2\pi]^{n-1}}h(\varphi_{j}(\yv))d\yv.\]
	As seen in the proof of \Cref{lem: tau}, if $ \theta_{j+1}(\varphi_{j}(\yv))>\theta_{j}(\varphi_{j}(\yv)) $ then 
	$ \tau_{\textbf{1}}(\varphi_{j}(\yv)) $ is equal to the unique $ t\in\R $ such that $ \theta_{j+1}(\varphi_{j}(\yv)+t\textbf{1})=\theta_{j}(\varphi_{j}(\yv)) $. In such case, using \Cref{prop: phase functions} part (1) and the definition of $ \varphi_{j} $, we get $ \tau_{\textbf{1}}(\varphi_{j}(\yv))=\theta_{j+1}(\yv,0)-\theta_{j}(\yv,0) $. The number of $ j $'s for which $ \theta_{j+1}(\yv,0)=\theta_{j}(\yv,0) $ is exactly $ \sum_{j=1}^{|\bd|}\left(\mult(\varphi_{j}(\yv))-1\right)$, so 
	\[ \sum_{j=1}^{|\bd|}h(\tau_{\lv}(\varphi_{j}(\yv)))=\sum_{j=1}^{|\bd|}f\left(\theta_{j+1}(\yv,0)-\theta_{j}(\yv,0)\right),\]
	for every $ \yv $, and integrating gives, 
	\begin{equation}\label{eq: partial integral}
		\int f d\nu_{p,\textbf{1}}=\frac{1}{(2\pi)^{n-1}|\bd|}\int_{(0,2\pi]^{n-1}}\left[\sum_{j=1}^{|\bd|} f\left(\theta_{j+1}(\yv,0)-\theta_{j}(\yv,0)\right)dy\right]. 
	\end{equation}
	 Let $ g(\xv):=\sum_{j=1}^{|\bd|}f(\theta_{j+1}(\xv)-\theta_{j}(\xv)) $, and notice that $ g $ is continuous, satisfies $ g((\yv,0)+t\textbf{1})=g(\yv,0) $ by \Cref{prop: phase functions} part (1), and is $ 2\pi $ periodic by \Cref{prop: phase functions} part (4), so  
	\[\int_{\yv\in(0,2\pi]^{n-1}}g(\yv,0)d\yv=\frac{1}{2\pi}\int_{t=0}^{2\pi}\int_{\yv\in(0,2\pi]^{n-1}}g((\yv,0)+t\textbf{1})d\yv dt=\int_{\xv\in(0,2\pi]^{n}}g(\xv)d\xv.\]
	The needed result follows 
	\[\int f d\nu_{p,\textbf{1}}=\frac{1}{(2\pi)^{n-1}|\bd|}\int_{\xv\in(0,2\pi]^{n}}\left[\sum_{j=1}^{|\bd|} f\left(\theta_{j+1}(\xv)-\theta_{j}(\xv)\right)d\xv\right]\]   
\end{proof}




\appendix
\section{}\label{sec: appendix A}
The next lemma is being used throughout the paper.   
\begin{lem}\label{lem: partial sum}
	Let $ (a_{n})_{n\in\N} $ be a bounded sequence $ |a_{n}|<M $ and let $ (s_{n})_{n\in\N}  $ be the sequence of partial averages, $ s_{N}:=\frac{1}{N}\sum_{n=1}^{N}a_{n} $. Suppose that there exists a converging subsequence $ \lim_{j\to\infty}s_{n_{j}}= L $, with a uniform spacing bound $ n_{j+1}-n_{j}<M' $. Then, $ \lim_{n\to\infty}s_{n}= L $.
\end{lem} 
\begin{proof}
	Given any $ n_{j}\le n'\le n_{j+1} $, the uniform spacing bound gives $ \frac{n_{j}}{n'}\to 0 $ as $j\to\infty$, and we have 
	\[|s_{n'}-\frac{n_{j}}{n'}s_{n_{j}}|=\frac{|a_{n_{j}+1}+a_{n_{j}+2}+\ldots+a_{n'}|}{n'} \le\frac{M'M}{n_{j}}\to 0 \ \text{ as } j\to\infty.
	\]
\end{proof}

\bibliographystyle{plain}
\bibliography{ref,QGbib}

\end{document}